\documentclass[twocolumn,a4paper]{quantumarticle}
\pdfoutput=1

\usepackage[utf8]{inputenc}
\usepackage[T1]{fontenc}
\usepackage{amsmath}
\usepackage{amssymb}
\usepackage[numbers,sort&compress]{natbib}
\usepackage{microtype}
\usepackage{tikz}
\usetikzlibrary{calc}

\usepackage{hyperref}

\usepackage{amsthm}

\newif\ifdraftmode
\draftmodefalse         

\ifdraftmode
  \newcommand{\todo}[1]{\textcolor{red}{\textsf{\footnotesize [TODO: #1]}}}
  \newcommand{\tofill}[1]{\textcolor{blue!70!black}{\textsf{\footnotesize $\langle$#1$\rangle$}}}
  \newcommand{\note}[1]{\textcolor{gray}{\textsf{\footnotesize (#1)}}}
\else
  \newcommand{\todo}[1]{}
  \newcommand{\tofill}[1]{}
  \newcommand{\note}[1]{}
\fi

\theoremstyle{plain}
\newtheorem{theorem}{Theorem}
\newtheorem{proposition}[theorem]{Proposition}
\newtheorem{lemma}[theorem]{Lemma}
\newtheorem{corollary}[theorem]{Corollary}
\theoremstyle{definition}
\newtheorem{definition}[theorem]{Definition}

\theoremstyle{remark}

\newcommand{\C}{\mathbb{C}}
\newcommand{\Z}{\mathbb{Z}}
\newcommand{\Herm}{\mathrm{Herm}}
\newcommand{\PSD}{\mathrm{PSD}}
\newcommand{\SWAP}{\mathbb{S}}
\newcommand{\Tr}{\operatorname{Tr}}
\newcommand{\id}{\mathrm{id}}
\newcommand{\Id}{\mathbb{I}}

\newcommand{\Th}{\mathsf{T}}                    
\newcommand{\sTh}{\mathsf{sT}}                  
\newcommand{\QT}{\mathsf{QT}}
\newcommand{\RQT}{\mathsf{RQT}}
\newcommand{\CPT}{\mathsf{CPT}}

\newcommand{\stateset}{\Omega}                  
\newcommand{\effset}{\mathcal{E}}               
\newcommand{\opset}{\mathcal{O}}                
\newcommand{\effmax}{\effset_{\max}}

\newcommand{\PT}[1]{\Gamma_{#1}}                
\newcommand{\ptr}[2]{#1^{\mathsf{T}_{#2}}}      
\newcommand{\cpptp}{\text{C-PPT-P}}

\newcommand{\xieff}{\xi}                        
\newcommand{\corr}{\mathrm{corr}}

\newcommand{\ket}[1]{\lvert #1 \rangle}
\newcommand{\bra}[1]{\langle #1 \rvert}
\newcommand{\ketbra}[2]{\lvert #1 \rangle\!\langle #2 \rvert}
\newcommand{\proj}[1]{\ketbra{#1}{#1}}

\begin{document}

\title{Restricting the effects hides a nonphysical symmetry
       from every causal structure}

\author{Chon-Fai Kam}
\email{dubussygauss@gmail.com}
\affiliation{Dipartimento di Fisica e Chimica ``Emilio Segr\`e'',
             Universit\`a degli Studi di Palermo, Via Archirafi 36,
             I-90123 Palermo, Italy}
\affiliation{DSIMB, Inserm, BIGR U1134, Universit\'e Paris Cit\'e and
             Universit\'e de La R\'eunion, 75015 Paris, France}

\date{7 September 2026}   

\begin{abstract}
Real-amplitude quantum theory is the subtheory of quantum theory invariant
under complex conjugation, and experiments in networks of independent sources
have measured correlations above the real bound. We construct a theory on which
the same conjugation fails complete positivity and still leaves every
probability in the unit interval, and whose correlations are exactly those of
its own conjugation-invariant subtheory in every causal structure, the
bilocality scenario included. Its states are all the density matrices, its
effects are the operators every partial transpose of which is again a quantum
effect, and its operations are the completely PPT-preserving maps. The
symmetrized subtheory simulates it once each source carries a reference frame
rather than each system. One map therefore receives three different verdicts in
three theories, so the symmetry alone marks no boundary at all. Quantum theory
admits every effect its states permit, and no theory with that property can
hide a symmetry this way, so what sustains the separation measured in the
network experiments is the absence of a restriction rather than any feature of
conjugation. What does have a boundary is the class of theories whose
correlations coincide with those of their symmetrized subtheory. We show that
sectorial closure, meaning invariance of the effects and the operations under
the symmetry acting independently on each source, suffices for the absence of a
gap under any finite group, and that it cannot be weakened on the effects.
Fixing unrestricted states and conjugation makes the effects of this theory the
largest the symmetry admits and its operations the largest sectorially closed
ones, so it is not one construction among several. Locating where sectorial
closure fails, for a candidate effect set built from a fixed bound entangled
state, is a finite computation on a single ray of effects.
\end{abstract}

\maketitle

\section{Introduction}
\label{sec:intro}

Quantum theory can be written over the real numbers. Stueckelberg shows that a
real Hilbert space of twice the dimension, carrying an operator that commutes
with every observable, reproduces the predictions of the complex theory for any
single system~\cite{Stueckelberg1960}. That the single-system
axioms cannot by themselves select the complex field was already visible in the
classification of formally real Jordan algebras~\cite{JvNW1934}. The construction extends
to many systems as long as all of them share one reference frame, which
McKague, Mosca and Gisin make precise~\cite{McKague2009}, and which can be
implemented with a single universal shared qubit~\cite{ABW2013}. What defeats
it is a network~\cite{Tavakoli2022networks}. Renou \emph{et al.} show that in the bilocality
scenario~\cite{Branciard2010}, where two independent sources feed a common
party, real-amplitude quantum theory and quantum theory obey different bounds,
and that the difference is large enough to measure~\cite{Renou2021}. The separation can be made arbitrarily large~\cite{Sarkar2025gap}. Three
experiments have since reported violations of the real
bound~\cite{LiReal2022,ChenReal2022,WuReal2022}.

Ying \emph{et al.} propose a way of asking the question that does not
turn on where one draws the boundary of real quantum
theory~\cite{Ying2025foil}. Real-amplitude quantum theory is the subtheory of
quantum theory invariant under complex conjugation, and conjugation is a
symmetry of order two. One may then ask, of any theory and any symmetry,
whether passing to the invariant subtheory costs correlations in some causal
structure. Their main positive result is that it never does when the symmetry
is physical, meaning that each of its transformations is already an operation
of the theory. When the symmetry is not physical they distinguish two cases.
Adjoining its transformations may leave every probability in the unit interval,
in which case the symmetry is weakly nonphysical, or it may not, in which case
it is strongly nonphysical. Conjugation is strongly nonphysical in quantum
theory, and the separation found by Renou \emph{et al.} is the corresponding
instance. Whether a strongly nonphysical symmetry must always produce a
separation somewhere for a nonclassical theory is their central conjecture, and
it is open.

The middle case is the one that has stayed empty. Ying \emph{et al.} give a
single weakly nonphysical example, a fragment of quantum theory whose
operations are restricted by hand, and the simulation showing that it has no
gap is patched to fit that restriction. The scarcity has a reason. In
quantum theory a transformation that is logically possible is thereby
physical, and the reason, as those authors observe, is the no-restriction
hypothesis. Anything that could be added to quantum theory without producing a
number outside the unit interval is in it already. A weakly nonphysical
symmetry therefore cannot be exhibited inside quantum theory at all, and
finding one means giving up the no-restriction hypothesis.

This paper takes that step in the simplest way available, by restricting the
effects and leaving the states alone. What emerges is that the robustness of
the real-amplitude separation rests on the no-restriction hypothesis rather
than on any property of conjugation, and that a theory lacking it can carry the
same symmetry undetectably. Section~\ref{sec:theory} constructs a
theory $\Th$ whose states are all the density matrices, whose effects are the
operators every partial transpose of which is again a quantum effect, and
whose operations are the completely PPT-preserving maps of Rains and of Ishizaka and
Plenio~\cite{Rains1999,Rains2001,IshizakaPlenio2005}. Complex conjugation is a
symmetry of $\Th$. Section~\ref{sec:conj} shows that it fails complete
positivity, since conjugating one half of a maximally entangled state does not
give a state, and that it is nonetheless only weakly nonphysical. The two
properties come apart because $\Th$ restricts its effects, which leaves room
between the states of the theory and the operators that its effects cannot
tell apart from states.

Restricting the effects rather than the states is forced, and the argument is
a statement about cones. The middle case needs a theory whose symmetry carries a state out of the state cone
while remaining invisible to every effect, so it needs room between the states
and the operators that the effects cannot tell apart from states. A restriction
on states cannot open that room, because the bidual of the effect set returns
the state cone one started from. A restriction on effects can, and it is also
the only kind that leaves an effect set closed under a map that does not
preserve positivity, which is what the simulation needs. The same cone
computation therefore locates the room and names the restriction that makes it.

That choice also fixes the theory, which meets the objection that a
construction of this kind has been chosen to make its own argument work.
Fix unrestricted states and take conjugation as the symmetry, and the effect
set of $\Th$ is the largest one that admits the symmetry at all, in the precise
sense of Corollary~\ref{cor:emaxclosed}, while Corollary~\ref{cor:opsmaximal}
does the same for its operations among the sectorially closed classes. Theorem~\ref{thm:scprime} closes the other
side, since any effect set the encoding of Sec.~\ref{sec:enc} can reach is
contained in that one. So $\Th$ is the maximal admissible theory rather than one among several, and
Secs.~\ref{sec:theory} and~\ref{sec:enc} reach it by different routes.

Section~\ref{sec:enc} builds a simulation of $\Th$ inside its symmetrized
subtheory, and the construction turns on one choice. Earlier work attaches a
reference frame to each system, and the cross terms that result involve a
partial transpose, which does not preserve positivity. We attach one frame to
each source instead. A source is the object that a causal structure permits to emit
correlated systems, so the choice is allowed, and the
transpositions that then survive act on whole sources, which does preserve
positivity. Section~\ref{sec:main} concludes that $\Th$ and its symmetrized
subtheory have the same correlations in every causal structure. The bilocality
scenario is among them, so a structure that separates real-amplitude from
complex quantum theory fails to separate $\sTh$ from $\Th$, although it
satisfies every condition known to be necessary for a separation.

Section~\ref{sec:sectorial} abstracts the construction. What it requires is
that the effect and operation sets be invariant under the symmetry acting
independently on each source, a property we call sectorial closure, and that
property is sufficient for the absence of a gap under any finite group. A weakly nonphysical symmetry always has its effect set inside the largest
sectorially closed one. That containment runs the argument in reverse, given an exact flag and an
operation class meeting the hypotheses of Theorem~\ref{thm:sectorial}, so every
weakly nonphysical symmetry that supplies those comes with a gap-free companion
theory. The containment
alone does not suffice, a point we return to in
Sec.~\ref{sec:sectorial:not}. The theory $\Th$ is that companion for
conjugation, which is why its effect set could not have been chosen otherwise.

The nearest neighbour to this work is that of Erba and Perinotti, who modify
the rule by which quantum systems compose and obtain a family of theories
agreeing with quantum theory on Bell-type correlations even when the causal
structure is held fixed~\cite{ErbaPerinotti2024}. Here the composition rule is
the ordinary tensor product and only the effects and the operations are
restricted, and the comparison is between a theory and its own symmetrized
subtheory rather than between a new theory and quantum theory. Their
contribution is an example. Ours is a criterion, together with the theory that
saturates it.

\section{Preliminaries}
\label{sec:prelim}

This section fixes terminology and notation. Nothing in it is new, and
everything is stated only in the generality the rest of the paper needs. The theories of Secs.~\ref{sec:theory} to~\ref{sec:main} all have quantum
systems composed by the ordinary tensor product, so that we may work inside the
Hermitian operators rather than in an abstract ordered vector space.
Section~\ref{sec:sectorial} and the discussion return to the general
setting.

\subsection{Theories and restrictions}
\label{sec:prelim:gpt}

We follow the standard formulation of generalized probabilistic
theories~\cite{Hardy2001,Barrett2007,CDP2011}, for which several reviews are
available~\cite{JanottaHinrichsen2014,Plavala2023}.

A generalized probabilistic theory assigns to each system a convex set of
states $\stateset$, a convex set of effects $\effset$ and a set of operations
$\opset$, together with a rule for composing systems. Three conditions are
imposed. Every state pairs with every effect to give a number in $[0,1]$. The
effect set contains $0$ and the unit effect and is closed under
$E \mapsto \Id - E$. And states and effects are tomographic for one
another. For the theories considered here $\stateset$ is a set of density
matrices, $\effset$ a set of operators in $[0,\Id]$, $\opset$ a set of channels
and instruments, and the pairing is the trace.

Given a state set, write
\begin{equation}
  \label{eq:maximaleffects}
  \stateset^{\vee} \;=\; \bigl\{ E \;:\; 0 \le \langle E,\omega\rangle \le 1
  \ \ \forall \omega \in \stateset \bigr\}
\end{equation}
for the largest effect set consistent with it. A theory satisfies the
\emph{no-restriction hypothesis} for effects if $\effset = \stateset^{\vee}$,
and for states if $\stateset$ is likewise the largest state set consistent
with $\effset$~\cite{Chiribella2010,Selby2023accessible}. Quantum theory
satisfies both. A theory that does not is sometimes called a fragment, and the
gap between what a theory admits and what its structure would permit is the
subject of Sec.~\ref{sec:conj:dual}.

\subsection{Symmetrized worlds}
\label{sec:prelim:sym}

A \emph{symmetry} of a theory $T$ is a reversible transformation of its states
and effects that preserves the pairing between them and maps $\opset$ to itself
under conjugation. Given a finite group $G$ acting by symmetries, the
\emph{symmetrized world} $sT$ is the subtheory whose states, effects and
operations are the $G$-invariant ones. Within quantum theory, symmetrization with respect to a symmetry implemented by
unitaries is called \emph{twirling} and with respect to a nonunitary one
\emph{swirling}~\cite{Ying2025foil}. For a general theory the two names track
the physicality of the symmetry rather than its unitarity, and the two distinctions do not coincide. Ying \emph{et al.} exhibit a theory in which every symmetry is nonphysical,
unitary ones included, and another in which every symmetry is physical, the
antiunitary conjugation included. Twirling is the
operation that implements a superselection rule~\cite{WWW1952}, and the
subject has a long history in the theory of quantum reference
frames~\cite{BRS2007}.

The example that motivates the subject is $G = \Z_{2}$ acting on quantum theory
by complex conjugation. The invariant subtheory is real-amplitude quantum
theory, and since conjugation is antiunitary this is a swirled world.

\subsection{Causal structures and gaps}
\label{sec:prelim:gap}

A causal structure specifies which systems are emitted by which sources and
which are held by which parties, together with the classical settings and
outcomes of each party. For a theory $T$ we write $\corr_{T}$ for the set of
distributions over the outcomes, conditioned on the settings, realisable in
that causal structure using states, operations and effects of $T$. We write $\CPT$ for classical probability theory, which is a subtheory of
every theory considered here, $\QT$ for quantum theory and $\RQT$ for
real-amplitude quantum theory.

Because $sT$ is a subtheory of $T$ one always has
$\corr_{sT} \subseteq \corr_{T}$, and the question is whether the inclusion can
be strict. When it is, the causal structure is said to admit a
\emph{symmetrized--nonsymmetrized causal compatibility gap} between $sT$ and
$T$~\cite{Ying2025foil}. Renou \emph{et al.} show that the bilocality scenario of Branciard, Gisin
and Pironio~\cite{Branciard2010,Tavakoli2022networks} admits one between
real-amplitude and complex quantum theory~\cite{Renou2021}, and that no gap arises in the ordinary Bell scenario,
where real-amplitude quantum theory already saturates the Tsirelson bound.

\paragraph*{The status of the falsification claim.}
Several authors argue that the falsification does not go through, by
replacing the tensor product with another composition rule or by weakening the
assumption that
independent sources are described by a product
state~\cite{BarriosHita2026,HoffreumonWoods2026}, and the assumption of
independence has been examined and weakened in its own
right~\cite{Yao2024causal,Weilenmann2025partial}. The debate is live, and the compatibility of those proposals with fermionic information theory has been
questioned in turn~\cite{FermionComment2026}. There is also an older point
that cuts across the dispute. Every operational theory admits a representation
over a real vector space, so the question can never be whether real numbers
suffice in the abstract~\cite{HardyWootters2012}. What the real bound
constrains is a real theory that keeps the other axioms of the textbook
formulation, in particular that probabilities are bilinear in the vectors,
which a generic real representation does not. We take no position on the
falsification debate, and nothing below depends on how it is settled.

Absence of a gap is established by exhibiting a simulation. That means a way
of replacing each source, operation and measurement of a $T$-protocol by one of
$sT$, respecting the causal structure and reproducing the statistics. The
simulations in the literature, and the one constructed in
Sec.~\ref{sec:enc}, all proceed by adjoining a reference-frame register to
each wire, in the manner familiar from the theory of quantum reference
frames~\cite{BRS2007}.

\subsection{Physical, weakly and strongly nonphysical symmetries}
\label{sec:prelim:trichotomy}

A symmetry of $T$ is \emph{physical} if each of its transformations is an
operation of $T$. If it is not, two cases are distinguished. It is
\emph{weakly nonphysical} if adjoining the transformations to $T$ produces no
logical inconsistency, in the sense that no closed circuit returns a number
outside $[0,1]$ where a probability is expected, and \emph{strongly
nonphysical} if some circuit does~\cite{Ying2025foil}. The framework in which
these are defined assumes tomographic locality, which is why
Sec.~\ref{sec:theory:tomo} verifies that property for the theory constructed
below.

Two results of Ref.~\cite{Ying2025foil} frame what follows. A physical symmetry
never produces a gap, in any causal structure, and the proof is by explicit
simulation with independent local reference frames. Ying \emph{et al.} conjecture that a strongly
nonphysical one always produces a gap somewhere for a nonclassical theory, and
the conjecture is open.
The weakly nonphysical case is settled by neither, and rests on a single
example, a prepare-measure fragment of quantum theory for which the simulation
is patched by hand.

That raises the question of which pairs the classification says anything about,
and the answer is not all of them. If a theory has a factor on which the
symmetry acts trivially, the classification can be made to report anything at
all, since the correlations may already be saturated by the inert factor while
the symmetry misbehaves elsewhere. The following lemma says that the family
studied here is free of that degeneracy, and it is what makes the question
below well posed rather than a remark about our particular construction.

\begin{lemma}[Uniformity]
\label{lem:uniformity}
Let $T$ have quantum systems composed by the tensor product with all density
matrices as states, and let $\alpha$ be a symmetry of $T$. By the
Wigner--Kadison theorem~\cite{Bargmann1964,Kadison1965} the action on each wire
is $\mathrm{Ad}_{U}$ or
$\mathrm{Ad}_{U} \circ \PT{}$, and the two kinds cannot coexist. Either every
wire carries a unitary action or every wire carries an antiunitary one.
\end{lemma}

\begin{proof}
Suppose the first wire carries $\mathrm{Ad}_{U_{1}}$ and the second
$\mathrm{Ad}_{U_{2}} \circ \PT{}$. The collective action on the pair is
$\mathrm{Ad}_{U_{1} \otimes U_{2}} \circ \PT{2}$, which sends
$\proj{\phi^{+}}$ to a unitary conjugate of $\SWAP/d$, with least eigenvalue
$-1/d$. That is not a state, so $\alpha$ does not preserve the state space and
is not a symmetry.
\end{proof}

So a symmetry of a theory with unrestricted states and tensor-product
composition acts the same way on every wire or not at all, and no inert factor
is available to absorb the classification.

\section{A quantum theory restricted by partial transposition}
\label{sec:theory}

\subsection{States, effects, operations}
\label{sec:theory:def}

Throughout, a \emph{wire} carries a finite-dimensional complex Hilbert space,
and a composite is formed by the ordinary tensor product. We emphasise this at the outset. The theory constructed below differs from
quantum theory only in
which effects and which operations it admits, and not in how systems compose.

Partial transposition enters below as the operation defining the restriction,
rather than as a separability test in the sense of Peres and the
Horodeckis~\cite{Peres1996,HHH1996,HHHH2009}, although the two coincide often
enough that the distinction is worth keeping in view.

Fix once and for all an orthonormal basis on every wire. Complex conjugation
$C$ in that basis acts on Hermitian operators as transposition, and for a
subset $S$ of the wires of a composite we write $\PT{S}$ for the corresponding
partial transposition and $\ptr{X}{S} := \PT{S}(X)$. Two elementary properties are used constantly and we record them here. $\PT{S}$ is an involution, and
\begin{equation}
  \label{eq:ptcompose}
  \PT{S} \circ \PT{V} \;=\; \PT{S \triangle V} ,
\end{equation}
where $\triangle$ is the symmetric difference. So $S \mapsto \PT{S}$ is an
isomorphism from the power set of the wires under symmetric difference onto a
group of linear maps, and the partial transpositions of a composite of $n$
wires form the elementary abelian group $\Z_2^{\,n}$, commutative and with
every element an involution.

Two consequences are used repeatedly. Because $\PT{S^{c}}$ differs from
$\PT{S}$ by the global transpose, which preserves spectra, any spectral
condition imposed for $S$ is equivalent to the one imposed for $S^{c}$, so
there are $2^{\,n-1}$ inequivalent conditions rather than $2^{\,n}$. And of the
$2^{n}$ elements of the group exactly one, the identity, is an operation of
quantum theory: the global transpose $\PT{W}$ is positive but not completely
positive, and every other element fails even positivity. Definition~\ref{def:T}
below can be read off the group on that account. Its effect set is the largest
subset of the unit order interval the group stabilises, and its operation set
the largest class of channels stable under the conjugation action
$\Lambda \mapsto \PT{S}\Lambda\PT{S}$.

Adjoining the wire relabellings enlarges the group to the semidirect product
$\Z_{2}^{\,n} \rtimes S_{n}$, since $\pi \PT{S} \pi^{-1} = \PT{\pi(S)}$, but
relabelling belongs to the wiring rather than to the operations and plays no
further part.

\begin{definition}[The theory $\Th$]
\label{def:T}
$\Th$ is the generalized probabilistic theory whose systems are
finite-dimensional quantum systems composed by the tensor product, and whose
states, effects and operations are
\begin{align}
  \stateset
    &= \bigl\{\, \rho \;:\; \rho \ge 0, \ \Tr\rho = 1 \,\bigr\},
      \label{eq:defstates}\\[2pt]
  \effset
    &= \bigl\{\, E \in \Herm \;:\; 0 \le \ptr{E}{S} \le \Id \ \ \forall S \,\bigr\},
      \label{eq:defeffects}\\[2pt]
  \opset
    &= \bigl\{\, \Lambda \;:\; \PT{V} \circ \Lambda \circ \PT{V}
       \ \text{is CP} \ \ \forall V \,\bigr\} ,
      \label{eq:defops}
\end{align}
where $\Lambda$ ranges over the trace-preserving maps, $S$ over the subsets of
the wires an effect acts on, and $V$ over the subsets of the wire labels of an
operation.
\end{definition}

The three lines are of quite different character and each needs a comment.

\paragraph*{States.} Equation~\eqref{eq:defstates} imposes no restriction at all. Every density
matrix on every system and every composite is a state of
$\Th$. This is deliberate. A restriction on states could not produce the
phenomenon we are after, because the state cone would then be recovered as the
bidual of the effect set, and the two conditions examined in
Sec.~\ref{sec:conj} would not come apart. We return to this in
Proposition~\ref{prop:dualcone}.

\paragraph*{Effects.} Equation~\eqref{eq:defeffects} is the restriction that
does the work. Taking $S = \emptyset$ recovers the ordinary requirement
$0 \le E \le \Id$, and the remaining subsets impose the same requirement on every
partial transpose of $E$. Equivalently, and more transparently, $E$ is an
effect of $\Th$ precisely when $\ptr{E}{S}$ is a legitimate quantum effect for
every $S$. Section~\ref{sec:theory:whynot} explains why the weaker and more obvious condition, positivity of every partial
transpose with the upper bound imposed only on $E$ itself, does not define a
generalized
probabilistic theory at all.

\paragraph*{Operations.} Equation~\eqref{eq:defops} is the class of completely
PPT-preserving maps of Rains~\cite{Rains1999,Rains2001,Eggeling2001}, in the
multipartite form of Ishizaka and Plenio~\cite{IshizakaPlenio2005}.
It has not, to our knowledge, been taken before as the operation set of a
generalized probabilistic theory. Appendix~\ref{app:cpptp} fixes the
conventions and records what the class was introduced for.

The convention governing $\PT{V}$ in~\eqref{eq:defops} matters and is not the
only one available. An operation $\Lambda: \mathcal{L}(\mathcal{H}_{\rm in})
\to \mathcal{L}(\mathcal{H}_{\rm out})$ has its input and output wires labelled
by a common index set, and $\PT{V}$ transposes the input \emph{and} the output
wires carrying a label in $V$, simultaneously. That the labels can be matched
across the operation is a substantive assumption rather than a bookkeeping convenience. Section~\ref{sec:enc:frames} shows that
allowing an output wire to
carry a label independent of its input would force the Choi matrix to be PPT
across the input:output cut and would thereby exclude every unitary.

Three further points of convention. Permutations of wires belong to the wiring
of a circuit rather than to the operations performed inside it, as is usual in
a process theory, and the encoding of Sec.~\ref{sec:enc} carries them without
comment because a wire's frame register travels with it. Next, Ishizaka and
Plenio state the multipartite
condition with the transposition ranging over \emph{single parties}. For three parties this coincides with ranging over all bipartitions, since $V$ and
$V^{c}$ give the same condition, but from four parties onwards the
single-party family is strictly weaker, because positivity across a $2$:$2$
cut is not implied by positivity across the $1$:$3$ cuts.
Definition~\ref{def:T} takes $V$ over all subsets, which is what the sector
decomposition of Sec.~\ref{sec:enc:ops} requires. We also follow the standard
terminology in distinguishing \emph{completely} PPT-preserving maps, for which
$\PT{V}\Lambda\PT{V}$ is required to be completely positive, from the weaker
PPT-preserving maps, for which it is required only to be positive. It is the former that we need.

A word on the basis. Both the effect restriction and the symmetry are defined
relative to the basis fixed above, and it is the same basis in each case, since
conjugation in a basis acts on Hermitian operators as transposition in that
basis. A single choice therefore fixes $\Th$ and $C$ together, and choices
related by a unitary give unitarily equivalent pairs. In this sense $\Th$ is
not a theory that happens to carry a conjugation symmetry. The conjugation determines it, a statement made precise in
Sec.~\ref{sec:sectorial:emax}.

Finally we fix the object of comparison.

\begin{definition}[The symmetrized subtheory]
\label{def:sT}
$\sTh$ is the subtheory of $\Th$ invariant under $C$: its states are the real
density matrices, its effects the real elements of $\effset$, and its
operations the elements of $\opset$ commuting with $C$.
\end{definition}

Since every state, effect and operation of $\sTh$ is real, $\sTh$ is a
subtheory of real-amplitude quantum theory. Any valid upper bound on the
correlations of $\RQT$ in a given causal structure therefore bounds those of $\sTh$ as well. The main theorem does not need this, being
constructive, but it does mean that a violation found on the $\Th$ side could
never be dismissed as an artefact of the restriction.

\subsection{Why ``PPT effects'' is not the right effect set}
\label{sec:theory:whynot}

The effect set of Definition~\ref{def:T} may look needlessly elaborate, and the obvious simplification of it fails for an instructive reason. Two features of
$\Th$ constrain what that effect set can be. It must be a proper restriction,
since the states are unrestricted and we want conjugation to be nonphysical. And it must be closed under partial transposition, since that is what makes the
push-through argument of Sec.~\ref{sec:conj:weak} available. The natural
candidate meeting both demands is the set of \emph{PPT effects},
\begin{equation}
  \label{eq:pptnaive}
  \effset_{\mathrm{PPT}}
  \; := \;
  \bigl\{\, E \in \Herm \;:\; 0 \le E \le \Id , \ \
  \ptr{E}{S} \ge 0 \ \ \forall S \,\bigr\},
\end{equation}
the effect-side analogue of the PPT cone of states. It is convex and closed, it
contains every product effect, and every pairing $\Tr[E\rho]$ with
$\rho \in \stateset$ lies in $[0,1]$. It is nevertheless not the effect set of
any generalized probabilistic theory.

The obstruction is complementation. An effect that belongs to no measurement is not an effect. For a two-outcome
measurement to register $E$ on one outcome
there must be a complementary effect $\Id - E$ registering the other. The
following lemma shows that $\effset_{\mathrm{PPT}}$ fails this already for the
most familiar two-outcome measurement on a pair of identical systems.

\begin{lemma}
\label{lem:psym}
Let $\SWAP$ denote the swap operator on $\C^{d}\otimes\C^{d}$ and let
$P_{\rm sym} = \tfrac12(\Id + \SWAP)$ and
$P_{\rm asym} = \tfrac12(\Id - \SWAP)$ be the projectors onto the symmetric and
antisymmetric subspaces. Then, for every $d \ge 2$,
\begin{equation}
  P_{\rm sym} \in \effset_{\mathrm{PPT}},
  \qquad
  P_{\rm asym} \notin \effset_{\mathrm{PPT}} .
\end{equation}
The set $\effset_{\mathrm{PPT}}$ is therefore not closed under
$E \mapsto \Id - E$.
\end{lemma}

\begin{proof}
Write $\SWAP = \sum_{i,j} \ketbra{ij}{ji}$. Partial transposition on the first
factor sends $\ketbra{i}{j} \otimes \ketbra{j}{i}$ to
$\ketbra{j}{i} \otimes \ketbra{j}{i}$, whence
\begin{equation}
  \label{eq:swappt}
  \ptr{\SWAP}{A} \;=\; \sum_{i,j} \ketbra{jj}{ii}
  \;=\; d \, \proj{\phi^{+}} ,
\end{equation}
where $\ket{\phi^{+}} = d^{-1/2}\sum_{i}\ket{ii}$. Since $\proj{\phi^{+}}$ is a
rank-one projector,
\begin{equation}
  \label{eq:psympt}
  \ptr{P_{\rm sym/asym}}{A}
  \;=\; \tfrac12\bigl(\Id \pm d \proj{\phi^{+}}\bigr) ,
\end{equation}
where the upper sign goes with $P_{\rm sym}$, have the spectra
\begin{equation}
  \Bigl\{\, \tfrac{1 \pm d}{2} \,\Bigr\}
  \;\cup\;
  \Bigl\{\, \tfrac12 \,\Bigr\}^{\times (d^{2}-1)} .
\end{equation}
For $d = 2$ these read
$\{3/2,\, 1/2,\, 1/2,\, 1/2\}$ for $P_{\rm sym}$
and $\{-1/2,\, 1/2,\, 1/2,\, 1/2\}$ for $P_{\rm asym}$.

Only $S = \emptyset$ and $S = \{A\}$ need examination, the other two reducing
to these. For $S = \emptyset$ one has
$P_{\rm sym}, P_{\rm asym} \ge 0$, both being projectors. For $S = \{A,B\}$ the map is
the global transpose, which preserves the spectrum. And
$\ptr{X}{B} = \bigl(\ptr{X}{A}\bigr)^{\mathsf T}$ for Hermitian $X$, so
$S = \{B\}$ gives the same spectrum as $S = \{A\}$. Hence
$\ptr{P_{\rm sym}}{S} \ge 0$ for every $S$, so $P_{\rm sym} \in \effset_{\mathrm{PPT}}$,
while $\ptr{P_{\rm asym}}{A}$ carries the eigenvalue $(1-d)/2 < 0$, so
$P_{\rm asym} \notin \effset_{\mathrm{PPT}}$.
\end{proof}

The pair $\{P_{\rm sym}, P_{\rm asym}\}$ discriminates between the symmetric and
the antisymmetric subspace, which is what a swap test performs. There is nothing exotic about it, and the violation
is not marginal, the offending eigenvalue being $-1/2$ for a pair of qubits.
What Lemma~\ref{lem:psym} exhibits is therefore a structural defect of the definition~\eqref{eq:pptnaive} rather than
an accident of a badly chosen example.

The source of the defect is visible in the two displayed spectra. The partial
transpose $\PT{S}$ is not a positive map, which is why demanding
$\ptr{E}{S} \ge 0$ is a nontrivial restriction rather than a consequence of
$E \ge 0$. By the same token $E \le \Id$ does not imply $\ptr{E}{S} \le \Id$, and
indeed $\ptr{P_{\rm sym}}{A}$ has an eigenvalue $3/2$. Definition~\eqref{eq:pptnaive}
imposes the lower bound on every partial transpose of $E$ but the upper bound
on none of them except $S = \emptyset$. Complementation is precisely the
operation that exchanges the two bounds,
\begin{equation}
  \label{eq:complbounds}
  \ptr{(\Id - E)}{S} \;=\; \Id - \ptr{E}{S}
  \quad \forall S ,
\end{equation}
so no set specified by such an asymmetric list of constraints can be closed
under it. The defect is not peculiar to the positive partial transpose. Any effect set
defined by a one-sided cone condition has it, and the separable cone gives the
same witness: $P_{\rm sym}$ is separable, being a multiple of the average of
$\proj{\psi}^{\otimes 2}$ over pure states, while $P_{\rm asym}$ is not even
positive under partial transposition.

Symmetrising the list is the only repair, and it returns the effect set of
Definition~\ref{def:T}, which may be restated in one line:
\begin{equation}
  \label{eq:effrestated}
  E \in \effset
  \iff
  \ptr{E}{S} \ \text{is a quantum effect } \forall S .
\end{equation}

Two closure properties are then immediate, and both are used later. First,
$\effset$ is closed under complementation, by~\eqref{eq:complbounds} applied
uniformly in $S$. This is what Proposition~\ref{prop:validgpt} of Sec.~\ref{sec:theory:valid} needs. Second,
$\effset$ is closed under partial transposition, since
$\bigl(\ptr{E}{S}\bigr)^{\mathsf T_{V}} = \ptr{E}{S \triangle V}$ and the
defining condition~\eqref{eq:effrestated} quantifies over all subsets. This is what the push-through argument of Sec.~\ref{sec:conj:weak} needs.

The repair costs nothing in expressive power where it matters. For a product
effect $E_{A} \otimes E_{B}$ one has
$\ptr{(E_{A} \otimes E_{B})}{A} = E_{A}^{\mathsf T} \otimes E_{B}$, and
transposition preserves the spectrum, so
$0 \le E_{A}^{\mathsf T} \otimes E_{B} \le \Id$. Every product effect survives
in $\effset$, and with it the tomographic locality established in
Sec.~\ref{sec:theory:tomo}. What is lost is the swap test. Neither $P_{\rm sym}$ nor $P_{\rm asym}$ belongs to $\effset$, the first because $\ptr{P_{\rm sym}}{A} \not\le \Id$ and
the second because $\ptr{P_{\rm asym}}{A} \not\ge 0$.

We record finally that the upper bound in~\eqref{eq:effrestated} is not
imposed here merely to rescue complementation. It is demanded a second time,
and independently, by the encoding of Sec.~\ref{sec:enc}: by
Lemma~\ref{lem:positivity} the encoded effect $\xieff(E)$ is bounded above by
the unit effect of $\sTh$ if and only if $\ptr{E}{S} \le \Id$ for every $S$. We
return to this coincidence in Sec.~\ref{sec:enc:pos}.

\subsection{\texorpdfstring{$\Th$}{T} is a valid GPT}
\label{sec:theory:valid}

Checking that Definition~\ref{def:T} specifies a generalized probabilistic
theory is not a formality here. Lemma~\ref{lem:psym} has already shown that a
natural candidate for the effect set fails, so the question has content, and
the inheritance arguments available for foil theories built by twirling do not
transfer to a theory specified by cone restrictions. Appendix~\ref{app:validgpt}
explains why and supplies the proofs.

\begin{proposition}
\label{prop:validgpt}
$\Th$ is a valid generalized probabilistic theory. Specifically:
\begin{enumerate}
  \item[(a)] $\stateset$ and $\effset$ are convex and closed, and
        $\Tr[E\rho] \in [0,1]$ for all $\rho \in \stateset$, $E \in \effset$;
  \item[(b)] $0, \Id \in \effset$ and $\effset$ is closed under
        $E \mapsto \Id - E$;
  \item[(c)] states and effects are tomographic for one another;
  \item[(d)] $\stateset$ and $\effset$ are closed under steering;
  \item[(e)] $\opset$ contains the identity, is convex, and is closed under
        sequential and parallel composition;
  \item[(f)] every $\Lambda \in \opset$ maps $\stateset$ into $\stateset$,
        and $\Lambda^{*}$ maps $\effset$ into $\effset$;
  \item[(g)] $\effset$ is closed under tensor products, so that separate
        parties may measure independently.
\end{enumerate}
\end{proposition}

Items (a)--(e) and (g) are routine once the definitions are unwound and are
proved in Appendix~\ref{app:validgpt}. Item (f) is not routine, and it is the step at
which the naive theory of Sec.~\ref{sec:theory:whynot} would have collapsed
even if its effect set had been repaired. With unrestricted operations a
\textsc{cnot} pulls the product effect $\proj{+0}$ back to $\proj{\phi^{+}}$,
whose partial transpose has an eigenvalue $-1/2$. We therefore give the second
half of (f) here.

\begin{proof}[Proof of the second half of {\rm (f)}]
Fix $\Lambda \in \opset$ and $E \in \effset$, and fix a subset $S$ of the
labels. Write
\begin{equation}
  \Lambda_{S} \;:=\; \PT{S}^{\rm out} \circ \Lambda \circ \PT{S}^{\rm in} ,
\end{equation}
which is completely positive by Definition~\ref{def:T} and trace preserving
because partial transposition preserves the trace. Partial transposition is
self-adjoint for the Hilbert--Schmidt pairing, so
\begin{equation}
  \label{eq:adjointsector}
  (\Lambda_{S})^{*}
  \;=\; \PT{S}^{\rm in} \circ \Lambda^{*} \circ \PT{S}^{\rm out} ,
\end{equation}
and consequently, using $\bigl(\PT{S}\bigr)^{2} = \id$,
\begin{align}
  (\Lambda_{S})^{*}\bigl( \ptr{E}{S} \bigr)
  &= \PT{S}^{\rm in}\Bigl( \Lambda^{*}\bigl( \PT{S}^{\rm out}(\ptr{E}{S}) \bigr) \Bigr)
     \nonumber\\
  &= \PT{S}^{\rm in}\bigl( \Lambda^{*}(E) \bigr)
   \;=\; \bigl( \Lambda^{*}(E) \bigr)^{\mathsf{T}_{S}} .
  \label{eq:sectorchain}
\end{align}
Now $(\Lambda_{S})^{*}$ is completely positive and unital, and
$0 \le \ptr{E}{S} \le \Id$ because $E \in \effset$. Positive unital maps
preserve the order interval, so
\begin{equation}
  0 \;\le\; \bigl( \Lambda^{*}(E) \bigr)^{\mathsf{T}_{S}} \;\le\; \Id .
\end{equation}
As $S$ was arbitrary, $\Lambda^{*}(E) \in \effset$.
\end{proof}

The argument uses the label-matching convention of Sec.~\ref{sec:theory:def} in
an essential way.
$\PT{S}$ acts on the input wires on one side of~\eqref{eq:adjointsector} and
on the output wires on the other, and the two are identified by their common
labels. It requires complete positivity of $\Lambda_{S}$ for \emph{every} $S$,
not merely for $S = \emptyset$. This is what forces the operation class to be
the completely PPT-preserving one rather than all channels. And it goes through verbatim for instruments. Take a collection $\{\Lambda_{k}\}$ of completely positive trace-non-increasing
maps, each satisfying~\eqref{eq:defops} and summing to an element of $\opset$.
Each $(\Lambda_{k,S})^{*}$ is then completely positive and subunital, so
that
\begin{equation}
  0 \;\le\; \bigl( \Lambda_{k}^{*}(E) \bigr)^{\mathsf{T}_{S}}
    \;\le\; (\Lambda_{k,S})^{*}(\Id) \;\le\; \Id .
\end{equation}
Instruments are needed at intermediate nodes of a causal structure, and we take
$\opset$ to include them throughout.

The proof also invites a question about the definition. Why not simply take
the operations of $\Th$ to be those channels whose adjoint preserves
$\effset$? What we have just shown is that the completely PPT-preserving class
is contained in that one. The containment is in fact strict, as
Sec.~\ref{sec:main:scope} shows with an explicit witness, but nothing in the
present section depends on that. We keep Definition~\ref{def:T} for three reasons. The class is standard, and
membership in it is decidable by a semidefinite program. Decisively for what
follows, the sector decomposition of Sec.~\ref{sec:enc:ops} produces exactly
the conditions that $\PT{V}\Lambda\PT{V}$ be completely positive, one per
sector, and nothing weaker.

One consequence of the operation restriction should be recorded before we go
on, since it is the kind of thing that can be mistaken for a defect. The state
$\proj{\phi^{+}}$ belongs to $\stateset$, which is unrestricted, but no operation of $\Th$ prepares it. The channel taking the trivial system to
$\proj{\phi^{+}}$ has Choi matrix $\proj{\phi^{+}}$, whose partial transpose
$\SWAP/d$ is not positive. The state space of $\Th$ is therefore strictly
larger than what its own dynamics generates.

There is a precedent for this. The prepare-measure quantum theory of
Ref.~\cite{Ying2025foil}, whose operations are convex mixtures of the identity,
swaps and measure-and-reprepare channels, likewise has every quantum state
available while preparing almost none of them, and is treated there as a
legitimate GPT. In the setting of causal structures the point is in any case immaterial. The states emitted by sources are
primitive data of the scenario,
not outputs of operations performed within it.

\subsection{\texorpdfstring{$\Th$}{T} is tomographically local and nonclassical}
\label{sec:theory:tomo}

Almost everything about $\effset$ follows from one observation. Partial
transposition acts on a product operator factorwise,
\begin{equation}
  \label{eq:prodpt}
  \ptr{(E_{A} \otimes E_{B})}{A} \;=\; E_{A}^{\mathsf T} \otimes E_{B} ,
\end{equation}
and transposition leaves a spectrum unchanged. If $0 \le E_{A} \le \Id$ and
$0 \le E_{B} \le \Id$, then every partial transpose of $E_{A} \otimes E_{B}$ is
again a product of operators with those same spectra, and so lies in
$[0,\Id]$. Every product of quantum effects is therefore an effect of $\Th$,
and the same argument applies to any number of factors.

An elementary system therefore carries no restriction at all. The only
nontrivial subset of a single wire is the wire itself, and the resulting
condition $0 \le E^{\mathsf T} \le \Id$ duplicates $0 \le E \le \Id$. And since the products of quantum effects span
$\Herm(\mathcal{H}_{A} \otimes \mathcal{H}_{B})$, distinct states of a
composite are separated by local measurements. Hence $\Th$ is tomographically
local. This is what places $\Th$ within the scope of the trichotomy of
Sec.~\ref{sec:prelim:trichotomy}, which is formulated for tomographically
local theories~\cite{Ying2025foil}.

The restriction thus bites only on composites, and it bites hard there. The
Bell projector $\proj{\phi^{+}}$ is excluded, since
$\ptr{\proj{\phi^{+}}}{A} = \SWAP/d$ has negative eigenvalues, so the
Bell-state measurement is not available in $\Th$, and by Lemma~\ref{lem:psym}
neither is the swap test. One point of bookkeeping follows. A $d^{2}$-dimensional
elementary system of $\Th$ has all quantum effects, while a $d \otimes d$
composite does not, so a system of $\Th$ is specified by a list of wires and
not merely by a dimension. Composition concatenates the lists, and since the
subsets $S$ in~\eqref{eq:defeffects} range over wires, the effect set of a
multipartite composite does not depend on how the composition is bracketed.

For nonclassicality it is enough to run the standard Bell argument, which the
first paragraph has already licensed. Take $\proj{\phi^{+}}$ from the source
and the settings
\begin{equation}
  A_{0} = Z, \quad A_{1} = X, \quad
  B_{0} = \tfrac{Z+X}{\sqrt2}, \quad
  B_{1} = \tfrac{Z-X}{\sqrt2} ,
\end{equation}
whose associated effects $(\Id \pm A_{x})/2$ and $(\Id \pm B_{z})/2$ are
quantum effects, so that all sixteen products
$\tfrac14(\Id \pm A_{x}) \otimes (\Id \pm B_{z})$ belong to $\effset$. The
resulting correlator is
\begin{equation}
  \langle A_{0}B_{0}\rangle + \langle A_{0}B_{1}\rangle
  + \langle A_{1}B_{0}\rangle - \langle A_{1}B_{1}\rangle
  \;=\; 2\sqrt2 ,
\end{equation}
which exceeds the value $2$ available to any classical model of the Bell
causal structure. $\Th$ is therefore not a subtheory of classical probability
theory, which is the nonclassicality required by Corollary~1 of
Ref.~\cite{Ying2025foil}. We note that the argument is self-contained and does
not rely on the characterisation of classicality that
Ref.~\cite{Ying2025foil} takes from a work listed there as forthcoming.

The same computation makes a second point, about the Bell scenario rather than
about the theory. By
Proposition~2 of Appendix~J of Ref.~\cite{Ying2025foil} a causal structure can exhibit an
$\sTh$--$\Th$ gap only if it exhibits a gap between $\Th$ and classical
probability theory, and the Bell scenario does. The same has to be checked
separately for the bilocality scenario, which is done in
Sec.~\ref{sec:main:necessary} and Appendix~\ref{app:bilocal}. But the four settings above
are real matrices and so is $\proj{\phi^{+}}$, so the entire protocol runs
inside $\sTh$, which reaches $2\sqrt2$ as well. The Bell scenario passes the
necessary condition and is nonetheless blind to the distinction we are after.
This is the same obstacle that led Renou \emph{et al.} to the bilocality
scenario~\cite{Renou2021}, and we return to it in Sec.~\ref{sec:main:necessary}.

\section{Weak nonphysicality without complete positivity}
\label{sec:conj}

\subsection{Failure of complete positivity}
\label{sec:conj:notcp}

On the real vector space of Hermitian operators, complex conjugation in the
basis fixed in Sec.~\ref{sec:theory:def} is a linear involution and coincides
with the transpose,
\begin{equation}
  \label{eq:conjistranspose}
  C(X) \;=\; \bar X \;=\; X^{\mathsf T}
  \qquad (X = X^{\dagger}),
\end{equation}
so that $C = \PT{W}$ with $W$ the full set of wires, and the partial action of
$C$ on a subset $S$ of the wires of a composite is $\PT{S}$.

That $C$ is a symmetry of $\Th$ is a matter of five checks, each of which has
already been prepared. It is reversible, being an involution. It preserves
$\stateset$, since $\bar\rho$ is a density matrix whenever $\rho$ is. It
preserves $\effset$, because $\bigl(\ptr{E}{W}\bigr)^{\mathsf T_{S}} =
\ptr{E}{S^{c}}$ and the defining condition~\eqref{eq:effrestated} quantifies
over all subsets, so that $S \mapsto S^{c}$ merely permutes the conditions. The same argument with $\PT{V} \mapsto \PT{V^{c}}$ shows that
$\Lambda \mapsto C \Lambda C$ preserves $\opset$. And it preserves the pairing,
since $\Tr[\bar E \bar\rho] = \Tr[(\rho E)^{\mathsf T}] = \Tr[E \rho]$.

It is not, however, an operation of $\Th$. A transformation of a generalized
probabilistic theory must act not only on a system but on that system as part
of a composite, and conjugation fails at the first opportunity. Applying
$\PT{A}$ to~\eqref{eq:swappt} and using $\PT{A}^{2} = \id$,
\begin{equation}
  \label{eq:conjphiplus}
  (C_{A} \otimes \id)\bigl( \proj{\phi^{+}} \bigr)
  \;=\; \ptr{\proj{\phi^{+}}}{A}
  \;=\; \frac{\SWAP}{d} ,
\end{equation}
whose spectrum is $\{ +1/d, -1/d \}$, with multiplicities $d(d+1)/2$ and
$d(d-1)/2$. Since $\stateset$ is unrestricted, $\proj{\phi^{+}}$ is a state of
$\Th$, and its image is not. Hence $C$ is nonphysical in $\Th$.

The rest of the paper turns on two readings
of~\eqref{eq:conjphiplus}. The failure belongs to the composite and not to $C$
itself. On a single system $C$ maps states to states,
and in the language of quantum channels it is positive but not completely
positive. The second is that $\SWAP/d$ is Hermitian and has unit trace, so of
the three defining properties of a state it violates exactly one. Nothing in
the argument so far rules out the possibility that the missing positivity is
invisible to the effects of $\Th$, and Sec.~\ref{sec:conj:weak} shows that it
is.

\subsection{Adjoining the symmetry breaks nothing}
\label{sec:conj:weak}

Section~\ref{sec:conj:notcp} established that $C$ is nonphysical. We now show
that it is only weakly so, in the sense of
Sec.~\ref{sec:prelim:trichotomy}. Table~\ref{tab:classification}, in
Sec.~\ref{sec:sectorial:emax}, records what that buys and what it leaves
open.
The proof needs nothing beyond two closure properties already in hand. The
effect set $\effset$ is closed under partial transposition, by the discussion
following Eq.~\eqref{eq:effrestated}, and it is closed under the Heisenberg
action of the operations of $\Th$, by Proposition~\ref{prop:validgpt}(f). What
makes the argument work is that these two closures are available
simultaneously, so that a conjugation inserted anywhere in a circuit can be
absorbed into the effect.

\begin{proposition}
\label{prop:weaklynonphysical}
Complex conjugation is weakly nonphysical in $\Th$. Explicitly, let a closed
circuit be assembled from a state $\rho \in \stateset$, operations
$\Lambda_{1},\dots,\Lambda_{n} \in \opset$ and a measurement
$\{E_{k}\} \subset \effset$ with $\sum_{k} E_{k} = \Id$, and let partial
conjugations $\PT{S}$ be inserted at arbitrary points of the circuit and on
arbitrary subsets $S$ of the wires present there. Then the resulting numbers
$p(k)$ lie in $[0,1]$ and satisfy $\sum_{k} p(k) = 1$.
\end{proposition}

\begin{proof}
Write the circuit value in the Heisenberg picture. If
$\mathcal{M}_{1},\dots,\mathcal{M}_{m}$ denotes the sequence of maps applied to
$\rho$, each $\mathcal{M}_{j}$ being either an operation $\Lambda \in \opset$
or an inserted conjugation $\PT{S}$, then
\begin{equation}
  \label{eq:heisenbergvalue}
  p(k)
  \;=\; \bigl\langle E_{k},\, (\mathcal{M}_{m} \circ \cdots \circ
        \mathcal{M}_{1})(\rho) \bigr\rangle
  \;=\; \langle E'_{k},\, \rho \rangle ,
\end{equation}
where
\begin{equation}
  \label{eq:pushedeffect}
  E'_{k} \;=\;
  \bigl( \mathcal{M}_{1}^{*} \circ \cdots \circ \mathcal{M}_{m}^{*}
  \bigr)(E_{k}) .
\end{equation}
Each factor in~\eqref{eq:pushedeffect} maps $\effset$ into itself. For a factor
$\Lambda^{*}$ this is Proposition~\ref{prop:validgpt}(f). For a factor
$\PT{S}^{*} = \PT{S}$, partial transposition being self-adjoint for the
Hilbert--Schmidt pairing, this is the closure of $\effset$ under partial
transposition. Hence $E'_{k} \in \effset$, and since $\rho \in \stateset$ we
obtain $p(k) = \langle E'_{k}, \rho \rangle \in [0,1]$.

For normalisation, note that every factor in~\eqref{eq:pushedeffect} is linear
and unital: $\Lambda^{*}(\Id) = \Id$ because $\Lambda$ is trace preserving, and
$\PT{S}(\Id) = \Id$. Applying~\eqref{eq:pushedeffect} to
$\sum_{k} E_{k} = \Id$ gives $\sum_{k} E'_{k} = \Id$, so
$\sum_{k} p(k) = \Tr\rho = 1$.
\end{proof}

The proposition does not say that the intermediate objects in the circuit are
states of $\Th$. They are not. By
Eq.~\eqref{eq:conjphiplus} a conjugation applied to one half of $\proj{\phi^{+}}$ produces $\SWAP/d$, which
is not positive, and no reading of Proposition~\ref{prop:weaklynonphysical}
repairs that. What the proposition says is that the departure from the state
cone is never in a direction that an effect of $\Th$ can resolve. The theory is consistent because what goes wrong inside it is invisible from
the outside, and not because nothing goes wrong.

The contrast with quantum theory is instructive, and it is a contrast of
exactly the kind the title of this paper advertises. Quantum theory has the
larger effect set, and being larger is precisely what prevents it from being closed under partial transposition. The projector
$\proj{\phi^{+}}$ is a legitimate quantum
effect and $\ptr{\proj{\phi^{+}}}{A}$ is not. The push-through therefore fails
at the first step, and the failure is not abstract. Taking $P_{\rm asym}$ from
Lemma~\ref{lem:psym} as the effect and $\SWAP/d$ as the object presented to it,
\begin{equation}
  \label{eq:negativeprob}
  \Tr\Bigl[ P_{\rm asym} \, \frac{\SWAP}{d} \Bigr]
  \;=\; \frac{\Tr[\SWAP] - \Tr[\SWAP^{2}]}{2d}
  \;=\; \frac{1-d}{2} ,
\end{equation}
which equals $-1/2$ for a pair of qubits. Adjoining $C$ to quantum theory thus
assigns a negative probability to the antisymmetric outcome of a swap test, and
$C$ is strongly nonphysical there. The effect that detects the inconsistency is
$P_{\rm asym}$, and that is exactly the effect that Lemma~\ref{lem:psym} removed
from $\effset$ when the effect set was symmetrized. Restricting the effects
does not repair the symmetry. It removes the instruments that would have
registered its failure.

This leaves an obvious question. If $\SWAP/d$ is not a state of $\Th$ and yet
pairs admissibly with every effect of $\Th$, where does it live? The answer is
the subject of Sec.~\ref{sec:conj:dual}, and it is the structural reason the
two conditions examined in this section can come apart at all.

\subsection{Where complete positivity and strong nonphysicality come apart}
\label{sec:conj:dual}

Proposition~\ref{prop:weaklynonphysical} is a computation, and a computation
does not by itself say why its conclusion should have been expected. The reason
is a statement about cones, and it identifies precisely which feature of $\Th$
is responsible.

Work in the real vector space $\Herm$ of Hermitian operators on a composite of
$n$ wires, with the Hilbert--Schmidt pairing. The cone generated by the effect
set is
\begin{equation}
  \label{eq:effcone}
  \mathcal{K} \;:=\; \mathrm{cone}(\effset)
  \;=\; \bigcap_{S} \PT{S}(\PSD) ,
\end{equation}
since the upper bounds in~\eqref{eq:defeffects} are removed by rescaling. Only
$2^{\,n-1}$ of these cones are distinct, because $\PT{S^{c}}$ and $\PT{S}$
differ by the global transpose, which preserves positivity.

\begin{proposition}
\label{prop:dualcone}
The dual of the effect cone is
\begin{equation}
  \label{eq:dualcone}
  \mathcal{K}^{*} \;=\; \sum_{S} \PT{S}(\PSD) \;\supsetneq\; \PSD ,
\end{equation}
the inclusion being strict whenever $n \ge 2$. Consequently the set
\begin{equation}
  \label{eq:omegahat}
  \widehat{\stateset}
  \;:=\; \bigl\{ X \in \mathcal{K}^{*} : \Tr X = 1 \bigr\}
\end{equation}
of normalised operators pairing admissibly with every effect of $\Th$ strictly
contains $\stateset$, and $\SWAP/d \in \widehat{\stateset} \setminus
\stateset$.
\end{proposition}

\begin{proof}
Each $\PT{S}$ is self-adjoint for the Hilbert--Schmidt pairing and is an
involution, and $\PSD$ is self-dual, so
$\bigl( \PT{S}(\PSD) \bigr)^{*} = \PT{S}(\PSD)$. Dualising~\eqref{eq:effcone}
turns the intersection into the closure of the sum of the duals. That sum is already closed. If $\sum_{S} \PT{S}(P_{S}) = 0$ with every $P_{S} \ge 0$, then
taking traces and using $\Tr[\PT{S}(P_{S})] = \Tr P_{S}$ gives
$\sum_{S} \Tr P_{S} = 0$, whence every $P_{S}$ vanishes. Closed pointed cones that are positively independent in this sense have closed
sum. This gives~\eqref{eq:dualcone} with the inclusion $\PSD \subseteq
\mathcal{K}^{*}$ obtained from the term $S = \emptyset$.

For strictness, and for the last claim, take $S = \{A\}$ and
$P_{S} = \proj{\phi^{+}}$. By Eq.~\eqref{eq:conjphiplus},
$\PT{A}(\proj{\phi^{+}}) = \SWAP/d$, which therefore lies in
$\mathcal{K}^{*}$. It has unit trace, so it lies in $\widehat{\stateset}$. And
it is not positive, so it lies in neither $\PSD$ nor $\stateset$.
\end{proof}

In this language Proposition~\ref{prop:weaklynonphysical} acquires a one-line
proof. A conjugation applied to a state lands in $\PT{S}(\PSD)$, which
is a summand of $\mathcal{K}^{*}$ by~\eqref{eq:dualcone}, and every operation
of $\Th$ preserves each summand separately, since
$\Lambda\bigl(\PT{V}(P)\bigr) = \PT{V}\bigl( \PT{V}\Lambda\PT{V}(P) \bigr)$ and
$\PT{V}\Lambda\PT{V}$ is completely positive. Nothing a circuit can do to a
state of $\Th$, using the operations of $\Th$ and any number of conjugations,
ever leaves $\mathcal{K}^{*}$. The effects of $\Th$ are by construction
nonnegative on $\mathcal{K}^{*}$.

The gap between $\stateset$ and $\widehat{\stateset}$ is a failure of the
no-restriction hypothesis, and $\Th$ fails it in both of its usual forms~\cite{Selby2023accessible}. The hypothesis for effects
demands that every operator pairing admissibly with all states be an effect. Here $\stateset^{*} = \PSD$, so it would demand $\effset = [0,\Id]$, and
Lemma~\ref{lem:psym} exhibits the shortfall. The hypothesis for states demands
$\stateset = \widehat{\stateset}$, which Proposition~\ref{prop:dualcone}
denies. The two failures are not independent. If the effect set were unrestricted one would have $\effset^{*} = \stateset^{**} = \stateset$ by
biduality, so no-restriction for effects implies no-restriction for states, and
a theory can violate the latter only by violating the former.

That implication is the reason $\Th$ is built the way it is, and it settles the
question left open in Sec.~\ref{sec:theory:def} of why the states were left
unrestricted. Restricting states alone can never produce a gap between
$\stateset$ and $\widehat{\stateset}$, because the bidual returns the state
cone one started from. The gap requires a restriction on effects, and only a restriction on effects can make an effect
set closed under a map that does not preserve positivity.

The same observation explains why no example of this kind can be found inside
quantum theory. There the no-restriction hypothesis holds, so
$\widehat{\stateset} = \stateset$, and the two conditions examined in this section collapse onto one. A symmetry whose
partial action leaves the state
cone thereby leaves $\widehat{\stateset}$, and is strongly nonphysical. Ying
\emph{et al.} note this themselves, remarking that in quantum theory logical
possibility is not merely necessary for physicality but sufficient, and
attributing the coincidence to the no-restriction
hypothesis~\cite{Ying2025foil}. Failure of complete positivity and strong
nonphysicality are therefore not two names for one condition. They are two conditions that quantum theory happens to identify, and the identification is
an artefact of its lack of restrictions rather than a feature of the symmetry.

\section{The per-source reference-frame encoding}
\label{sec:enc}

\subsection{One frame per source}
\label{sec:enc:constr}

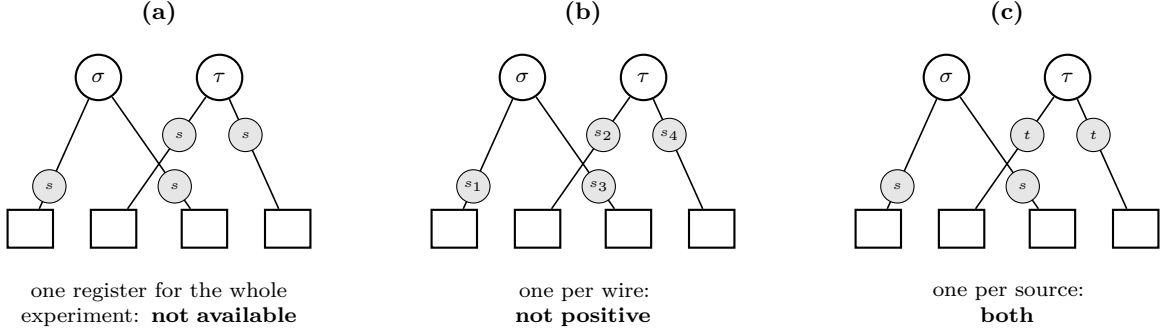
\begin{figure*}[t]
\centering
\begin{tikzpicture}[x=1cm,y=1cm,
  wire/.style={semithick},
  src/.style={circle,draw,thick,fill=white,inner sep=1pt,minimum size=6mm,
              font=\footnotesize},
  pty/.style={rectangle,draw,thick,fill=white,inner sep=1pt,minimum width=6mm,
              minimum height=5mm,font=\footnotesize},
  reg/.style={circle,draw,fill=black!10,inner sep=0pt,minimum size=4.4mm,
              font=\tiny},
  hd/.style={font=\small\bfseries},
  vd/.style={font=\footnotesize,align=center}]
\foreach \i/\x/\ttl/\la/\lb/\lc/\ld in {%
   1/0/{(a)}/{$s$}/{$s$}/{$s$}/{$s$},
   2/5.6/{(b)}/{$s_1$}/{$s_2$}/{$s_3$}/{$s_4$},
   3/11.2/{(c)}/{$s$}/{$t$}/{$s$}/{$t$}}{
  \node[src] (A\i) at (\x+1.3,2.0) {$\sigma$};
  \node[src] (B\i) at (\x+2.9,2.0) {$\tau$};
  \node[pty] (P1\i) at (\x+0.4,0) {};
  \node[pty] (P2\i) at (\x+1.5,0) {};
  \node[pty] (P3\i) at (\x+2.7,0) {};
  \node[pty] (P4\i) at (\x+3.8,0) {};
  \draw[wire] (A\i)--(P1\i); \draw[wire] (A\i)--(P3\i);
  \draw[wire] (B\i)--(P2\i); \draw[wire] (B\i)--(P4\i);
  \node[reg] at ($ (P1\i)!0.28!(A\i) $) {\la};
  \node[reg] at ($ (P2\i)!0.62!(B\i) $) {\lb};
  \node[reg] at ($ (P3\i)!0.28!(A\i) $) {\lc};
  \node[reg] at ($ (P4\i)!0.62!(B\i) $) {\ld};
  \node[hd] at (\x+2.1,2.85) {\ttl};
}
\node[vd,text width=4.3cm] at (2.1,-1.0)
  {one register for the whole\\experiment: \textbf{not available}};
\node[vd,text width=4.3cm] at (7.7,-1.0)
  {one per wire:\\\textbf{not positive}};
\node[vd,text width=4.3cm] at (13.3,-1.0)
  {one per source:\\\textbf{both}};
\end{tikzpicture}
\caption{\label{fig:frames}Where to attach the reference frame. Circles are
sources, squares are parties, shaded discs are frame registers carrying a sign,
and the two sources $\sigma$ and $\tau$ each feed two parties, with their wires
crossing so that grouping by source is visibly not grouping by position.
In~(a) a single register is shared by the whole experiment, which is the
simulation of Ref.~\cite{McKague2009}, and a causal structure with independent
sources does not permit it. In~(b) each wire carries its own register. That is
permitted, but the sign patterns that survive the encoding then transpose
single systems, and by Eq.~\eqref{eq:conjphiplus} positivity fails. In~(c) the
register is attached to the source, so that wires from the same source carry
the same sign and wires from different sources are uncorrelated. This is
permitted for the same reason as~(b), and the surviving patterns now transpose
whole sources, which the effect set of Definition~\ref{def:T} is closed
under.}
\end{figure*}

The simulation of a symmetrized world by local reference frames fails for
$\Th$ in the form given by Ying \emph{et al.}, and it fails for a reason that
points at its own repair. Their map attaches an \emph{independent} frame to
each system. Applied to a bipartite state, the resulting sum contains terms in which the
conjugation acts on one factor and not the other. Those terms are built from a
partial transpose, which is not positive, so the image of an entangled state
need not be a state~\cite{Ying2025foil}. The
diagnosis is that the frame has been attached to the wrong object. Conjugation
is harmless when it acts on a whole state at once, since a full transpose of a
density matrix is a density matrix, and it is dangerous only when it acts on
part of one. What is needed, then, is a frame per \emph{source}, as in
Fig.~\ref{fig:frames}(c). The construction that
follows is a reference-frame encoding in the sense
of Ref.~\cite{BRS2007}, differing from the standard one only in what the frame
is attached to.

This is permitted by the causal structure. A source is the thing that is allowed to emit correlated systems, and the
frames introduced below are
correlated only among the wires of a single source, never across two of them.
Theorem~3 of Ref.~\cite{Ying2025foil}, stated there for twirled worlds and
extended to swirled ones in the same appendix, forbids preparing a reference
frame shared by independent sources. No such frame is used here.

Fix a network causal structure. Its sources are indexed by
$\sigma = 1,\dots,m$, source $\sigma$ emitting a state $\rho_{\sigma}$ on a set
$W_{\sigma}$ of $k_{\sigma}$ wires, and every wire belongs to exactly one
source. Its parties are indexed by $\alpha$, party $\alpha$ holding a set
$W^{\alpha}$ of wires that may be drawn from several sources. To each wire $w$
we adjoin a frame qubit $R_{w}$, emitted by the same source as $w$ and held by the same party. This is the standing assumption that a wire may carry an
enlarged system. Write $P_{\pm}$ for the eigenprojectors of $Y$ on a frame
qubit, and for a sign pattern $\vec s$ on a set of wires write
$V(\vec s) = \{ w : s_{w} = - \}$.

\begin{definition}[The per-source encoding]
\label{def:encoding}
A source emitting $\rho$ on $k$ wires is replaced by the source emitting
\begin{equation}
  \label{eq:stateenc}
  \tilde\rho \;=\; \tfrac12 \Bigl(
    \rho \otimes P_{+}^{\otimes k}
    \;+\; \rho^{\mathsf T} \otimes P_{-}^{\otimes k}
  \Bigr)
\end{equation}
on those wires and their frames, the transpose being taken on all $k$ wires at
once. A party holding $n$ wires and measuring $E$ measures instead
\begin{equation}
  \label{eq:effenc}
  \xieff(E) \;=\; \sum_{\vec s} \ptr{E}{V(\vec s)}
  \otimes \bigotimes_{w} P_{s_{w}} ,
\end{equation}
the sum running over the $2^{n}$ sign patterns on the wires it holds.
\end{definition}

The frames of one source are perfectly correlated by~\eqref{eq:stateenc}. The frames of different sources are independent, since the sources are. Note that
$\xieff$ makes no reference to which source a wire came from. A party simply measures the pair consisting of each wire and its frame.

That $\tilde\rho$ is a state of $\sTh$ takes three lines. It has unit trace,
each of the two terms in~\eqref{eq:stateenc} having unit trace. It is positive,
because $P_{+}^{\otimes k}$ and $P_{-}^{\otimes k}$ are orthogonal projectors,
so the two terms occupy orthogonal sectors and the eigenvalues of $\tilde\rho$
are those of $\rho/2$ together with those of $\rho^{\mathsf T}/2$. And it is
real, because conjugation exchanges $\rho$ with $\rho^{\mathsf T}$ and $P_{+}$
with $P_{-}$ simultaneously, so it exchanges the two terms. Since the states of $\Th$ are unrestricted, those of $\sTh$ are the real
density matrices,
and $\tilde\rho$ is one. It is not in general preparable by an operation of
$\sTh$, but as observed in Sec.~\ref{sec:theory:valid} that is immaterial. What a source
emits is primitive data of the causal structure.

The frame registers alone carry the state
$\tfrac12 ( P_{+}^{\otimes k} + P_{-}^{\otimes k} )$, and for $k \ge 2$ this is
entangled in $\sTh$, exactly as Theorem~3 of Ref.~\cite{Ying2025foil} requires
of anything that can serve as a shared frame. It is separable as a quantum state, being a mixture of two products, but it
satisfies $\langle Y \otimes Y \rangle = 1$. Any product of real density
matrices has $\langle Y \rangle = 0$ on each factor, the trace of a real
symmetric operator against an imaginary antisymmetric one vanishing. No mixture of
product states of $\sTh$ reproduces it.

It is convenient to name the state that the encoded sources jointly prepare.
Writing $N(\vec\epsilon) = \{\sigma : \epsilon_{\sigma} = -\}$ and
$W_{N} = \bigcup_{\sigma \in N} W_{\sigma}$, expansion
of~\eqref{eq:stateenc} gives
$\bigotimes_{\sigma} \tilde\rho_{\sigma} = \Xi\bigl( \bigotimes_{\sigma}
\rho_{\sigma} \bigr)$, where
\begin{equation}
  \label{eq:networkenc}
  \Xi(X) \;=\; 2^{-m} \sum_{\vec\epsilon}
  \ptr{X}{W_{N(\vec\epsilon)}}
  \otimes \bigotimes_{\sigma} P_{\epsilon_{\sigma}}^{\otimes k_{\sigma}} .
\end{equation}
The map $\Xi$ is defined on all of $\Herm$ and not only on product states.
That will matter in Sec.~\ref{sec:main}, where it is applied to the state
obtained after a sequence of operations.

We can now record the identity that makes~\eqref{eq:stateenc}
and~\eqref{eq:effenc} an encoding at all. In the absence of intermediate
operations, which are treated in Sec.~\ref{sec:enc:ops}, the statistics of the
network are reproduced exactly.

\begin{proposition}
\label{prop:pairing}
For any Hermitian $X$ on the system wires and any Hermitian $E_{\alpha}$,
\begin{equation}
  \label{eq:pairingidentity}
  \Tr\Bigl[ \bigotimes_{\alpha} \xieff(E_{\alpha}) \; \Xi(X) \Bigr]
  \;=\;
  \Tr\Bigl[ \bigotimes_{\alpha} E_{\alpha} \; X \Bigr] ,
\end{equation}
the wires being reordered as needed on each side. In particular this holds for
$X = \bigotimes_{\sigma}\rho_{\sigma}$, for which
$\Xi(X) = \bigotimes_{\sigma}\tilde\rho_{\sigma}$.
\end{proposition}

\begin{proof}
Pair the expansion~\eqref{eq:networkenc} of $\Xi(X)$ with the
expansion~\eqref{eq:effenc} of the encoded effects. The frame factors
contribute $\Tr[ P_{s_{w}} P_{\epsilon_{\sigma(w)}} ] = \delta_{s_{w},
\epsilon_{\sigma(w)}}$ for each wire $w$, so a term survives only when the sign
pattern $\vec s$ on the wires is constant on each source and agrees there with
$\vec\epsilon$.

Fix such a term, with $N = N(\vec\epsilon)$. On the effect side the surviving
transposition is $\PT{W_{N}}$ applied to $\bigotimes_{\alpha} E_{\alpha}$, and
on the state side it is $\PT{W_{N}}$ applied to $X$. Since $\PT{W_{N}}$ is self-adjoint for the
Hilbert--Schmidt pairing and involutive,
\begin{equation}
  \Tr\bigl[ \PT{W_{N}}(X) \, \PT{W_{N}}(Y) \bigr] \;=\; \Tr[XY] ,
\end{equation}
every surviving term equals the right-hand side
of~\eqref{eq:pairingidentity}. There are $2^{m}$ such terms, one for each
$\vec\epsilon$, each weighted by $2^{-m}$.
\end{proof}

The proof exhibits the mechanism, which deserves a name. The signs that
survive are constant on each source, so the transposition that appears in a
surviving term is the transposition of a whole source and never of part of one.
On the state side that is a full transpose of $\rho_{\sigma}$, which is a density matrix. On the effect side it is a partial transpose across the
source-of-origin partition, which is admissible precisely because
$\effset$ was defined by~\eqref{eq:effrestated}. Had the frames been attached
to systems rather than to sources, the surviving patterns would have been
arbitrary, a source could have been transposed in part, and the state side of
the identity would have left the positive cone. This is the failure of the
original map, and it is repaired here not by weakening what is asked of the
simulation but by attaching the frame to the object that the causal structure
actually treats as a unit.

The encoding is not yet ready for a general causal structure.
Section~\ref{sec:enc:ops} shows that the operations of $\Th$ can be
encoded as well, with the completely PPT-preserving condition emerging one
sector at a time. Section~\ref{sec:enc:frames} settles what a box does to the
frames of the wires passing through it, and shows that the answer is forced.

\subsection{Encoded effects}
\label{sec:enc:pos}

For the encoding of Definition~\ref{def:encoding} to be of any use, the encoded
effect $\xieff(E)$ must be an effect of $\sTh$ on the enlarged system carrying
both the wires and their frames, and the requirement has three parts of which
the third is easy to overlook. The operator $\xieff(E)$ must be real, since $\sTh$
admits only real effects. It must lie between $0$ and the unit effect. And it
must satisfy the restriction~\eqref{eq:defeffects} on the enlarged wire set,
because $\sTh$ is a subtheory of $\Th$ and inherits that restriction. Positivity on its own would not be enough. The following lemma settles all three at once,
and it turns out that they impose no more than the single condition
$E \in \effset$.

\begin{lemma}[Positivity lemma]
\label{lem:positivity}
Let $E$ be Hermitian on $n$ wires and let $\xieff(E)$ be as in
Definition~\ref{def:encoding}. Then
\begin{enumerate}
  \item[(i)] $\xieff(E)$ is real;
  \item[(ii)] for every subset $U$ of the $2n$ wires of the enlarged system,
        \begin{equation}
          \label{eq:specinvariance}
          \mathrm{spec}\,\ptr{\xieff(E)}{U}
          \;=\; \bigcup_{V} \mathrm{spec}\,\ptr{E}{V} ,
        \end{equation}
        the union running over all subsets $V$ of the original wires and taken
        with multiplicities; in particular the spectrum does not depend
        on $U$;
  \item[(iii)] consequently $\xieff(E)$ is an effect of $\sTh$ if and only if
        $E$ is an effect of $\Th$.
\end{enumerate}
Moreover $\xieff$ is linear, injective and unital.
\end{lemma}

\begin{proof}
The projectors $P_{\pm}$ are orthogonal, of rank one, and sum to the identity
on a frame qubit, so the operators $\bigotimes_{i} P_{s_{i}}$ form a complete
orthogonal family on the frame registers. The sum defining $\xieff(E)$ is
therefore block diagonal with respect to
\begin{equation}
  \mathcal{H}_{\rm sys} \otimes \mathcal{H}_{R}
  \;=\; \bigoplus_{\vec s} \mathcal{H}_{\rm sys} \otimes
        \mathrm{ran}\Bigl( \bigotimes_{i} P_{s_{i}} \Bigr) ,
\end{equation}
each summand of the frame factor being one-dimensional, and the block indexed
by $\vec s$ is $\ptr{E}{V(\vec s)}$. Since $\vec s \mapsto V(\vec s)$ is a
bijection onto the subsets of $\{1,\dots,n\}$, this proves
Eq.~\eqref{eq:specinvariance} in the case $U = \emptyset$.

For (i), recall that conjugation acts on Hermitian operators as the transpose,
so $\overline{\ptr{E}{V}} = \ptr{E}{V^{c}}$ and, since $Y^{\mathsf T} = -Y$,
$\overline{P_{\pm}} = P_{\mp}$. Conjugating the sum therefore complements every
$V(\vec s)$ and flips every sign $s_{i}$ at the same time. These two operations are locked to one another. Replacing $\vec s$ by
$-\vec s$ sends $V(\vec s)$ to
$V(\vec s)^{c}$, so the substitution merely permutes the terms of the sum and
$\overline{\xieff(E)} = \xieff(E)$.

For the general case of (ii), write $U = A \cup B$ with $A$ a set of system
wires and $B$ a set of frame wires, and let $B' = \{ i : R_{i} \in B \}$.
Transposing the system wires in $A$ sends the block $\ptr{E}{V(\vec s)}$ to
$\ptr{E}{V(\vec s) \triangle A}$, by Eq.~\eqref{eq:ptcompose}. Transposing the
frame wires in $B$ exchanges $P_{+}$ and $P_{-}$ on each of them, that is, it
flips $s_{i}$ for $i \in B'$. Reindexing the sum by the flipped signs gives
\begin{equation}
  \label{eq:ptofencoded}
  \ptr{\xieff(E)}{U}
  \;=\; \sum_{\vec t}
  \ptr{E}{\,V(\vec t) \,\triangle\, (A \triangle B')}
  \otimes \bigotimes_{i} P_{t_{i}} ,
\end{equation}
which is again block diagonal in the same frame basis. As $\vec t$ runs over
all sign patterns, $V(\vec t) \triangle (A \triangle B')$ runs over all subsets
of $\{1,\dots,n\}$, because symmetric difference with a fixed set is a
bijection. The multiset of blocks is thus the same as for $U = \emptyset$,
which is Eq.~\eqref{eq:specinvariance}.

Part (iii) follows. By (i) the operator $\xieff(E)$ is real. By (ii) the conditions
$0 \le \ptr{\xieff(E)}{U} \le \Id$, imposed over all subsets $U$ of the enlarged
wire set, reduce to the single family $0 \le \ptr{E}{V} \le \Id$ over subsets
$V$ of the original wires, which is~\eqref{eq:effrestated}.

Linearity is clear from the definition, and injectivity follows because $E$ is
recovered as the block at $\vec s = (+,\dots,+)$. For unitality,
$\ptr{\Id}{V} = \Id$ for every $V$, so
$\xieff(\Id) = \Id \otimes \sum_{\vec s} \bigotimes_{i} P_{s_{i}}
= \Id \otimes \Id_{R}$.
\end{proof}

Unitality is what allows a measurement to be encoded as a measurement. If $\sum_{k} E_{k} = \Id$ then
$\sum_{k} \xieff(E_{k}) = \Id$, so no normalisation is lost in passing to
$\sTh$. And Eq.~\eqref{eq:specinvariance} says something stronger than the
positivity we asked for. The spectrum of the encoded effect is completely
insensitive to which wires of the enlarged system are transposed, so once
$\xieff(E)$ is positive it automatically satisfies the whole
restriction~\eqref{eq:defeffects} on the enlarged system. The reason is
visible in Eq.~\eqref{eq:ptofencoded}: transposing a system wire and
transposing its frame both act on the sum by symmetric difference, and the
family of blocks is closed under that action by construction.

We can now state what was promised in Sec.~\ref{sec:theory:whynot}. Part (iii)
of the lemma says exactly
\begin{equation}
  \label{eq:preimage}
  \effset \;=\; \xieff^{-1}\bigl( \effset(\sTh) \bigr) ,
\end{equation}
so the effect set of $\Th$ is not an ingredient of the construction but its
preimage. This is the second of the demands that fix
Definition~\ref{def:T}, the first being validity as a generalized probabilistic
theory, which required closure under complementation and thereby the same upper
bound. Neither of these two demands is met by the positive-partial-transpose
condition of~\eqref{eq:pptnaive}, and both are met
by~\eqref{eq:effrestated}. A third demand appears in
Sec.~\ref{sec:sectorial:emax}. The theory
$\Th$ does not have to be adjusted for the encoding to work. The encoding returns it.

\subsection{Encoded operations}
\label{sec:enc:ops}

A box standing between a source and a measurement has to do something about
the frames of the wires passing through it, and its options are narrow. It may read them locally, but nothing it reads may leave the box, since the
outcome would then have to be reconciled with the other parties holding wires
from the same source. And it must not disturb them, since downstream parties
still need them. What it can do is act conditionally, leaving the frames alone and applying, on the systems, the channel appropriate to the sign
pattern they carry.

\begin{definition}[Encoded operations]
\label{def:encop}
For $\Lambda \in \opset$ acting on $n$ wires, write
$\Lambda_{V} := \PT{V} \circ \Lambda \circ \PT{V}$ and let
$\ket{\vec s} = \bigotimes_{w} \ket{y_{s_{w}}}$ denote the frame sign basis.
The encoded operation is
\begin{equation}
  \label{eq:opencoding}
  \tilde\Lambda(X)
  \;=\; \sum_{\vec s} \Lambda_{V(\vec s)}
        \bigl( \bra{\vec s} X \ket{\vec s}_{R} \bigr)
        \otimes \proj{\vec s}_{R} ,
\end{equation}
where $\bra{\vec s} X \ket{\vec s}_{R}$ is the block of $X$ at $\vec s$ in that
basis.
\end{definition}

Each $\Lambda_{V}$ is completely positive, by Definition~\ref{def:T}, and so
admits Kraus operators $\{M_{V,j}\}$. In terms of these,
\begin{equation}
  \label{eq:opkraus}
  \tilde\Lambda(X)
  \;=\; \sum_{\vec s,\,j} K_{\vec s,j} \, X \, K_{\vec s,j}^{\dagger} ,
  \qquad
  K_{\vec s,j} = M_{V(\vec s),j} \otimes \proj{\vec s} .
\end{equation}
This step is where the construction meets the restriction imposed in
Sec.~\ref{sec:theory:def}. The operators $M_{V,j}$ exist
for every $V$ precisely because $\Lambda$ is completely PPT-preserving. Had we
required only that $\Lambda$ itself be completely positive, the sectors with
$V \neq \emptyset$ would supply no Kraus operators and~\eqref{eq:opencoding}
would not define a channel at all.

\begin{lemma}
\label{lem:encops}
Let $\Lambda$ be a trace-preserving map on the system wires. Then
\begin{enumerate}
  \item[(i)] $\tilde\Lambda$ is completely positive if and only if
        $\Lambda \in \opset$, and it is then trace preserving and real;
  \item[(ii)] for every subset $U$ of the wires of the enlarged system,
        $\PT{U} \circ \tilde\Lambda \circ \PT{U}$ is again a direct sum of the
        maps $\Lambda_{V}$ over the same index set, so that
        $\tilde\Lambda \in \opset(\sTh)$ as soon as $\Lambda \in \opset$;
  \item[(iii)] $\tilde\Lambda \circ \Xi = \Xi \circ \Lambda$;
  \item[(iv)] encoding commutes with composition:
        $\widetilde{\Lambda_{2} \Lambda_{1}}
         = \tilde\Lambda_{2} \tilde\Lambda_{1}$, and
        $\widetilde{\Lambda_{1} \otimes \Lambda_{2}}
         = \tilde\Lambda_{1} \otimes \tilde\Lambda_{2}$ for operations on
        disjoint wires.
\end{enumerate}
\end{lemma}

\begin{proof}
(i) Sufficiency is~\eqref{eq:opkraus}. For necessity, note that $\Lambda_{V}$
is recovered from $\tilde\Lambda$ as
$Y \mapsto \bra{\vec s} \tilde\Lambda( Y \otimes \proj{\vec s} ) \ket{\vec s}$
with $V = V(\vec s)$, a composition of completely positive maps, so
$\tilde\Lambda$ completely positive forces every $\Lambda_{V}$ to be so. Trace
preservation follows from~\eqref{eq:opkraus} because each $\Lambda_{V}$ is
trace preserving, partial transposition preserving the trace, whence
$\sum_{\vec s, j} K_{\vec s,j}^{\dagger} K_{\vec s,j}
= \sum_{\vec s} \Id \otimes \proj{\vec s} = \Id$. For reality, conjugation
sends $\Lambda_{V}$ to $\PT{W} \Lambda_{V} \PT{W} = \Lambda_{V^{c}}$ and the
frame sector $\vec s$ to $-\vec s$. Since $V(-\vec s) = V(\vec s)^{c}$, the two
relabellings cancel and the sum~\eqref{eq:opencoding} is returned to itself.

(ii) Write $U = A \cup B$ with $A$ a set of system wires and
$B' = \{ w : R_{w} \in B\}$. Transposing a frame wire exchanges $P_{+}$ and
$P_{-}$ on it and leaves the off-diagonal frame blocks in place, so $\PT{U}$
permutes the sectors by $\vec t \mapsto \sigma_{B'}(\vec t)$ while acting on
the systems by $\PT{A}$. Conjugating $\tilde\Lambda$ by $\PT{U}$ therefore
gives a map whose action on the sector $\vec t$ is
\begin{equation}
  \PT{A} \, \Lambda_{V(\vec t) \triangle B'} \, \PT{A}
  \;=\; \Lambda_{V(\vec t)\, \triangle\, (A \triangle B')} ,
\end{equation}
by Eq.~\eqref{eq:ptcompose}. As $\vec t$ runs over all sign patterns the index
$V(\vec t) \triangle (A \triangle B')$ runs over all subsets, so the family of
sector maps is the same for every $U$, and by (i) the whole family is
completely positive exactly when $\Lambda \in \opset$.

(iii) Fix $\vec\epsilon$ and let $V$ be the set of wires held by the box that
belong to sources in $N(\vec\epsilon)$. On that sector $\Xi(X)$ carries
$\ptr{X}{W_{N}}$, and $\tilde\Lambda$ applies $\Lambda_{V}$ to it. Now
$W_{N} = V \cup (W_{N} \setminus V)$ with the second set disjoint from the
wires of the box, so $\PT{W_{N} \setminus V}$ commutes with both $\Lambda$ and
$\PT{V}$, and
\begin{equation}
  \Lambda_{V}\bigl( \ptr{X}{W_{N}} \bigr)
  \;=\; \PT{V} \Lambda \PT{W_{N}\setminus V}(X)
  \;=\; \ptr{\bigl(\Lambda(X)\bigr)}{W_{N}} ,
\end{equation}
which is the sector-$\vec\epsilon$ block of $\Xi(\Lambda(X))$. The sectors that
are not constant on sources carry nothing on either side.

(iv) On the sector $\vec s$ the composite acts by
$\PT{V}\Lambda_{2}\PT{V}\PT{V}\Lambda_{1}\PT{V}
= \PT{V}\Lambda_{2}\Lambda_{1}\PT{V}$, since $\PT{V}^{2} = \id$. For
disjoint sets of wires $\PT{V}$ factorises across them.
\end{proof}

Parts (ii) and (iii) of the lemma are the operational counterparts of parts
(ii) and (iii) of Lemma~\ref{lem:positivity}, and they fail or succeed for the
same reason. In both cases the enlarged system carries two kinds of wire, and
transposing either kind acts on the family of sectors by symmetric difference
with a fixed set. The family is closed under that action, so nothing depends on
which wires are transposed. The same lock is responsible for reality, since conjugation complements the
sector index and the transposition set at the same time, and
does so in Eqs.~\eqref{eq:stateenc}, \eqref{eq:effenc}
and~\eqref{eq:opencoding} alike.

The lemma extends componentwise to instruments. If $\{\Lambda_{k}\}$ are
completely positive and trace non-increasing, each satisfying~\eqref{eq:defops}
and summing to an element of $\opset$, then each $\tilde\Lambda_{k}$ is defined
by~\eqref{eq:opencoding} and the proofs above go through with ``trace
preserving'' weakened to ``trace non-increasing'' throughout. This is what
allows measurements to be performed at intermediate nodes of a causal
structure.

The lemma also closes a question raised in Sec.~\ref{sec:theory:valid}. The
operation set of $\Th$ was there declared to be the completely PPT-preserving
class rather than, say, the channels whose adjoint preserves $\effset$, and the
justification promised was that the encoding produces exactly that condition.
It does: by (i) and (ii) the encoded operation is admissible in $\sTh$ if and
only if every $\Lambda_{V}$ is completely positive, one condition per sector,
with $V$ ranging over all subsets. Nothing weaker will do, and nothing stronger
is asked.

Everything so far has assumed that a wire leaving a box carries the same frame
as the wire entering it. That assumption is not innocuous, and
Sec.~\ref{sec:enc:frames} shows that it is forced.

\subsection{Frame inheritance and extension}
\label{sec:enc:frames}

Three questions about the frames have been deferred and have to be settled
before the encoding can be applied to an arbitrary causal structure. What does
a box do when it holds wires from several sources at once? What frame
does a wire carry if it leaves a box without having entered one? And can a
frame be duplicated at all inside $\sTh$? The answers are forced, and the
second of them is the reason the label-matching convention of
Sec.~\ref{sec:theory:def} was stated as an assumption rather than a
bookkeeping choice.

\paragraph*{Boxes with several inputs need no frame merging.}
A party holding wires from two independent sources holds two frames, and it is
natural to ask whether it must align them. It must not, and it need not. By
Definition~\ref{def:encop} the encoded operation acts on the frame registers as
a measurement in the sign basis followed by a repreparation of the same
outcome, which is the identity on frame-diagonal states, and the encoded
network state~\eqref{eq:networkenc} is frame-diagonal. Explicitly, the frame
marginal of $\tilde\Lambda(X)$ is the sign-basis diagonal of the frame marginal
of $X$, so on encoded inputs it is unchanged. The box therefore leaves $\vec s$ as it found it, and no correlation between
the frames of different sources is created anywhere in the circuit. What the box does jointly process
is a set of registers it already physically holds, which is what a party in the original causal structure does as well.

The price of not merging is the quantification over sectors. Were the two
frames aligned, the sum in~\eqref{eq:opencoding} would collapse to a single
term and the condition on $\Lambda$ would weaken to plain complete positivity.
That this is not available is not a defect of the construction but the content of the problem. A real theory in which all parties share one
frame is complex
quantum theory, which is the observation of McKague, Mosca and
Gisin~\cite{McKague2009} that the
network scenarios were designed to circumvent, and it is forbidden here by
Theorem~3 of Ref.~\cite{Ying2025foil}.

\paragraph*{A fresh output frame would be fatal.}
Suppose instead that a wire leaving a box carried a frame prepared
independently of the frame that entered. The sector index of the output would
then be free of the sector index of the input, and the block of the encoded
operation joining the input sector $\vec s$ to the output sector $\vec t$ would
implement $\PT{V(\vec t)} \circ \Lambda \circ \PT{V(\vec s)}$. Complete
positivity would have to be demanded of every such pair rather than of the
diagonal ones alone. Taking $V(\vec s) = \emptyset$ and $V(\vec t)$ equal to
the full set of output wires, the demand becomes that
$\PT{\rm out} \circ \Lambda$ be completely positive, that is, that the Choi
matrix of $\Lambda$ have positive partial transpose across the cut separating
inputs from outputs. The Choi matrix of a unitary channel is a rank-one
projector onto a vector of Schmidt rank $d$ across that cut, so its partial
transpose has negative eigenvalues whenever $d > 1$. Independent input and
output frames would therefore exclude every unitary from the simulation. They
are not merely wasteful. They empty the construction.

This is a second point at which the present encoding departs from the map of
Ying \emph{et al.}, which twirls the input and output frames independently. For
a symmetry implemented by unitaries the two pictures agree, since a block
$\mathcal{U}_{h} \Lambda \mathcal{U}_{g}$ is completely positive whenever
$\Lambda$ is, and nothing is lost by letting $g$ and $h$ vary independently.
The blocks $\PT{W}\Lambda\PT{V}$ with $W \neq V$ have no such property, and the
difference between a physical and a nonphysical symmetry shows up here in its
most elementary form.

\paragraph*{Frames can be extended inside $\sTh$.}
A wire created by a box must therefore inherit a copy of an existing frame,
aligned with it, and the copy has to be made by an operation of $\sTh$.
It can be. Consider the channel $\mathcal{F}$ with Kraus operators
\begin{equation}
  \label{eq:frameext}
  F_{\pm} \;=\; \ket{y_{\pm} y_{\pm}} \bra{y_{\pm}} ,
\end{equation}
which measures a frame qubit in the sign basis and reprepares two qubits in the
observed sign. It is trace preserving, since
$\sum_{\pm} F_{\pm}^{\dagger}F_{\pm} = P_{+} + P_{-} = \Id$, and it sends
$P_{s}$ to $P_{s} \otimes P_{s}$, which is the required alignment. It is real,
because conjugation sends $F_{\pm}$ to $F_{\mp}$ and so returns the Kraus
family to itself. And it satisfies the restriction~\eqref{eq:defops} on all its
wires: $\mathcal{F}$ is entanglement breaking, being a measurement followed by a
repreparation, and its output is a product, so its Choi matrix is fully
separable and therefore has positive partial transpose across every cut.

The coherent alternative fails, which is what gives the restriction its bite
here. The
isometry sending $\ket{y_{\pm}}$ to $\ket{y_{\pm}y_{\pm}}$ has a rank-one
entangled Choi matrix and is not admissible in $\Th$. Nothing is lost by using
the incoherent copy, because the encoded states are diagonal in the sign basis
by construction and no coherence between sectors is ever called upon. A wire
that is discarded takes its frame with it, and the frames of the surviving
wires of that source remain perfectly correlated among themselves.

The encoding is now complete, and Sec.~\ref{sec:main} assembles its parts into
a statement about arbitrary causal structures. A basis-free description of the
same construction, together with its relation to the standard real-space
simulation, is given in Appendix~\ref{app:poslemma}, and the frame-extension
channel is examined in Appendix~\ref{app:frameext}.

\section{No gap in any causal structure}
\label{sec:main}

\subsection{The main theorem}
\label{sec:main:thm}

Everything needed has been assembled. Sources are replaced according
to~\eqref{eq:stateenc}, measurements according to~\eqref{eq:effenc},
operations according to~\eqref{eq:opencoding}, and a wire created inside the
circuit receives an aligned copy of an existing frame by~\eqref{eq:frameext}.
The one convention that has to be stated as part of the result is that a wire
may carry, alongside its system, a two-dimensional register that originates
with it and travels with it.

\begin{theorem}
\label{thm:main}
Let a causal structure be given, in which each wire is permitted to carry a
reference-frame register emitted by the source of that wire and held by
whichever party holds it. Then
\begin{equation}
  \label{eq:maineq}
  \corr_{\Th} \;=\; \corr_{\sTh}
\end{equation}
in that causal structure. In particular $\Th$ exhibits no
symmetrized--nonsymmetrized causal compatibility gap in any causal structure.
\end{theorem}

\begin{proof}
The inclusion $\corr_{\sTh} \subseteq \corr_{\Th}$ is immediate, $\sTh$ being a
subtheory of $\Th$ by Definition~\ref{def:sT}.

For the converse, let $P$ be realised in the given causal structure by sources
$\rho_{\sigma}$, operations and instruments drawn from $\opset$, and
measurements drawn from $\effset$. Replace each source $\rho_{\sigma}$ by
$\tilde\rho_{\sigma}$, each operation $\Lambda$ by $\tilde\Lambda$, and each
measurement $\{E_{k}\}$ by $\{\xieff(E_{k})\}$, and let every wire created
inside the circuit acquire a frame through $\mathcal{F}$.

Each replacement is a legitimate object of $\sTh$. The encoded sources are real
density matrices, by the discussion following Definition~\ref{def:encoding}.
The encoded effects are effects of $\sTh$ by Lemma~\ref{lem:positivity}(iii),
and form a measurement because $\xieff$ is unital. The encoded operations are
operations of $\sTh$ by Lemma~\ref{lem:encops}(i) and (ii), and instruments
likewise. And $\mathcal{F}$ is an operation of $\sTh$ by
Sec.~\ref{sec:enc:frames}.

The replacement respects the causal structure. No wire is added, removed or redirected. Each frame register is emitted by the source of its wire and
accompanies that wire wherever it goes, so the pattern of causal dependences is
the one the original circuit had. In particular the frames of two independent
sources are never correlated, and by Sec.~\ref{sec:enc:frames} no operation
either reads or disturbs them.

Finally the statistics agree. Write $X$ for the state on the system wires
immediately before the final measurements, obtained from
$\bigotimes_{\sigma}\rho_{\sigma}$ by the operations of the circuit. Applying
Lemma~\ref{lem:encops}(iii) once for each operation, in order, the encoded
circuit prepares $\Xi(X)$ at the same point. Proposition~\ref{prop:pairing},
which holds for arbitrary Hermitian $X$, then gives
\begin{equation}
  \Tr\Bigl[ \bigotimes_{\alpha} \xieff(E_{\alpha}) \; \Xi(X) \Bigr]
  \;=\; \Tr\Bigl[ \bigotimes_{\alpha} E_{\alpha} \; X \Bigr] ,
\end{equation}
outcome by outcome. Hence $P \in \corr_{\sTh}$.
\end{proof}

One consequence carries more than its two-line proof suggests. The theorem
gives $\corr_{\Th} = \corr_{\sTh}$, and $\sTh$ is a subtheory of
real-amplitude quantum theory by the remark following
Definition~\ref{def:sT}, so in every causal structure, and under the same
convention on enlarged wires, $\corr_{\Th} \subseteq \corr_{\RQT}$. The
containment $\sTh \subseteq \RQT$ was recorded in Sec.~\ref{sec:theory:def}
only as a defensive remark, to the effect that a violation found on the $\Th$
side could not be dismissed as an artefact of the restriction. Combined with
the theorem it becomes a statement about $\Th$ itself. No correlation of $\Th$, in any
causal structure, exceeds what real-amplitude quantum theory can
produce. In particular $\Th$ cannot violate the bilocality bound of Renou
\emph{et al.}, which explains structurally why one should not look for a gap by
exhibiting a correlation of $\Th$ above that bound.

The shortest statement of what this paper shows is that a symmetry of a
generalized probabilistic theory can fail to be completely positive and
nevertheless produce no causal compatibility gap in any causal structure. By
Sec.~\ref{sec:conj:notcp} the conjugation symmetry of $\Th$ is not completely
positive, and Theorem~\ref{thm:main} says that it is undetectable. Read
together with Sec.~\ref{sec:conj:dual}, that locates the reason. What a
correlation experiment can see is governed by $\widehat{\stateset}$ rather than
by $\stateset$, and the symmetry never leaves the former.

\subsection{The scope of the claim}
\label{sec:main:scope}

Three qualifications belong with the theorem, and none of them is a formality.

The convention about enlarged wires is the first. It is not a licence peculiar to the present construction. The simulation used by
Ying \emph{et al.}
to prove that no physical symmetry produces a gap attaches a reference-frame
register to every system in exactly the same way, and their result would fail
without it~\cite{Ying2025foil}. We take the same room and use it differently, attaching one register per
source rather than one per system. A reader who declines the convention rejects both theorems
together.

The second qualification concerns what an exact frame buys. The frame constructed
in Sec.~\ref{sec:enc} is perfectly distinguishable, and hypothesis (H1) of
Sec.~\ref{sec:sectorial:G}, which the general criterion needs, demands exactly
that. If a theory offers only an approximate flag, with
$\Tr[Q_{g}\varphi_{g}] = a < 1$, the positivity half of the construction
survives untouched, since the argument of Lemma~\ref{lem:positivity} needs only
that the $Q_{g}$ be positive and sum to the identity, and not that they be
orthogonal. What fails is the statistics. An approximate frame is a perfect one
followed by an independent classical misreading at rate $1-a$ on each wire, so
the encoded outcome distribution differs from the intended one by at most
$2(1-a^{K})$ in $\ell_{1}$ distance, with $K$ the number of wires. The branch
in which every flag is read correctly carries weight $a^{K}$ and reproduces the
target exactly, the remaining weight $1 - a^{K}$ is nonnegative, and the bound
follows. It is attained when the misread branches avoid the support of the
target, so exactness needs $a = 1$.

The consequence is a weaker conclusion, $\corr_{T}$ contained in the closure of
$\corr_{sT}$ rather than equal to it, and the difference is not formal.
Correlation sets need not be closed, as is known already for the Bell scenario,
so a theory whose symmetrized correlations are merely dense in its own would
exhibit a gap by the definition used here while no experiment could ever
resolve one. For the line of work that begins with the claim that
real-amplitude quantum theory is experimentally falsifiable, density is as good
as equality. We record this because it is the natural failure mode of any
attempt to run the construction in a theory that has no exact flag to write on.

Such theories are constrained more than one might expect. The quality of the
best available flag is a continuous function of the state, so on a compact
state space its supremum is attained, and a theory in which arbitrarily good
flags exist but no exact one must therefore have either an unbounded family of
systems or a state space that is not closed. No restriction of a combinatorial
kind will produce one: a bound on circuit depth, on Schmidt rank, or on the
support of a state will not exclude an exact flag, since on a single qubit the
required state is one gate away.

The third qualification concerns the operation class, and it is sharper than
a mere absence of proof. The restriction imposed in Definition~\ref{def:T} is
demonstrably \emph{not} the weakest one compatible with the effects.

Write $\opset_{\max}$ for the channels whose adjoint preserves $\effset$, which
is the largest class any theory with these states and effects could have.
Dualising, and using $\effset = \mathcal{K} \cap (\Id - \mathcal{K})$ together
with unitality of $\Lambda^{*}$, membership of $\opset_{\max}$ is
$\Lambda(\mathcal{K}^{*}) \subseteq \mathcal{K}^{*}$, that is, preservation of
the sum in Eq.~\eqref{eq:dualcone}. The completely PPT-preserving condition
asks for preservation of every summand separately, which is the observation
recorded after Proposition~\ref{prop:dualcone}, so the containment is clear and
it is strict. A witness on two qubit wires is the replacement channel
$X \mapsto \Tr[X]\,\proj{\phi^{+}}$, whose adjoint sends $E$ to
$\Tr[E \proj{\phi^{+}}]\,\Id$ and therefore preserves $\effset$, while its
Choi matrix has a partial transpose with least eigenvalue $-1/2$. This is the
observation of Sec.~\ref{sec:theory:valid} in another guise. No operation of
$\Th$ prepares $\proj{\phi^{+}}$, and the channel that does witnesses that the operation class is smaller than the effects alone require.

The question can be made precise. The class
$\opset_{\max}$ is not stable under tensor products, since $\effset$ on a
composite is not determined by $\effset$ on the factors, so it is not itself the
answer. The largest valid class, which we write $\opset_{\rm valid}$, is
obtained by demanding preservation of $\mathcal{K}^{*}$ in the presence of an
ancilla. In Choi form, and with the convention of
Appendix~\ref{app:sectors},
\begin{equation}
  \label{eq:opsvalid}
  \opset_{\rm valid} \;=\;
  \bigcap_{A} \PT{A_{\rm in}}\Bigl( \sum_{U} \PT{U_{\rm out}}(\PSD) \Bigr)
  \;\cap\; \PSD ,
\end{equation}
with $A$ ranging over subsets of the input labels and $U$ over subsets of the
output labels. We have checked this characterisation against the definition it
came from rather than only asserting it. Testing directly whether
$\Lambda^{*} \otimes \id$ preserves $\effset$ on the system together with an
ancilla wire sorts the identity, the swap, a depolarising channel and the
replacement channel of the previous paragraph into the class and leaves
{\sc cnot} outside it,
which is what~\eqref{eq:opsvalid} predicts in each case.

Against this the encoding reaches further than
Definition~\ref{def:encop} suggests. A box may relabel the frame sector, and
the statistics then require only that $\Lambda$ decompose as
$\sum_{V} \PT{V} \circ \Xi_{V}$ with each $\Xi_{V}$ completely PPT-preserving.
Call that class $\opset_{2}$. Both of its containments are strict, and each
has a witness. For the lower one, take two qubit wires and the channel $\Lambda$
with Kraus operators
\begin{equation}
  \label{eq:o2witness}
\begin{split}
  &\tfrac{1}{\sqrt2}\ketbra{11}{00} , \;\;
  \tfrac{1}{\sqrt2}\ketbra{00}{11} , \;\;
  \tfrac{1}{\sqrt2}\ketbra{10}{01} , \;\;
  \tfrac{1}{\sqrt2}\ketbra{01}{10} , \\
  &\tfrac{1}{\sqrt2}\bigl( \proj{01} - \proj{10} \bigr) , \;\;
  \tfrac{1}{\sqrt2}\ketbra{01}{\phi^{+}} , \;\;
  \tfrac{1}{\sqrt2}\ketbra{10}{\phi^{-}} ,
\end{split}
\end{equation}
and set $\Xi = \PT{\{1,2\}} \circ \Lambda$, so that
$\Lambda = \PT{\{1,2\}} \circ \Xi$ and the decomposition has the single term
$V = \{1,2\}$. The Kraus operators satisfy
$\sum_{k} L_{k}^{\dagger} L_{k} = \Id$, so $\Lambda$ is a channel, and its
Choi matrix has spectrum $\{0^{9}, (1/2)^{6}, 1\}$. The Choi matrix of $\Xi$
has the same spectrum, and both of the partial transposes that act on the input
and output copies of the same wire have spectrum
$\{0^{7}, (1/4)^{6}, (3/4)^{2}, 1\}$, so $\Xi$ is completely PPT-preserving.
But $\PT{1} \Lambda \PT{1}$ is not completely positive, the corresponding
partial transpose of the Choi matrix of $\Lambda$ having least eigenvalue
$-1/2$ on $(\ket{0000} + \ket{1111})/\sqrt2$. So
$\Lambda \in \opset_{2} \setminus \cpptp$. Since the $L_{k}$ are real,
$\Xi$ is also $\Lambda$ with its input transposed,
$\Xi = \Lambda \circ \PT{\{1,2\}} = \PT{\{1,2\}} \circ \Lambda$, so the two
transpositions may be taken on either side. Such a map exists because an
element of $\opset_{2}$ preserves the intersection
$\bigcap_{S} \PT{S}(\PSD)$, by
$\PT{W}\Lambda(\rho) = \sum_{V} \PT{W \triangle V}(\Xi_{V}(\rho))$, whereas a
completely PPT-preserving map has to preserve each cone $\PT{W}(\PSD)$
separately.

The upper strictness has a one-line proof. If $\Xi$
is completely PPT-preserving then $\PT{A}(\Xi(\Id)) = (\PT{A}\Xi\PT{A})(\Id)$ is
positive for every $A$, so $\Xi(\Id) \in \mathcal{K}$. Since $\mathcal{K}$ is
invariant under partial transposition and closed under addition, every element
of $\opset_{2}$ sends $\Id$ into $\mathcal{K}$. The replacement channel of the
previous paragraph does not. Its output at the unit is
$\Tr[\Id]\,\proj{\phi^{+}}$, whose partial transpose has least eigenvalue $-2$
on two qubit wires, while its Choi matrix $\Id \otimes \proj{\phi^{+}}$ is
fixed by every input-side partial transposition and is positive, so it lies in
$\opset_{\rm valid}$. The obstruction in one line is that a decomposition into
completely PPT-preserving pieces cannot create a state with negative partial
transpose, and $\opset_{\rm valid}$ can.

Comparing~\eqref{eq:opsvalid} with the description of $\opset_{2}$ shows that
the two take a sum and an intersection in opposite orders. In Choi form the
completely PPT-preserving class is
\begin{equation}
  \label{eq:sumint}
  \mathcal{N} \;=\; \bigcap_{W} \PT{W_{\rm in}}\PT{W_{\rm out}}(\PSD) ,
  \qquad
  \opset_{2} \;=\; \Bigl( \sum_{V} \PT{V_{\rm out}}(\mathcal{N}) \Bigr)
  \cap \PSD ,
\end{equation}
the second identity following from
$J(\PT{V} \circ \Xi) = \PT{V_{\rm out}}(J(\Xi))$, which is
Eq.~\eqref{eq:choimismatch} with output set $V$ and empty input set. So $\opset_{2}$
intersects first and sums afterwards, while $\opset_{\rm valid}$
of~\eqref{eq:opsvalid} sums first, forming
$\mathcal{D} = \sum_{U} \PT{U_{\rm out}}(\PSD)$, and intersects afterwards.
The order is what separates them, and both ends of~\eqref{eq:opschain} are
strict, the lower by Eq.~\eqref{eq:o2witness} and the upper by the certificate
above. One caution belongs here. The two
are not the same pair of cones with the operations exchanged, and in particular
$\bigcap_{A} \PT{A_{\rm in}}(\mathcal{D})$ is not contained in $\PSD$, so the
intersection with $\PSD$ in~\eqref{eq:opsvalid} cannot be dropped. On one
input and one output qubit wire, $\PT{\rm out}(\proj{\phi^{+}})$ lies in
$\PT{A_{\rm in}}(\mathcal{D})$ for both $A$ and has least eigenvalue
$-\tfrac12$.

The conclusion of Theorem~\ref{thm:main} does not extend to all of
$\opset_{\rm valid}$, and saying so precisely needs the right formulation. The
requirement that the blocks $\Xi_{V}$ be operations of the theory pins them
down only once the operation class is fixed. Treating the class as the unknown,
the object to look for is the largest
$\opset' \subseteq \opset_{\rm valid}$ such that every
$\Lambda \in \opset'$ admits a decomposition of the above form with blocks
again in $\opset'$. Everything else in the proof of Theorem~\ref{thm:main} is
indifferent to which class this is, so that class, whatever it is, is gap-free.
The class $\opset_{2}$ is the first step of the iteration from below and
$\opset_{\rm valid}$ is the ceiling.

The ceiling is not attained. On two qubit wires drawn from \emph{distinct}
sources, $\bigcap_{A} \PT{A_{\rm in}}(\opset_{\rm valid})$ is not contained in
$\sum_{V} \PT{V}(\opset_{\rm valid})$, so $\opset_{\rm valid}$ is not its own
fixed point and
\begin{equation}
  \label{eq:opschain}
  \cpptp \;\subsetneq\; \opset_{2} \;\subseteq\; \opset^{*}
  \;\subsetneq\; \opset_{\rm valid} ,
\end{equation}
with $\opset^{*}$ the largest gap-free class the technique reaches. The
separation is certified by two explicit operators recorded with the ancillary
material, one in $\bigcap_{A} \PT{A_{\rm in}}(\opset_{\rm valid})$ and one in the dual of the
right-hand side, whose pairing is $-0.0573$. The archive holds those two
together with the twelve operators of the decomposition that supports them.
Thirteen of the fourteen carry an eigenvalue claim apiece, the exception being
the pairing, which is a trace computed from two of them rather than a spectrum.
Table~\ref{tab:numerics} records their number and the worst margin. Every
eigenvalue claim in the certificate is the least eigenvalue of an operator given
in closed form, so no solver enters the verification. The smallest margin is $1.6 \times 10^{-3}$ against residuals of
order $10^{-16}$.

Two wires from distinct sources is not a convenience. The condition the witness
violates is the one indexed by a single input wire, and in the encoding the
index set is the collection of the box's wires belonging to sources in the
negative sector, so a singleton is available only when the box holds wires from
two different sources. With both wires from one source the only nontrivial
index is the full set, and that case is settled by taking $V$ to be the outputs,
which returns the global transpose of the Choi matrix, under which
$\opset_{\rm valid}$ is invariant.

Where the iteration below $\opset_{\rm valid}$ terminates we do not know, but
the strict inclusion is stable wherever it does: the map is monotone, so every
later stage of the iteration excludes the same witness. Two qualifications
attach to the definition rather than to the result. The decomposition considered
here is the frame-diagonal one of Definition~\ref{def:encop}, in which the box
reads the frame in the sign basis and reprepares a sector, and an operation of
the enlarged symmetrized theory with frames entangled across sectors is a larger
class we have not examined. And $\opset'$ as defined asks only for the
decomposition property, whereas validity as a generalized probabilistic theory
also asks for closure under composition. Imposing that can only shrink the
class, so the strict inclusion survives.

The obstruction here is not the one that separated $\opset_{2}$ from
$\opset_{\rm valid}$, and the two witnesses together say something sharper than
either alone. The replacement channel sends $\Id$ out of $\mathcal{K}$ and is
nonetheless decomposable with blocks in $\opset_{\rm valid}$. The present
witness sends $\Id$ into $\mathcal{K}$, with a strictly positive margin, and is
not decomposable. So the condition $\Lambda(\Id) \in \mathcal{K}$ is neither
necessary nor sufficient for membership in the decomposable class
$\sum_{V} \PT{V}(\opset_{\rm valid})$. Note that this is a statement about the
image of the unit and not about $\Lambda(\mathcal{K}) \subseteq \mathcal{K}$,
which is an invariant that every element of $\opset_{2}$ satisfies and which
the present witness has not been tested against.

That the ceiling is not reached does not give a gap. Showing that an
encoding fails is not showing that no encoding succeeds, and an impossibility
argument of the kind Renou \emph{et al.} supply for real-amplitude quantum
theory would be needed, bounding the correlations of the symmetrized theory
from above and exhibiting a correlation of the unsymmetrized theory above the
bound~\cite{Renou2021}. By Sec.~\ref{sec:main:necessary} such an argument cannot live in the
bilocality scenario, which is blind to the operation class. Intermediate
processing is not by itself enough either, and it is worth saying how much more
is needed.

\begin{proposition}
\label{prop:defer}
Suppose that in a causal structure every operation has all of its output wires
delivered to a single party. Then the correlations do not depend on the
operation class at all, and are those of a product measurement on the state the
sources emit.
\end{proposition}

\begin{proof}
An operation whose outputs all reach one party may be deferred to that party,
which applies it on receipt while forwarding untouched whatever it does not
process. By Sec.~\ref{sec:theory:def} the causal structure records which party
holds which wire and not what is done to it, so this leaves the structure
intact. After deferring every operation each one is internal to a party, so the
adjoint of the whole circuit factors as $\bigotimes_{\alpha} \Lambda_{\alpha}^{*}$
and the effective measurement of party $\alpha$ has elements
$\Lambda_{\alpha}^{*}(E_{a|x})$. These lie in $\effset$ because the operation
set preserves it, and they sum to the identity because $\Lambda_{\alpha}^{*}$
is unital, so the effective measurement is a measurement of the theory.
\end{proof}

Arbitrary depth, ancillas, instruments and adaptivity within a party are all
absorbed this way, and so is a relay that receives wires, processes them and
passes the whole batch on. What survives is a single shape: some party must
correlate what it keeps with what it sends on, which is to say it must act as a
source in its own right. Any argument that the operation class matters has to
be mounted there.

\subsection{The bilocality scenario}
\label{sec:main:necessary}

The value of Theorem~\ref{thm:main} is easiest to see against the necessary
conditions that Ying \emph{et al.} derive for a causal structure to admit a
gap, and the bilocality scenario is the natural test case: it is the structure
in which the real--complex distinction was first made
operational~\cite{Renou2021}. We take the scenario exactly as it appears there,
with settings at the two outer parties and none at the central one. A
distribution realisable classically in that structure has the form
\begin{equation}
  \label{eq:bilocalclassical}
  \int\!\!\int q_{1}(\lambda_{1}) q_{2}(\lambda_{2})
  P(a|x,\lambda_{1}) \, P(b|\lambda_{1},\lambda_{2}) \, P(c|z,\lambda_{2}) ,
\end{equation}
with $q_{1}$ and $q_{2}$ independent.

All four of the necessary conditions are met.

Proposition~2 of Appendix~J of Ref.~\cite{Ying2025foil} requires a gap between $\Th$ and
classical probability theory. There is one, and it is exhibited by a
measurement whose elements are products, so that they belong to $\effset$ by
Sec.~\ref{sec:theory:tomo}. Both sources emit $\proj{\phi^{+}}$. The central
party measures its first wire in the eigenbasis of $Z$, obtaining $j$, and then
measures its second wire with a setting that depends on $j$, obtaining $m$. Its
outcome is the pair $b = (j,m)$ and its POVM elements are
$\proj{j} \otimes G_{m|j}$. Alice measures $Z$ and Charlie two \textsc{chsh}
settings. Alice's outcome then agrees with $j$ with certainty, which in any classical
model of~\eqref{eq:bilocalclassical} forces $j$ to be a function of
$\lambda_{1}$ alone. Independence of the sources then makes the conditional
distribution of $\lambda_{2}$ given $j$ equal to its marginal, and what remains
is a local model in which $j$ plays the role of a setting. The \textsc{chsh} value
between $j$ and $z$ is therefore at most $2$ classically, while the quantum
value is $2\sqrt2$. Appendix~\ref{app:bilocal} gives the details, and also
shows that the same measurement violates the bilocality inequality of
Ref.~\cite{Branciard2012} outright, at $2^{1/4}$, once Alice is given a second
setting.

Corollary~2 of Appendix~J of Ref.~\cite{Ying2025foil} requires the causal
structure to be nonalgebraic, which the
bilocality structure is~\cite{Ying2025foil}. Proposition~1 of the same appendix
requires a district
admitting a gap, and here there is only one. The district decomposition cuts the observed
classical wires but never the latent ones, and the two sources and
three parties remain connected through the central party. Proposition~3 of
that appendix
requires that no reference frame state be preparable by all the latent systems
jointly, and Ying \emph{et al.} classify the bilocality scenario as an instance
where this is so, the two sources being independent.

By Theorem~\ref{thm:main} there is nonetheless no gap. More is true, and it
locates $\Th$ on the scale exactly.

\begin{proposition}
\label{prop:noswap}
Let $n$ independent sources emit $\rho_{i}$ on the wires $A_{i}, i$, and let a
central party apply the effect $E$ to the wires $1, \dots, n$, leaving the
outer parties in the unnormalised conditional state $\sigma$. If
$\ptr{E}{S} \ge 0$ for a subset $S$ of the central wires, then
$\ptr{\sigma}{A_{S}} \ge 0$.

In particular, since $\effset$ imposes $\ptr{E}{S} \ge 0$ for every $S$, the
conditional state prepared by $\Th$ is never negative under partial
transposition. For $n = 2$ this is positivity across the cut $A_{1} : A_{2}$,
which is separability exactly when $d_{A_{1}} d_{A_{2}} \le 6$~\cite{HHH1996},
the dimensions of the central wires playing no part. On outer wires of that
size $\Th$ therefore cannot swap entanglement at all, and
$\corr_{\Th} \subsetneq \corr_{\RQT}$ in the bilocality scenario.
\end{proposition}

Two remarks bound that last sentence, and neither is covered by the proof
above. Above that size the conclusion is weaker, and $\Th$ can swap bound
entanglement. Maximally entangled sources on $d \times d$ wires give
$\sigma = E^{\mathsf T}/d^{2}$ exactly, so a central effect $\lambda F$ with
$F$ a positive-partial-transpose entangled state, and $\lambda$ small enough
that $\lambda F$ and $\Id - \lambda F$ both lie in $\effset$, leaves the outer
parties entangled. Such $F$ exist whenever both factors are at least $2$ and
$d_{A_{1}} d_{A_{2}} \ge 8$~\cite{Horodecki1997bound,Bennett1999upb}.

For $n \ge 3$ the qubit case fails as well, and the reason is not dimension.
Take three qubit sources and let $F = (\Id - P)/4$, with $P$ the projector onto
the Shifts unextendible product basis~\cite{Bennett1999upb}. Its four vectors
are real and product, so $\PT{S}(P) = P$ and hence $\PT{S}(F) = F$ for every
$S$, with spectrum $\{0^{4}, (1/4)^{4}\}$. So $F \in \effset$, and the outer
parties are left in $F^{\mathsf T} = F$, which is not fully separable. For
three or more sources, positivity under every partial transposition is the
strongest statement available.

\begin{proof}
Partial transposition on the outer wires $A_{S}$ commutes with the trace over
the central wires, and $\rho_{i}^{\mathsf{T}_{A_{i}}} =
(\rho_{i}^{\mathsf T})^{\mathsf{T}_{i}}$ for each $i \in S$, so moving those
transpositions across the Hilbert--Schmidt pairing onto $E$ gives
\begin{equation}
    \ptr{\sigma}{A_{S}} \;=\;
  \Tr_{1 \cdots n}\Bigl[ \Bigl( \bigotimes_{i \in S} \rho_{i}^{\mathsf T}
  \otimes \bigotimes_{i \notin S} \rho_{i} \Bigr)
  \bigl( \Id \otimes \ptr{E}{S} \bigr) \Bigr] .
\end{equation}
Every $\rho_{i}^{\mathsf T}$ and every $\rho_{i}$ is positive, and
$\ptr{E}{S}$ is positive by hypothesis. Diagonalising $\ptr{E}{S}$ writes the
right-hand side as a positive combination of diagonal blocks of a positive
operator, so $\ptr{\sigma}{A_{S}} \ge 0$.

For the last claim take $n = 2$ and $S = \{1\}$. For two qubits a state with
positive partial transpose is separable~\cite{HHH1996}, so the conditional
state admits a local model and its {\sc chsh} value is at most $2$. Real-amplitude quantum theory does
better. Both sources emit $\proj{\phi^{+}}$, the central party measures in the
Bell basis, whose four vectors are real, and the outer parties use the real
settings of Sec.~\ref{sec:theory:tomo}. The protocol is real throughout, hence
available in $\RQT$, and it leaves the outer parties with a maximally
entangled state conditioned on each outcome, so their {\sc chsh} value is
$2\sqrt2$.
\end{proof}

The containment just recorded therefore tightens to
\begin{equation}
  \label{eq:chain}
  \corr_{\CPT} \;\subsetneq\; \corr_{\Th} \;\subsetneq\; \corr_{\RQT}
  \;\subsetneq\; \corr_{\QT}
\end{equation}
in the bilocality scenario, the first strictness by the argument of
Sec.~\ref{sec:main:necessary}, the second by
Proposition~\ref{prop:noswap} and the third by Ref.~\cite{Renou2021}. The
middle one carries more information than the containment it replaces. It is not
that $\Th$ happens to stay under the real bound, but that it cannot produce a
conditional state with negative partial transpose, and so cannot produce a
distillable one.

A related construction, in which a noncontextuality inequality derived from the
Pusey--Barrett--Rudolph argument is translated into a classical causal
compatibility inequality for the same scenario, has been given
recently~\cite{Ying2026pbr}.

One limitation of this test should be stated, because it is easy to overlook
and it bears on the qualification of Sec.~\ref{sec:main:scope}. The bilocality
scenario contains no intermediate transformations at all, only two state
injections and three measurements, so its correlations are a function of the
states and the effects alone and are entirely insensitive to which operations
the theory admits. The four checks above therefore say nothing about the
restriction on operations, and no separation obtained in this scenario ever
could. Any gap arising from an enlargement of the operation class must live in
a causal structure with intermediate processing.

The conclusion to draw is about the conditions rather than about $\Th$.
Corollary~2 and Propositions~1 and~3 of Ref.~\cite{Ying2025foil} constrain the
causal structure alone, and
the bilocality structure passes all of them while exhibiting a gap for
real-amplitude quantum theory and none for $\Th$. No condition on the causal
structure can therefore decide the question once the theory is allowed to vary,
and for a fixed theory the conditions of Appendix~J of
Ref.~\cite{Ying2025foil}, taken together, are not sufficient.

That the absence of a shared reference frame is not sufficient for a gap is
already recorded in their Corollary~3, so we should be precise about what the
present example adds. The witnesses offered there fail for reasons internal
to the causal structure. The \textsc{pbr} scenario and a further example are
algebraic, and a third splits into two districts each of which is a Bell
scenario. Each of them therefore admits no gap for \emph{any} theory. The present example is of a
different kind. The structure is nonalgebraic, has a single district, and does
admit a gap. It simply does not admit one for $\Th$. What fails to be
sufficient is not sufficient for a fixed theory, and the obstruction is not
visible in the causal structure at all.

\section{A criterion for the absence of a gap}
\label{sec:sectorial}

Nothing in Secs.~\ref{sec:enc} and~\ref{sec:main} is peculiar to complex
conjugation, and little of it is peculiar to $\Th$. This section extracts what
was actually used. The result is a sufficient condition for the absence of a
gap, stated for an arbitrary finite symmetry group, together with an
identification of the largest theory that satisfies it.

\subsection{The \texorpdfstring{$G$}{G}-frame encoding}
\label{sec:sectorial:G}

Let $G$ be a finite group acting on a GPT $T$ by symmetries $\alpha_{g}$, and
let $sT$ denote the invariant subtheory. For a wire $w$ let $\alpha_{g}^{w}$ be
the action on that wire alone, and for a tuple $\vec g$ indexed by wires write
$\alpha_{\vec g} = \bigotimes_{w}\alpha_{g_{w}}^{w}$. The construction of
Sec.~\ref{sec:enc} generalises verbatim once the frame qubit is replaced by a
register carrying the regular representation.

The general form is the following.

\begin{definition}[$G$-frame encoding]
\label{def:Gencoding}
Adjoin to each wire a register $R_{w}$ with orthonormal basis
$\{\ket{g}\}_{g \in G}$. A source emitting $\rho$ on $k$ wires is replaced by
\begin{equation}
  \label{eq:Gstateenc}
  \tilde\rho \;=\; \frac{1}{|G|} \sum_{g \in G}
  \alpha_{g}^{\otimes k}(\rho) \otimes \proj{g}^{\otimes k} ,
\end{equation}
a party holding $n$ wires and measuring $E$ measures
\begin{equation}
  \label{eq:Geffenc}
  \xieff(E) \;=\; \sum_{\vec g} \alpha_{\vec g}(E)
  \otimes \bigotimes_{w} \proj{g_{w}} ,
\end{equation}
and an operation $\Lambda$ is replaced by the map acting on the sector
$\vec g$ as $\alpha_{\vec g} \circ \Lambda \circ \alpha_{\vec g}^{-1}$.
\end{definition}

In Definition~\ref{def:Gencoding} the group element is recorded classically and
identically on all wires of a source, and is neither read nor disturbed by any box, exactly as in
Sec.~\ref{sec:enc:frames}. Taking $G = \Z_{2}$ with $\alpha$ the conjugation
returns Definition~\ref{def:encoding}: the regular representation of $\Z_{2}$
is carried by a qubit, the two basis states being $\ket{y_{\pm}}$, and
$\alpha_{\vec g}$ is $\PT{V}$ with $V$ the set of wires carrying the nontrivial
element.

Four hypotheses are needed, and the fourth is not implied by the third.
\begin{enumerate}
  \item[(H1)] $T$ contains a system carrying a perfectly distinguishable
        $|G|$-element flag that the symmetry permutes by the regular action.
        Explicitly, there are states $\{\varphi_{g}\}$ and a measurement
        $\{Q_{g}\}$ of $T$ with $\alpha_{h}\varphi_{g} = \varphi_{hg}$,
        $\alpha_{h}Q_{g} = Q_{hg}$ and
        $\Tr[Q_{g}\varphi_{g'}] = \delta_{gg'}$;
  \item[(H2)] the frame state $|G|^{-1}\sum_{g} \proj{g}^{\otimes k}$,
        perfectly correlated across the $k$ wires of a source, is a state
        of $sT$;
  \item[(H3)] $\effset(T)$ and $\opset(T)$ are invariant under every
        $\alpha_{\vec g}$;
  \item[(H4)] the encoded effects and operations satisfy the restrictions of
        $T$ on the \emph{enlarged} systems, wires and frames together.
\end{enumerate}
Hypothesis (H1) is weaker than requiring a system that carries the regular
representation with unrestricted effects on it, and the difference matters. Of
its two demands, the dimension is nearly free and the freedom from restriction
is what binds. Indeed the second demand alone settles the first.

\begin{proposition}
\label{prop:h1free}
Let $G = \Z_{2}$ act on a system of $T$ by a symmetry $\alpha$ that is
nontrivial on its state space, let that state space be compact, convex and
closed under $\alpha$, and suppose the theory satisfies the no-restriction
hypothesis on that system, so that its effects are all of
$\stateset^{\vee}$. Then {\rm (H1)} holds.
\end{proposition}

\begin{proof}
Pick $\rho$ with $\alpha(\rho) \ne \rho$ and a linear functional $g$
separating them, and set $f = (g - g \circ \alpha)/2$, so that
$f \circ \alpha = -f$ and $f \ne 0$. Let $M = \max_{\stateset} f$, finite by
compactness and positive because $f$ takes both signs, the state space being
closed under $\alpha$. For the same reason $\min_{\stateset} f = -M$. Writing
$F$ for the operator representing $f$, put
$Q = \tfrac12 \Id + F/2M$. Then $\langle Q, \rho \rangle \in [0,1]$ for every
state, so $Q \in \stateset^{\vee}$ and $Q$ is an effect by hypothesis, and
$\alpha(Q) = \Id - Q$. Taking $\varphi$ to be a maximiser of $f$ gives
$\langle Q, \varphi \rangle = 1$ and
$\langle Q, \alpha(\varphi) \rangle = 0$, which is (H1) with
$\varphi_{\pm}$ the maximiser and its image and $Q_{\pm} = Q, \Id - Q$.
\end{proof}

The flag $Q$ so produced need not be positive. It lies in the dual cone without
lying in the unit order interval. That is the whole content of the
proposition. A theory can fail (H1) only by restricting its effects, and the
question of when an exact flag exists is a question about the effect set and not
about the states. In particular the criterion that the state space meet the
locus $\langle \varphi, \bar\varphi \rangle = 0$ is a criterion for a
\emph{positive} flag, and is equivalent to (H1) only when the effects are
required to be positive. For $G = \Z_{2}$ the hypothesis asks only for a real
antisymmetric operator with trivial kernel, a complex structure up to scale,
whose spectral projectors then
satisfy $P_{+} + P_{-} = \Id$ and $P_{+}^{\mathsf T} = P_{-}$ automatically. A
qubit supplies one, namely $iY$. So does a continuous-variable mode, since the
momentum operator obeys $\hat p^{\mathsf T} = -\hat p$ in the position basis
and vanishes only on a null set, and the frame states may be taken to be any
wave packets supported on a half of momentum space. Nothing here is approximate. Definition~\ref{def:T} is stated in finite
dimension, so this is an indication that a continuous-variable counterpart of
$\Th$ should exist rather than a construction of one.

Hypothesis (H3) concerns a system in isolation and (H4) a composite, and in
general the second does not follow from the first, since $\xieff(E)$ is a
specific block-diagonal operator on wires and frames whose status is not
settled by the status of its blocks as effects of the factors. There is
nevertheless a broad circumstance in which the implication does hold.

\begin{proposition}
\label{prop:h4free}
Suppose the effect set of $T$ on every system has the form
$\mathcal{C} \cap (\Id - \mathcal{C})$ for a convex cone $\mathcal{C}$ that is
invariant under the action of $G$ and closed under tensor products, meaning
that $X \in \mathcal{C}_{A}$ and $Y \in \mathcal{C}_{B}$ imply
$X \otimes Y \in \mathcal{C}_{AB}$. Then {\rm (H1)} and {\rm (H3)} imply
{\rm (H4)}.
\end{proposition}

\begin{proof}
Each term of $\xieff(E) = \sum_{\vec g} \alpha_{\vec g}(E) \otimes
\bigotimes_{w} P_{g_{w}}$ has its first factor in $\mathcal{C}$ by (H3) and its
second in the frame's cone by (H1), so tensor-closure puts the term in
$\mathcal{C}$ and convexity puts the sum there. Replacing $E$ by $\Id - E$
gives the upper bound. The same argument applies to the encoded operations
sector by sector.
\end{proof}

Lemmas~\ref{lem:positivity}(ii) and~\ref{lem:encops}(ii) are the case
$\mathcal{C} = \mathcal{K}$ of this, verified directly because the manuscript
reaches them before the general statement is available. The proposition also
says where the implication genuinely fails, which is when the effect set is not
of that shape or the cone is not closed under tensor products.

\subsection{Sectorial closure suffices}
\label{sec:sectorial:crit}

\begin{definition}[Sectorial closure]
\label{def:sectorial}
The effect set and operation set of $T$ are \emph{sectorially closed} under $G$
if they are invariant under $\alpha_{\vec g}$ for every tuple $\vec g$ and
every number of wires, that is, if (H3) holds.
\end{definition}

\begin{theorem}
\label{thm:sectorial}
Let $G$ act on $T$ by symmetries and suppose {\rm (H1)--(H4)}. Then
$\corr_{T} = \corr_{sT}$ in every causal structure, subject to the convention
on enlarged wires of Theorem~\ref{thm:main}.
\end{theorem}

\begin{proof}
The argument of Sec.~\ref{sec:main:thm} uses only the following. That the
encoded source is a state of $sT$: it is invariant by construction, being an
average over $G$, and it is a state of $T$ by (H1) and (H2). That the encoded
effects and operations are admissible in $sT$: invariance again holds by
construction, membership in $T$ on a single system by (H3), and membership on
the enlarged systems by (H4). That the frames are inherited rather than issued
afresh, which is the content of Sec.~\ref{sec:enc:frames} and depends only on
the group structure. And that the statistics agree, which follows from the
orthogonality of the frame basis together with the identity
$\langle \alpha_{g}(E), \alpha_{g}(\rho) \rangle = \langle E, \rho \rangle$,
valid because $\alpha_{g}$ is a symmetry, exactly as in the proof of
Proposition~\ref{prop:pairing}. The converse inclusion is again immediate.
\end{proof}

\subsection{Weak nonphysicality and the maximal theory}
\label{sec:sectorial:emax}

\begin{table}[t]
\caption{\label{tab:classification}Where a gap can live. The rows are the
trichotomy of Sec.~\ref{sec:prelim:trichotomy}, which by
Proposition~\ref{prop:emax} is the question of whether $\effset$ lies inside
$\effmax$. The columns are sectorial closure,
Definition~\ref{def:sectorial}. The upper left entry is impossible rather than
unknown. Sectorial closure puts $\effset$ inside $\effmax$ by
Corollary~\ref{cor:emaxclosed}, and it also puts
$\alpha_{\vec g}\Lambda\alpha_{\vec g}^{-1}$ in $\opset$ for every $\Lambda$
and every $\vec g$, so both hypotheses of Proposition~\ref{prop:emax} hold
and the symmetry is at worst weakly nonphysical.}
\centering
\footnotesize
\begin{tabular}{p{0.20\columnwidth}p{0.27\columnwidth}p{0.27\columnwidth}}
\hline\hline
 & sectorially closed & not closed \\ \hline
 strongly nonphysical
   & impossible \par Cor.~\ref{cor:emaxclosed}
   & gap conjectured \par Ref.~\cite{Ying2025foil} \\[3pt]
 weakly nonphysical
   & no gap \par Thm.~\ref{thm:sectorial}
   & open \par Sec.~\ref{sec:sectorial:not} \\
\hline\hline
\end{tabular}
\end{table}

Sectorial closure is a condition on the theory, not on the symmetry, and it is
natural to ask how it relates to the trichotomy. The comparison is cleanest in
terms of the largest effect set compatible with a given state space. Write
$\stateset^{\vee}$ for the set of all operators pairing admissibly with every
state of $T$, that is, the effect set $T$ would have if it satisfied the
no-restriction hypothesis, and put
\begin{equation}
  \label{eq:emaxdef}
  \effmax \;:=\; \bigl\{\, E \;:\; \alpha_{\vec g}(E) \in \stateset^{\vee}
  \ \ \forall \vec g \,\bigr\} ,
\end{equation}
the tuples $\vec g$ ranging over every number of wires.

\begin{proposition}
\label{prop:emax}
Suppose $\effset(T) \subseteq \effmax$ and that
$\alpha_{\vec g} \Lambda \alpha_{\vec g}^{-1}$ is completely positive for every
$\Lambda \in \opset(T)$ and every $\vec g$. Then $G$ is at worst weakly
nonphysical. Conversely, if $G$ is weakly nonphysical then
$\effset(T) \subseteq \effmax$.
\end{proposition}

\begin{proof}
For the converse direction, suppose $E \in \effset(T)$ with
$\alpha_{\vec g}(E) \notin \stateset^{\vee}$ for some $\vec g$. Then some state
$\omega$ of $T$ satisfies
$\langle \alpha_{\vec g}(E), \omega \rangle \notin [0,1]$, and the circuit that
prepares $\omega$, applies $\alpha_{\vec g}$ and measures $E$ returns that
number. So $G$ is strongly nonphysical.

For the forward direction, push the inserted symmetries onto the effect as in
the proof of Proposition~\ref{prop:weaklynonphysical}. A factor $\Lambda^{*}$
maps $\effmax$ into itself, since for any $\vec g$
\begin{equation}
  \alpha_{\vec g}\bigl( \Lambda^{*}(E) \bigr)
  \;=\; \bigl( \alpha_{\vec g} \Lambda \alpha_{\vec g}^{-1} \bigr)^{*}
        \bigl( \alpha_{\vec g}(E) \bigr) ,
\end{equation}
and the map $\alpha_{\vec g}\Lambda\alpha_{\vec g}^{-1}$ is completely positive
by hypothesis, so its adjoint is positive and unital and preserves
$\stateset^{\vee}$. Note that this is all the proof uses. It does not need
that map to belong to $\opset(T)$, which is what sectorial closure of the
operations would supply, and for $\Th$ the weaker hypothesis is the
condition~\eqref{eq:defops} defining the completely PPT-preserving class. A factor $\alpha_{\vec h}$ maps $\effmax$
into itself because $\alpha_{\vec g}\alpha_{\vec h} = \alpha_{\vec g \vec h}$.
Hence the pushed effect lies in $\effmax$ and pairs admissibly with every
state.
\end{proof}

The two halves of the proposition are not symmetric, and the asymmetry is the
point. Weak nonphysicality is a constraint on $\effset(T)$ from above and is
inherited by every subset. Sectorial closure is an invariance and is not. A
theory can therefore be weakly nonphysical without being sectorially closed,
and Theorem~\ref{thm:sectorial} does not apply to it.

The converse direction of the proposition should also be read with its
hypothesis in view. Containment in $\effmax$ says only that no single
application of the symmetry to an effect leaves $\stateset^{\vee}$. A general
witness is a word in the adjoints of operations and the maps
$\alpha_{\vec g}$, and although each $\alpha_{\vec g}$ preserves $\effmax$, an
adjoint $\Lambda^{*}$ need not, absent the hypothesis on
$\alpha_{\vec g}\Lambda\alpha_{\vec g}^{-1}$. Containment alone therefore does
not establish weak nonphysicality, a point we return to in
Sec.~\ref{sec:sectorial:not}.

What is true unconditionally is the following.

\begin{corollary}
\label{cor:emaxclosed}
$\effmax$ is the largest sectorially closed subset of $\stateset^{\vee}$.
\end{corollary}

\begin{proof}
It is sectorially closed, since $\alpha_{\vec h}(\effmax) = \effmax$ follows
from $\alpha_{\vec g}\alpha_{\vec h} = \alpha_{\vec g\vec h}$ and the fact that
$\vec g \mapsto \vec g \vec h$ permutes the tuples. Conversely, if
$\effset \subseteq \stateset^{\vee}$ is sectorially closed then for
$E \in \effset$ and any $\vec g$ one has
$\alpha_{\vec g}(E) \in \effset \subseteq \stateset^{\vee}$, which
is~\eqref{eq:emaxdef}.
\end{proof}

The same principle fixes the other line of Definition~\ref{def:T}.

\begin{corollary}
\label{cor:opsmaximal}
Let $\opset'$ be any operation class making $(\stateset, \effmax, \opset')$ a
valid generalized probabilistic theory. If $\opset'$ is sectorially closed then
$\opset'$ is contained in the completely PPT-preserving class. That class is
therefore the largest sectorially closed operation class compatible with
unrestricted states and the effect set $\effmax$.
\end{corollary}

\begin{proof}
Validity requires $\Lambda^{*}(\effmax) \subseteq \effmax$ for every
$\Lambda \in \opset'$, and every element of $\opset'$ is in particular
completely positive. Sectorial closure puts $\PT{V}\Lambda\PT{V}$ in $\opset'$
for every $V$, hence completely positive, which is~\eqref{eq:defops}.
\end{proof}

The mechanism is elementary. Complete
positivity is not invariant under $\Lambda \mapsto \PT{V}\Lambda\PT{V}$, whereas
preservation of $\mathcal{K}^{*}$ is, since
$\PT{V}(\mathcal{K}^{*}) = \mathcal{K}^{*}$ by
Eq.~\eqref{eq:dualcone}. Intersecting over sectors therefore imposes the
non-invariant condition in every sector, and that intersection is exactly the
completely PPT-preserving class. This is the operational counterpart of
$\effmax = \bigcap_{S} \PT{S}([0,\Id])$, so a single principle fixes both lines
of Definition~\ref{def:T}.

Whenever $G$ is weakly nonphysical, and the theory supplies the exact flag of
(H1), the correlated frame state of (H2) and an operation class meeting (H3)
and (H4), there is therefore a gap-free theory with the same states and the
same symmetry, namely the one whose effect set is $\effmax$. For $G = \Z_{2}$
on unrestricted states the first two hold by
Proposition~\ref{prop:h1free}, and Corollary~\ref{cor:opsmaximal} supplies the
operation class. For a general finite group the operation class has to be
supplied separately. The classification of
Proposition~\ref{prop:emax} becomes self-describing. A symmetry is strongly nonphysical exactly when the effect set is contained in no sectorially closed
effect set at all.

Weak nonphysicality thus guarantees that the \emph{maximal} choice of effect
set is gap-free. It does not guarantee that every admissible choice is. The
theory $\Th$ of Sec.~\ref{sec:theory} is the maximal choice. Indeed, for
$G = \Z_{2}$ acting by conjugation on unrestricted states, one has
$\stateset^{\vee} = [0,\Id]$ and $\alpha_{\vec g} = \PT{S}$, so that
Eq.~\eqref{eq:emaxdef} reads
$\effmax = \{ E : \ptr{E}{S} \in [0,\Id] \ \forall S \}$, which
is~\eqref{eq:defeffects} verbatim. This is the third of the demands that fix
Definition~\ref{def:T}, after validity as a GPT and the positivity lemma, and it is the one that explains the other two. The effect set is the largest one the
symmetry permits, and the largest one is automatically the one the criterion
applies to.

All four hypotheses hold for $\Th$. For (H1), the frame qubit carries the
regular representation of $\Z_{2}$, conjugation exchanging $\ket{y_{+}}$ and
$\ket{y_{-}}$, and its effects are unrestricted because a single wire carries
no restriction, by Sec.~\ref{sec:theory:tomo}. For (H2), the perfectly
correlated frame state was shown to be a state of $\sTh$ in
Sec.~\ref{sec:enc:constr}. For (H3), the effect set is invariant under every
$\PT{S}$ by the discussion following Eq.~\eqref{eq:effrestated}, and the
operation set is invariant under $\Lambda \mapsto \PT{S}\Lambda\PT{S}$
because
\begin{equation}
  \PT{V}\bigl( \PT{S}\Lambda\PT{S} \bigr)\PT{V}
  \;=\; \PT{S \triangle V} \Lambda \PT{S \triangle V} ,
\end{equation}
which is completely positive for every $V$ whenever $\Lambda \in \opset$,
since $S \triangle V$ is again a subset. Note that this is more than the
invariance of $\opset$ under the global conjugation checked in
Sec.~\ref{sec:conj:notcp}, which is the case $S = W$. And (H4) is
Lemmas~\ref{lem:positivity}(ii) and~\ref{lem:encops}(ii).
Theorem~\ref{thm:main} is therefore the case $G = \Z_{2}$ of
Theorem~\ref{thm:sectorial}.

\subsection{The criterion is tight on the effects}
\label{sec:sectorial:tight}

Theorem~\ref{thm:sectorial} asks for sectorial closure of the effect set and of
the operation set, and one should ask whether either half can be weakened. The
answer is that the effect half cannot, and this subsection proves it. All the
slack in the hypothesis is in the operations, where
Sec.~\ref{sec:main:scope} has already shown there is some.

Begin by asking what the statistics actually require. Consider a party holding
wires from two sources and write its encoded effect as
$\sum_{s_{1}s_{2}} X_{s_{1}s_{2}} \otimes P_{s_{1}} \otimes P_{s_{2}}$ with the
blocks free. Reality forces $X_{-\vec s} = X_{\vec s}^{\mathsf T}$, leaving two
independent operators $M := X_{++}$ and $N := X_{-+}$. Collecting the four
sectors as in the proof of Proposition~\ref{prop:pairing}, the sectors $(+,+)$
and $(-,-)$ both contribute $M$, while $(-,+)$ and $(+,-)$ both contribute
$\PT{}(N)$, the second because $\PT{2}(X^{\mathsf T}) = \PT{1}(X)$ for any $X$. The
statistics are therefore reproduced as soon as
\begin{equation}
  \label{eq:scprime}
  E \;=\; \tfrac12 \bigl( M + \PT{}(N) \bigr) ,
\end{equation}
and the proof of Theorem~\ref{thm:sectorial} takes the particular solution
$M = E$, $N = \PT{}(E)$, which is what invokes the invariance.

\begin{definition}
\label{def:scprime}
An effect set satisfies {\rm (SC$'$)} if, for every bipartition of the wires
held by a party, every POVM $\{E_{b}\}$ in that set admits POVMs $\{M^{(b)}\}$
and $\{N^{(b)}\}$, again in that set, with
$E_{b} = \tfrac12\bigl(M^{(b)} + \PT{}(N^{(b)})\bigr)$ for each $b$.
Quantifying over all bipartitions rather than over the source-aligned ones is
no strengthening. A causal structure may consist of $n$ single-wire sources all
feeding one central party, and as the source signs range over their patterns
the transposed set ranges over every subset of that party's wires, so every
bipartition is source-aligned in some causal structure.
\end{definition}

Sectorial closure implies (SC$'$), by the solution just named. The point of
this subsection is that nothing is gained.

\begin{theorem}
\label{thm:scprime}
Let $T$ have unrestricted states, so that $\effset \subseteq [0,\Id]$. If
$\effset$ satisfies {\rm (SC$'$)} then $\effset \subseteq \effmax$. If
$\effset$ satisfies {\rm (SC$'$)} and is in addition convex and closed, then
$\effset$ is sectorially closed.
\end{theorem}

\begin{proof}
Fix a bipartition and write $\PT{}$ for the corresponding partial
transposition. Apply (SC$'$) to the two-outcome measurement $\{E, \Id - E\}$,
giving $E = \tfrac12(M_{1} + \PT{}N_{1})$ with $\{M_{1}, \Id - M_{1}\}$ and
$\{N_{1}, \Id - N_{1}\}$ measurements in $\effset$. Applying $\PT{}$, which is
an involution,
\begin{equation}
  \PT{}E \;=\; \tfrac12 \PT{}M_{1} \;+\; \tfrac12 N_{1} .
\end{equation}
Here $N_{1}$ is controlled, lying in $[0,\Id]$, and $\PT{}M_{1}$ is not. But
$\{M_{1}, \Id - M_{1}\}$ is itself a measurement in $\effset$, so (SC$'$)
applies to it in turn, and iterating $n$ times gives
\begin{equation}
  \label{eq:scseries}
  \PT{}E \;=\; \sum_{j=1}^{n} 2^{-j} N_{j} \;+\; 2^{-n}\, \PT{}M_{n} ,
\end{equation}
with every $N_{j}$ and $M_{n}$ in $[0,\Id]$.

The recursion is well founded, because {\rm (SC$'$)} requires each $M_{j}$ to
lie in $\effset$ and not merely in the order interval, so the hypothesis
applies to it again. It does not terminate, and it does not need to. Because
$\PT{}$ is a linear map on a finite-dimensional space, the quantity
\begin{equation}
  \label{eq:Cconst}
  \kappa \;:=\; \sup \bigl\{\, \lVert \PT{}X \rVert_{\infty} \;:\; 0 \le X \le \Id
  \,\bigr\}
\end{equation}
is finite. It is a supremum over the whole order interval, not over the
particular operators appearing in~\eqref{eq:scseries}, so it is a single
constant depending on the dimension and the bipartition but not on the depth
$n$. Since the sum in~\eqref{eq:scseries} has
least eigenvalue at least $0$ and greatest eigenvalue at most $1 - 2^{-n}$,
\begin{equation}
\begin{split}
  -2^{-n}\kappa &\;\le\; \lambda_{\min}\bigl( \PT{}E \bigr) , \\
  \lambda_{\max}\bigl( \PT{}E \bigr) &\;\le\; 1 - 2^{-n} + 2^{-n}\kappa ,
\end{split}
\end{equation}
and these hold for every $n$ while $\PT{}E$ is a fixed operator. Nothing is
required to converge here. The two inequalities are a family of bounds on one
operator, indexed by $n$, and taking the infimum and the supremum over $n$
gives $0 \le \PT{}E \le \Id$ directly. Every step here is an identity or an
inequality between operators, and no causal structure enters. In particular the
bound $\effset \subseteq [0,\Id]$ is read off from the theory having
unrestricted states, and not from pairing $\PT{}E$ against the states that some
scenario can prepare, which would be strictly weaker and is precisely the slack
that Proposition~\ref{prop:dualcone} is about. The argument used nothing about
which bipartition was fixed at the start, so it applies to each of them
separately with its own constant, and by the remark following
Definition~\ref{def:scprime} that is every subset $S$ of the wires. Hence
$E \in \effmax$.

The second statement does use a limit, and an operator limit at that. Write
$A_{n}$ for the sum in~\eqref{eq:scseries}, whose coefficients are positive
with total weight $1 - 2^{-n}$, and normalise,
\begin{equation}
  \label{eq:normalisedpartial}
  \frac{A_{n}}{1 - 2^{-n}}
  \;=\; \sum_{j=1}^{n} \frac{2^{-j}}{1-2^{-n}} \, N_{j} ,
\end{equation}
a convex combination of elements of $\effset$ and therefore an element of
$\effset$ whenever $\effset$ is convex. Since
$\lVert \PT{}E - A_{n} \rVert_{\infty} = 2^{-n} \lVert \PT{}M_{n}
\rVert_{\infty} \le 2^{-n}\kappa$ and $(1-2^{-n})^{-1} \to 1$, the normalised
partial sums converge in norm to $\PT{}E$. Closedness of $\effset$ then puts
the limit in $\effset$.
\end{proof}

Unlike Theorems~\ref{thm:main} and~\ref{thm:sectorial}, which are algebraic
and hold in any dimension, this one does not survive the passage to infinite
dimensions, because $\kappa$ diverges there. The constant plays no role beyond
being finite. For the record, on a
$d \otimes d$ system it equals $(1+d)/2$, attained at the projector onto the
symmetric subspace, whose partial transpose is
$(\Id + d\proj{\phi^{+}})/2$. That value is an exact computation rather than a search. Random projectors do
not come close to the supremum in dimension three
or above, so the numerics in Appendix~\ref{app:numerics} corroborate only that
the bound is not exceeded.

The encoding cannot be improved on the effect side. Whatever freedom
Eq.~\eqref{eq:scprime} appears to offer is illusory once the requirement is
imposed on every measurement of the theory rather than on one, because the
$M^{(b)}$ are themselves effects and the condition reapplies to them. The
degenerate solution $N^{(0)} = 0$, $N^{(1)} = \Id$ illustrates this. It gives
$M_{n} = 2^{n}E$, and demanding that these stay in $[0,\Id]$ for all $n$ forces
$E = 0$. That branch is the extreme case of~\eqref{eq:scseries} in which the
convex series is empty and the whole of $\PT{}E$ is carried by the vanishing
tail.

Sectorial closure is therefore not an arbitrary sufficient condition but
exactly what the encoding demands of the effects, and
Corollary~\ref{cor:emaxclosed} may be read again in that light, since $\effmax$ is the largest effect set the construction of Sec.~\ref{sec:enc} can reach.

Finally, this closes off a route to a counterexample to the conjecture of
Ref.~\cite{Ying2025foil}. One might hope for a theory whose effect set is
reachable by the encoding, so that no gap arises anywhere, while lying outside
$\effmax$, so that the symmetry is strongly nonphysical by the unconditional
half of Proposition~\ref{prop:emax}. By Lemma~\ref{lem:uniformity} such a
theory could not be dismissed as one padded with an inert sector, so it would
be a counterexample surviving any restatement that excludes a gap-free core.
Theorem~\ref{thm:scprime} says there is none. Reachability by the encoding and
containment in $\effmax$ are the same condition, and the closure that a
candidate would have to lack is supplied automatically by the boundedness of
$\PT{}$ together with the dyadic weights.

\subsection{How far the encoding reaches}
\label{sec:sectorial:not}

Theorem~\ref{thm:scprime} settles the effect set of a whole theory. It says
nothing about individual measurements, and the distinction matters. A theory
that is not sectorially closed may still have many measurements that the
encoding simulates, and the question of a gap turns on the ones it does not.

Proposition~\ref{prop:emax} leaves room between its two conditions, and it is
worth asking what lives there. A theory in that region has an effect set inside
$\effmax$ that is not invariant under the partial transpositions, and by
Corollary~\ref{cor:emaxclosed} such a set is a proper subset of $\effmax$. Two
things have to be said about it, and the first is a warning.

\paragraph*{A trap.}
The obvious way to shrink $\effmax$ while keeping it closed under the global
transpose is to take the cone $K$ generated by the separable operators together
with a PPT-entangled state $F$ and its transpose, and to put
\begin{equation}
  \label{eq:shrunkeffects}
  \effset' \;=\; \bigl\{\, E \;:\; E \in K, \ \Id - E \in K \,\bigr\} ,
\end{equation}
which is closed under complementation by construction, so the objection of
Sec.~\ref{sec:theory:whynot} does not apply to it. For this to fail sectorial
closure one needs $\ptr{F}{B} \notin K$.

The states one reaches for first do not have that property, and the reason
should be on record so that it is not rediscovered. A PPT-entangled state constructed from an unextendible product
basis~\cite{Bennett1999upb,DiVincenzo2003upb} has the form
$F \propto \Id - \sum_{i} \proj{a_{i}} \otimes \proj{b_{i}}$, and if the basis
is real then each $\proj{b_{i}}$ is a real symmetric rank-one projector and is
therefore fixed by transposition. Hence $\ptr{F}{B} = F$ exactly, the cone $K$
is $\PT{B}$-invariant, and $\effset'$ is sectorially closed after all. This
covers the Tiles and Pyramid states and every other real unextendible product
basis, which is to say the standard examples. The argument is exact and needs no numerics. We have nonetheless confirmed for
the Tiles state that the two operators agree entrywise to machine precision.

What the construction needs is a range condition rather than a separability
condition, and that is checkable. Suppose $\ptr{F}{B} = R + \mu F + \nu
F^{\mathsf T}$ with $R \ge 0$ and $\mu, \nu \ge 0$. All three terms are
positive, so $\mu > 0$ would force
$\mathrm{range}(F) \subseteq \mathrm{range}(\ptr{F}{B})$, and likewise for
$\nu$. If both inclusions fail then $\mu = \nu = 0$ and $\ptr{F}{B}$ is
separable, whence $F$ is separable, the separable cone being closed under
partial transposition. So it suffices to exhibit $F$ whose range is not
contained in that of its partial transpose, and a rank inequality settles that.

The Horodecki family in $3 \times 3$~\cite{Horodecki1997bound} does it. For $a \in (0,1)$ let $\rho_{a} = H/(8a+1)$ in the basis $\ket{00},\dots,\ket{22}$, where $H$
has $a$ at every diagonal entry and $a$ at every entry of the $3 \times 3$
block indexed by $\{00,11,22\}$, including its off-diagonal entries, and where
the four entries $H_{20,20} = H_{22,22} = (1+a)/2$ and
$H_{20,22} = H_{22,20} = \sqrt{1-a^{2}}/2$ then overwrite what the previous
sentence put there. Filling the $\{00,11,22\}$ block only on its diagonal
gives rank nine instead of seven and destroys the example.
Then $\rho_{a}$ has positive partial transpose, is entangled, and satisfies
$\mathrm{rank}\,\rho_{a} = 7 > 6 = \mathrm{rank}\,\ptr{\rho_{a}}{B}$, so
neither range containment can hold. We have verified the ranks, the positivity
of the partial transpose, and entanglement through the realignment
criterion~\cite{ChenWu2003,Rudolph2005}, whose value exceeds one for every $a$
we tested. See Appendix~\ref{app:numerics}. Taking $F = \rho_{a}$ therefore gives an effect
set inside $\effmax$ that is not sectorially closed.

\paragraph*{What the candidate would and would not show.}
Suppose such an $\effset'$ were exhibited. It would show that the hypothesis of
Theorem~\ref{thm:sectorial} is not necessary for the containment
$\effset \subseteq \effmax$, which is already clear from
Proposition~\ref{prop:emax}. It would not show that the theory is weakly
nonphysical, and here we must be careful about a step that is easy to take
without noticing. Weak nonphysicality does not follow from
$\effset' \subseteq \effmax$ alone. The general witness is a word
$(\mathcal{M}_{1}^{*} \circ \cdots \circ \mathcal{M}_{m}^{*})(E)$ whose factors
are adjoints of operations and partial transpositions, and while
$\alpha_{\vec g}$ preserves $\effmax$ by Corollary~\ref{cor:emaxclosed}, an
adjoint $\Lambda^{*}$ need not, unless $\alpha_{\vec g}\Lambda\alpha_{\vec
g}^{-1}$ is again an operation of the theory. That is the second hypothesis of
Proposition~\ref{prop:emax}, and a theory specified by its effect set alone has
not supplied it. Any candidate must come with its operation set, and the
region between the two conditions accordingly contains strongly nonphysical
symmetries as well as weakly nonphysical ones.

And even a weakly nonphysical candidate would not obviously produce a gap.
Shrinking the effect set shrinks $\corr_{T}$ as well as $\corr_{sT}$, and
whether the two shrink at the same rate we do not know.

One thing can be said at once about where a gap could not come from. An effect
that is itself invariant under the symmetry lies in $\sTh$ already, so no
measurement built only from such effects distinguishes the two theories. The
exceptional ray of the candidate constructed above runs through $\rho_{a}$,
which is a real matrix and is invariant in exactly this way, so if that
candidate has a gap it must come from the interplay of those effects with
states or operations that are not invariant, and not from the effects alone.
An effect set lying above $\effmax$ need not behave the same way, and
Sec.~\ref{sec:disc} exhibits one that does not.

What Theorem~\ref{thm:scprime} does tell us about such a candidate is where its
failure must sit. The set $\effset'$ is convex and closed and is not
sectorially closed, so it cannot satisfy {\rm (SC$'$)}, and the failure must
therefore appear on some measurement. It cannot appear on the measurements the
encoding reaches, and the theory's exceptional direction is the ray through
$F$, so the failure is confined to a segment of that ray. Locating the segment
is a finite computation once an explicit $F$ is fixed, and deciding whether the
correlations available there exceed those of the symmetrized theory is the
remaining question. Whether a weakly nonphysical symmetry that is not
sectorially closed can produce a gap is left open.

\section{Discussion}
\label{sec:disc}

Three theories in this paper carry the same symmetry, and it receives a
different verdict in each. In quantum theory, complex conjugation is strongly nonphysical, and adjoining it assigns a negative
probability to the antisymmetric
outcome of a swap test, by Eq.~\eqref{eq:negativeprob}. In the PPT-world of
Ref.~\cite{Ying2025foil}, where states \emph{and} effects are both restricted
to those with positive partial transpose, it is physical, and the authors note
that every symmetry of that theory is. Symmetrization there is a twirling and
Theorem~2 of their paper applies directly. In $\Th$ it is weakly nonphysical,
by Proposition~\ref{prop:weaklynonphysical}. The map is the same map
throughout. Any quantity computed from the symmetry alone is therefore constant across a
distinction it is supposed to explain. The trichotomy classifies pairs
consisting of a theory and a symmetry. It does not classify symmetries.

The comparison says where a criterion would have to live. What separates the three cases
is not the map but the relation between the state cone and the dual of the
effect cone, which is to say the theory's restriction structure. Quantum theory
has no restriction, so $\widehat{\stateset} = \stateset$ and there is nothing
for a nonpositive map to hide behind. PPT-world restricts the states as well,
which pulls $\stateset$ down until conjugation preserves it. $\Th$ restricts
only the effects, which pushes $\widehat{\stateset}$ up until it contains the
image. Sectorial closure, Definition~\ref{def:sectorial}, is one condition of this
kind. Theorem~\ref{thm:sectorial} shows it is sufficient and
Theorem~\ref{thm:scprime} shows it cannot be weakened on the effects, while
Corollary~\ref{cor:opsmaximal} identifies the largest operation class it
allows. Whether any condition of this kind is necessary we do not know.

The consistency of all this with the original result needs stating
explicitly, because the two are easily read as being in tension. Real-amplitude
quantum theory is the swirled world of the same $\Z_{2}$, and the gap there is
real~\cite{Renou2021}. It is real because quantum theory admits effects with negative partial transpose. The Bell-state measurement is
available, the
encoding of Sec.~\ref{sec:enc} produces an operator that is not positive, and
the correlations genuinely separate. Our construction leaves the symmetry's nonphysicality untouched and removes
the measurements that would register it. In this sense the robustness of the real-amplitude result
rests on a property of quantum theory that has nothing to do with the symmetry
being antiunitary, namely that quantum theory admits every effect its states
permit.

Examples of this kind are not rare, and it would be a mistake to read $\Th$ as
a curiosity. By Corollary~\ref{cor:emaxclosed}, whenever a symmetry is weakly
nonphysical on a theory the maximal effect set compatible with it is
sectorially closed, and gap-free wherever the remaining hypotheses of
Theorem~\ref{thm:sectorial} hold. Every such symmetry comes with a companion
theory of that kind, of which $\Th$ is the instance for $\Z_{2}$ acting on
unrestricted quantum states. What distinguishes $\Th$ from the rest of that family is that its effect set
has an independent characterisation, as the completely PPT-preserving structure
already familiar from entanglement theory, and that it is small enough to
compute with.

Nor are theories with a restricted effect set exotic. A superselection rule is
exactly such a restriction, and so is the LOCC constraint between two distant
laboratories, and so is the Gaussian fragment of quantum optics. In each case
the criterion of Sec.~\ref{sec:sectorial} poses a definite and checkable
question: is the free effect set invariant under the symmetry acting
independently on each source? Three familiar fragments answer no, and for the
same reason, which is worth naming because it says what a theory would have to
give up.

Stabilizer theory retains the singlet projector $P_{\rm asym}$ as a free
effect, the singlet being a stabilizer state, and
Eq.~\eqref{eq:negativeprob} then returns $-1/2$ when it is paired with the
partial conjugation of a maximally entangled state. Gaussian quantum optics
retains the projector onto a two-mode squeezed vacuum. Conjugation in the
position basis reverses momenta, so partial conjugation reverses the momenta of
one mode, which is Simon's criterion, and that projector has a partial
transpose that is Gaussian but not positive. In both fragments the free
effects therefore lie outside $\effmax$, conjugation is strongly nonphysical,
and Theorem~\ref{thm:sectorial} says nothing. The criterion therefore
separates cases rather than reporting that every restriction hides a
symmetry.

Superselection rules answer no as well, and by the same route. A global
particle-number rule leaves the singlet projector free, since the singlet lies
in a single charge sector, and a local parity rule on two qubits per party
leaves free the projector onto
$(\ket{00}_{A}\ket{00}_{B} + \ket{11}_{A}\ket{11}_{B})/\sqrt2$, which commutes
with both local parities. Each has a partial transpose with least eigenvalue
$-1/2$.

What these failures have in common is that the fragment can still detect
entanglement across a cut. A theory in which conjugation hides must give that
up. Its effects must all have positive partial transpose, which by the
discussion after Eq.~\eqref{eq:effrestated} and closure under complementation
is the same as lying inside $\effmax$, so the search reduces to finding
partial-transpose-invariant convex subcones of the positive-partial-transpose
cone. That is not a contrived demand but the defining feature of the class
Rains introduced for LOCC.

Such subcones exist below $\effmax$, so the family is genuinely a family.
Taking the separable cone in place of the completely PPT-preserving one gives a
second theory: states unrestricted, effects
$\mathrm{SEP} \cap (\Id - \mathrm{SEP})$, operations the separable maps. The
intersection is again forced, for the reason just given. Every product effect
survives, by $\Id - E_{A} \otimes E_{B} = (\Id - E_{A}) \otimes E_{B} +
\Id \otimes (\Id - E_{B})$, so the theory is tomographically local and the
bilocality protocol of Appendix~\ref{app:bilocal} runs in it word for word. It
is strictly smaller than $\Th$, since a suitably scaled PPT-entangled state is
an effect of $\Th$ and not of it. Between the two lies the hierarchy of symmetric extensions with positive
partial transpose~\cite{DPS2004}. For each fixed level $k$ the corresponding
cone $\mathcal{C}_{k}$ is convex, invariant under the partial transpositions
and closed under tensor products. Hypotheses (H1) and (H2) hold because the
states are unrestricted and the group is $\Z_{2}$, (H3) holds by that
invariance, and (H4) follows from Proposition~\ref{prop:h4free}, so the theory
it defines is gap-free in every causal structure. The first level is $\Th$ and the limit is the separable theory, and the first
level already differs from the limit, since the positive-partial-transpose cone
differs from the separable one. The word ``fixed'' carries weight here. A theory
whose level grew with the size of the system would have a cone that is not
tensor-closed, and the proposition would not apply.

The price of descending the hierarchy is cost. Membership at every fixed level
remains a semidefinite condition, of a size that grows with $k$, and only the
separable limit leaves that class altogether, which is what makes $\Th$ the
cheapest member of the family.

That structure has been studied for a quarter of a century, but not in this
capacity. It is a relaxation of LOCC introduced to bound distillation rates,
and Appendix~\ref{app:cpptp} gives the history. The machinery that does address restricted fragments is the theory of cone
equivalence~\cite{Selby2023accessible}. It is developed for prepare-measure
scenarios with no parallel composition, and so does not reach the questions
asked here. In any case the effect cone of $\Th$ is not the full positive cone,
so $\Th$ is not cone-equivalent to quantum theory. Finally, $\Th$ is tomographically local, yet its effect set on a composite is
not determined by the effect sets of the parts, and it may be worth adding to
the catalogue of theories in which entanglement splits into inequivalent
kinds~\cite{Baldijao2026tnl}.
The complementary line of work, in which the departure from real-valuedness is
treated as a quantifiable resource rather than as a structural
question~\cite{HickeyGour2018,WuImaginarity2021}, is orthogonal to what is
studied here but concerns the same asymmetry.

Four questions are left open, and they are not four unrelated ones. By
Corollaries~\ref{cor:emaxclosed} and~\ref{cor:opsmaximal} the theory $\Th$ sits
at a maximum in two independent coordinates, its effects and its operations, so
the way to ask what remains is to move off that maximum and see what breaks.
Three of the questions move along one coordinate or the other, one along the
operations and two along the effects, downward and upward. The fourth asks
whether the criterion survives outside the class of theories for which Ying
\emph{et al.} formulate it.

Take the operations first. Section~\ref{sec:main:scope} shows that the
completely PPT-preserving class is not the largest one compatible with the
effects, while Corollary~\ref{cor:opsmaximal} shows that it is the largest
sectorially closed one, and between the two lies the region that
Eq.~\eqref{eq:opschain} records without closing. Where the weakening stops we
do not know. Deciding whether any theory in that region exhibits a gap of its
own needs an impossibility argument of the kind Renou \emph{et al.} supply for
real-amplitude quantum theory~\cite{Renou2021}, and by
Sec.~\ref{sec:main:necessary} such an argument cannot live in the bilocality
scenario. It has to be mounted where some party correlates what it keeps with
what it sends on, which is the one shape Proposition~\ref{prop:defer} leaves
standing.

Moving the effects downward gives the second question. An effect set strictly
inside $\effmax$ need not be invariant under the partial transpositions, so
Theorem~\ref{thm:sectorial} does not reach it, and
Sec.~\ref{sec:sectorial:not} builds a candidate of that kind from the Horodecki
family. We do not decide it. What Theorem~\ref{thm:scprime} supplies is a
location rather than an answer. The candidate is convex and closed and is not
sectorially closed, so it fails (SC$'$) on some measurement, and that failure
sits on a segment of the ray through $F$. Whether a weakly nonphysical symmetry
that is not sectorially closed can produce a gap therefore reduces to a
question about one ray, though shrinking the effect set shrinks $\corr_{T}$
along with $\corr_{sT}$, and we do not know which of the two shrinks faster.

Moving the effects upward is the third question, and it leaves the reach of
every theorem here at once, since above $\effmax$ the symmetry is strongly
nonphysical by Proposition~\ref{prop:emax}. A natural place to start is the
effect set $\effset_{1}$ defined by imposing Eq.~\eqref{eq:effrestated} on
the empty set and the single wires only. In the language of
Sec.~\ref{sec:theory:def} it imposes the
condition on a generating set of $\Z_{2}^{\,n}$ rather than on the group, and
the two differ because the condition $0 \le \cdot \le \Id$ is not preserved
under composition. It is convex, closed under complementation and stable under
tensor products. On two and on three wires it coincides with $\effmax$: every
subset is then a singleton, a complement of one, the empty set or everything,
and the condition for $S^{c}$ follows from that for $S$ because
$\PT{S^{c}} = \PT{W}\PT{S}$ with $\PT{W}$ preserving the order interval. On
four wires the $2$:$2$ subsets are complements of each other and of nothing
smaller, so the collapse fails, $\effset_{1}$ is strictly larger, and the
symmetry on it is strongly nonphysical.

The bilocality scenario is therefore blind to it by construction, and the
smallest structure that could see it is a four-wire star. Even there
Proposition~\ref{prop:noswap} constrains what a gap could look like. The
defining condition of $\effset_{1}$ is exactly $\ptr{E}{\{i\}} \ge 0$ on each
single wire, which is the hypothesis of the proposition for every singleton
$S$, so every assemblage $\effset_{1}$ can prepare in the four-wire star, on
every outcome, is positive across every $1$:$3$ cut. Across $2$:$2$ cuts it
need not be. A gap for $\effset_{1}$ would therefore have to come from
four-party correlations whose entanglement is confined entirely to the $2$:$2$
cuts, with no outer party holding anything negative under partial transposition
against the other three. We do not know of a mechanism of that shape.

How far $\effset_{1}$ exceeds $\effmax$ can be said exactly, and in a
scale-invariant way. Appendix~\ref{app:e1} does so, and records two elements of
$\effset_{1} \setminus \effmax$ in closed form. The second of them matters
here. An effect invariant under the symmetry belongs to $\sTh$ already and
could never on its own distinguish the two theories, and one might hope that
the whole of the extra room is invariant in that way. It is not. The witness of
Eq.~\eqref{eq:imagwitness} has invariant part exactly $\Id/2$, so no argument
that treats the extra room of $\effset_{1}$ as inert can settle whether it
produces a gap.

The fourth question steps outside the family. Is fermionic quantum theory
sectorially closed under the parity symmetry? It too fails tomographic
locality~\cite{DAriano2014fermionic}, and unlike the theories considered here
it is not a foil, so an answer would say something about a theory people take
seriously rather than about a constructed one. The question is not immediate.
Ying \emph{et al.} formulate the trichotomy of
Sec.~\ref{sec:prelim:trichotomy} for tomographically local theories, and it
would have to be extended before the question can be posed at all.

It remains to say how all this sits with the conjecture of
Ref.~\cite{Ying2025foil} that a strongly nonphysical symmetry always produces a
gap somewhere for a nonclassical theory. Our construction concerns the weakly nonphysical case and does
not touch that conjecture. The conjecture as stated would, however, appear to admit degenerate
counterexamples of a kind our own family cannot supply, and the reason it
cannot is Lemma~\ref{lem:uniformity}.

The consequence is that a theory with unrestricted states cannot carry a
symmetry acting by conjugation on some wires and trivially on others. It has no
inert sector, and in particular it cannot be padded with one. A theory whose states \emph{are} restricted has no such protection. If its
correlations were already saturated by an inert sector, symmetrization would
cost nothing while the symmetry remained strongly nonphysical. We do not
exhibit such a theory here and make no claim that one exists, but the
possibility is not excluded by anything in the conjecture as stated, and it is
disjoint from the mechanism studied in this paper. Any repaired form of the
conjecture that rules it out is untouched by our results.

The division also gives the title of this paper its sharper reading.
Restricting the effects hides a nonphysical symmetry by removing the
instruments that would register it, and Theorem~\ref{thm:main} says it does so
in every causal structure at once. Restricting the states hides one for an
unrelated reason, by making the sector it acts on redundant. Only the first of
the two mechanisms says anything about what a symmetry is.

\section*{Acknowledgements}

C.-F.K. is grateful to the European Union and the Region R\'eunion, France
(POE FEDER 2021--2027, n$^\circ$2025-0954-007180) for the funding support.

The author used large language models for language editing, and as an
additional check on the numerical verification recorded in
Appendix~\ref{app:numerics}. The results are the author's own, and the author
takes full responsibility for the content of this paper.

\appendix

\section{Proof that \texorpdfstring{$\Th$}{T} is a valid GPT}
\label{app:validgpt}

We prove the parts of Proposition~\ref{prop:validgpt} not established in the
body, namely (a)--(e) and (g). The second half of part (f) was proved in
Sec.~\ref{sec:theory:valid}. Throughout, $S$ and $V$ denote subsets of wires and $\PT{S}$ the corresponding
partial transposition. It is linear, self-adjoint for the Hilbert--Schmidt pairing, involutive, trace
preserving and unital, and it satisfies
$\PT{S}\PT{V} = \PT{S \triangle V}$.
None of these arguments can be borrowed from the corresponding proof for a
twirled world~\cite{Centeno2025twirled}. There the subtheory is carved out by an
invariance, which is a linear constraint, so convexity, admissible
probabilities and the existence of complements descend from the parent theory
without work. Here $\Th$ is specified by a family of cone restrictions and
each item has to be checked.

\paragraph*{(a) Convexity, closedness, and admissible probabilities.}
By~\eqref{eq:effrestated},
\begin{equation}
  \effset \;=\; \bigcap_{S} \PT{S}^{-1}\bigl( [0,\Id] \bigr) ,
\end{equation}
an intersection of preimages of a closed convex set under linear maps, hence
closed and convex. The state set is the set of density matrices, closed and
convex. Taking $S = \emptyset$ gives $0 \le E \le \Id$ for every
$E \in \effset$, so $\Tr[E\rho] \in [0,1]$ for every $\rho \in \stateset$.

\paragraph*{(b) Unit, zero and complements.}
$\PT{S}(0) = 0$ and $\PT{S}(\Id) = \Id$ for every $S$, so $0,\Id \in \effset$.
Closure under $E \mapsto \Id - E$ was established in
Sec.~\ref{sec:theory:whynot}: $\ptr{(\Id-E)}{S} = \Id - \ptr{E}{S}$, and the
family of conditions $0 \le \ptr{E}{S} \le \Id$ is symmetric under exchanging
the two bounds.

\paragraph*{(c) Tomographic completeness.}
Distinct states are separated by effects because the products of quantum
effects lie in $\effset$, by Sec.~\ref{sec:theory:tomo}, and span the Hermitian
operators on the composite. Distinct effects are separated by states because
the density matrices span the Hermitian operators.

\paragraph*{(d) Closure under steering.}
Let $E \in \effset$ act on the wires $A \cup \{C\}$ and let $\rho_{C}$ be a
state of the wire $C$. Put
\begin{equation}
  \label{eq:steeredeffect}
  F \;=\; \Tr_{C}\bigl[ E \,(\Id_{A} \otimes \rho_{C}) \bigr] .
\end{equation}
We claim $F \in \effset$ on the wires $A$.

The map $E \mapsto F$ commutes with $\PT{S}$ for every $S \subseteq A$, since
the two act on disjoint tensor factors. Explicitly, write
$F_{ij} = \sum_{k,l} E_{(ik),(jl)} (\rho_{C})_{lk}$ and transpose the indices of
$i$ and $j$ lying in $S$. Either order of operations gives the same
result. Hence
\begin{equation}
  \label{eq:steercommute}
  \ptr{F}{S} \;=\; \Tr_{C}\bigl[ \ptr{E}{S}\,(\Id \otimes \rho_{C}) \bigr] .
\end{equation}
Diagonalise $\rho_{C} = \sum_{m} p_{m} \proj{m}$ with $p_{m} \ge 0$. Then
\begin{equation}
  \Tr_{C}\bigl[ X (\Id \otimes \rho_{C}) \bigr]
  \;=\; \sum_{m} p_{m} \, \bra{m} X \ket{m}_{C} ,
\end{equation}
and each $\bra{m}X\ket{m}_{C}$ is a diagonal block of $X$, hence positive
semidefinite whenever $X$ is. Applying this to $X = \ptr{E}{S} \ge 0$ gives
$\ptr{F}{S} \ge 0$, and applying it to $X = \Id - \ptr{E}{S} \ge 0$, together
with $\Tr_{C}[\Id(\Id \otimes \rho_{C})] = \Id_{A}$, gives
$\ptr{F}{S} \le \Id_{A}$. As $S \subseteq A$ was arbitrary, $F \in \effset$.

The state side is immediate because $\stateset$ is unrestricted. For a state
$\rho$ on $A \cup \{C\}$ and an effect $e$ on $C$, the same block argument
applied to $e = \sum_{m} q_{m}\proj{m}$ with $q_{m} \in [0,1]$ shows that
$\Tr_{C}[\rho\,(\Id \otimes e)]$ is positive semidefinite with trace at most
one, hence a subnormalised state.

\paragraph*{(e) The operation set.}
The identity satisfies $\PT{V}\,\id\,\PT{V} = \id$ for every $V$ and so lies in
$\opset$. Convexity holds because
$\PT{V}(\sum_{i} p_{i}\Lambda_{i})\PT{V} = \sum_{i} p_{i}\PT{V}\Lambda_{i}\PT{V}$
and a convex combination of completely positive maps is completely positive.

For sequential composition, let $\Lambda_{1},\Lambda_{2} \in \opset$ be
composable. Using $\PT{V}^{2} = \id$ on the intermediate wires,
\begin{equation}
  \label{eq:seqclosure}
  \PT{V} \Lambda_{2} \Lambda_{1} \PT{V}
  \;=\; \bigl( \PT{V}\Lambda_{2}\PT{V} \bigr)
        \bigl( \PT{V}\Lambda_{1}\PT{V} \bigr) ,
\end{equation}
a composition of completely positive maps. This step uses the convention of
Sec.~\ref{sec:theory:def} that the output wires of $\Lambda_{1}$ carry the same
labels as the input wires of $\Lambda_{2}$, so that the same $\PT{V}$ can be
inserted between them.

For parallel composition, let $\Lambda_{1}$ and $\Lambda_{2}$ act on disjoint
sets of wires $W_{1}$ and $W_{2}$. Partial transposition factorises across
disjoint factors, $\PT{V} = \PT{V \cap W_{1}} \otimes \PT{V \cap W_{2}}$, so
\begin{equation}
  \label{eq:parclosure}
  \PT{V} (\Lambda_{1} \otimes \Lambda_{2}) \PT{V}
  \;=\; \bigl( \PT{V_{1}}\Lambda_{1}\PT{V_{1}} \bigr) \otimes
        \bigl( \PT{V_{2}}\Lambda_{2}\PT{V_{2}} \bigr) ,
\end{equation}
with $V_{i} = V \cap W_{i}$, a tensor product of completely positive maps.
Finally every $\Lambda \in \opset$ is in particular completely positive and
trace preserving, and therefore maps density matrices to density matrices.
Since $\stateset$ is unrestricted, this is the first half of (f).

\paragraph*{(g) Effects close under tensor products.}
If $E_{A}$ and $E_{B}$ are effects of $\Th$ on disjoint sets of wires then, by
the same factorisation,
$\PT{S}(E_{A} \otimes E_{B}) = \PT{S \cap A}(E_{A}) \otimes \PT{S \cap B}(E_{B})$.
Both factors lie in $[0,\Id]$, and a tensor product of operators in $[0,\Id]$
lies in $[0,\Id]$, so $E_{A} \otimes E_{B} \in \effset$.

\paragraph*{Numerical corroboration.}
The steering and composition claims were also checked numerically, on random
instances and against a negative control. See Appendix~\ref{app:numerics}.

\section{\texorpdfstring{\cpptp{}}{C-PPT-P}: conventions and the multipartite
         generalisation}
\label{app:cpptp}

The operation set of Definition~\ref{def:T} is a class already familiar from
entanglement theory, but the conventions in the literature differ in ways that
matter here, and the multipartite version we use is not quite the one usually
stated. This appendix fixes both.

\paragraph*{The bipartite class.}
Alongside the local, separable and two-local classes, Rains introduces the
class
\begin{equation}
  \label{eq:rainsclass}
\begin{split}
  \mathcal{C}_{\Gamma} \;=\;
  &\mathrm{Op}(V \otimes W, V' \otimes W') \\
  \cap\; &\mathrm{Op}(V \otimes W, V' \otimes W')^{\Gamma_{W \otimes W'}} ,
\end{split}
\end{equation}
where $V,W$ are the input spaces of the two parties, $V',W'$ their output
spaces, and $\mathrm{Op}$ denotes the completely positive trace-preserving
maps~\cite{Rains2001}. Two things
in~\eqref{eq:rainsclass} matter below. The transposition is applied to the
input \emph{and} the
output space of the same party, $W$ and $W'$ together. This is the
label-matching convention of Sec.~\ref{sec:theory:def}, and it is Rains'
convention, not an addition of ours. Second, the class is defined by requiring
membership of $\mathrm{Op}$ both before and after conjugation, that is, by
requiring $\Gamma_{W}\Lambda\Gamma_{W}$ to be completely positive rather than
merely positive.

Rains' motivation was not ours, and the difference matters. In the ordering
$\mathcal{C}_{\epsilon} \subset \mathcal{C}_{1},\mathcal{C}_{1'} \subset
\mathcal{C}_{2} \subset \mathcal{C}_{\$} \subset \mathcal{C}_{\Gamma}$, with
all inclusions strict in general, the class $\mathcal{C}_{\Gamma}$ is the
smallest one containing the two-local operations for which membership can be
decided effectively~\cite{Rains2001}. It is thus a computational relaxation of
LOCC, introduced to bound distillation rates rather than to define a theory.

\paragraph*{Positive versus completely positive.}
Two conditions are current under similar names, and only one of them is used
here. A map is \emph{PPT-preserving} if it sends operators with positive
partial transpose to operators with positive partial transpose, which is
positivity of $\Gamma_{V}\Lambda\Gamma_{V}$. It is \emph{completely}
PPT-preserving if $\Gamma_{V}\Lambda\Gamma_{V}$ is completely positive. The
second is strictly stronger, and it is what Definition~\ref{def:T} imposes. The reason is visible in
Sec.~\ref{sec:enc:ops}, where the encoded operation is built
from Kraus operators of $\Gamma_{V}\Lambda\Gamma_{V}$, one sector at a time,
and those exist only under complete positivity.

\paragraph*{The multipartite generalisation.}
Ishizaka and Plenio extend the class to several parties by imposing, on the
Choi operator $\Omega(\Psi)$, the conditions
\begin{equation}
  \label{eq:ipconditions}
  \bigl( \Omega(\Psi)^{\Gamma_{V}} \bigr)^{\Gamma_{V_{i}} \otimes
  \Gamma_{V'_{i}}} \;\ge\; 0
  \qquad \text{for } i = A,B,C ,
\end{equation}
where $\Gamma_{V}$ is the transposition entering the Choi criterion for
complete positivity and $\Gamma_{V_{i}}, \Gamma_{V'_{i}}$ act on the input and
output space of party $i$~\cite{IshizakaPlenio2005}. They observe that the
bipartite case admits two equivalent choices of transposition, whereas in the
tripartite case there are three that are in general inequivalent.

The index $i$ in~\eqref{eq:ipconditions} runs over \emph{single parties}, and
for three parties that is the same as running over all bipartitions. The subsets of
$\{A,B,C\}$ come in complementary pairs, complementary subsets give
the same condition because the global transpose preserves positivity, and the
three singletons already exhaust the pairs. From four parties onward the two
prescriptions separate. There are then $2^{\,n-1}-1 = 7$ inequivalent subsets
but only four singletons, and positivity across a $2$:$2$ cut is not implied by
positivity across the $1$:$3$ cuts. The operator of
Eq.~\eqref{eq:imagwitness} witnesses the separation in closed form, being
$\{0,\tfrac12,1\}$ on every singleton cut and $\{-\tfrac12,\tfrac12,\tfrac32\}$
on every crossing $2$:$2$ cut. That value is the worst possible. Writing
$X = \ptr{E}{1}$, so that $0 \le X \le \Id$, and taking any unit vector $v$
with Schmidt coefficients $(s_{1}, s_{2})$ across the second wire,
$\bra{v} \ptr{E}{\{1,2\}} \ket{v} = \Tr[X \, \ptr{(\proj{v})}{2}]
\ge -s_{1}s_{2} \ge -\tfrac12$, because $\ptr{(\proj{v})}{2}$ has spectrum
$\{s_{1}^{2}, s_{2}^{2}, \pm s_{1}s_{2}\}$. Only
$0 \le \ptr{E}{1} \le \Id$ is used.

Definition~\ref{def:T} takes $V$ over all subsets. The reason is not
generality for its own sake but the sector decomposition of Sec.~\ref{sec:enc:ops}, in which the encoded
operation carries one sector for each
subset $V$ of the wires, its block there is $\PT{V}\Lambda\PT{V}$, and complete
positivity of the whole is complete positivity of every block. A reader
comparing our condition with the one cited will therefore find ours longer by
$2^{\,n-1}-1-n$ conditions for $n \ge 4$ parties, and this is deliberate.

\section{The encoding in basis-free form}
\label{app:poslemma}

The encoding of Definition~\ref{def:encoding} is written in the eigenbasis of
$Y$, which makes the block structure manifest and the positivity lemma
immediate. The same object has a form that refers to no basis. That form shows
the construction does not depend on the qubit realisation, and it is the one
that
carries over to the $G$-frame encoding of Sec.~\ref{sec:sectorial:G}.

\paragraph*{Complex structures.}
A frame register is a two-dimensional space equipped with a \emph{complex
structure}, that is, a real antisymmetric operator $\mathcal{J}$ with $\mathcal{J}^{2} = -\Id$. On
a qubit we may take
\begin{equation}
  \label{eq:Jdef}
  \mathcal{J} \;=\; iY \;=\;
  \begin{pmatrix} 0 & 1 \\ -1 & 0 \end{pmatrix} ,
\end{equation}
which is real, antisymmetric and squares to $-\Id$, and whose eigenvalues
$\pm i$ have $P_{\pm}$ as eigenprojectors. Transposition sends $\mathcal{J}$ to $-\mathcal{J}$ and exchanges $P_{+}$ with
$P_{-}$, conjugation alone fixing $\mathcal{J}$ because it is real. In the language of
Sec.~\ref{sec:sectorial:G} the frame register carries the regular
representation of $\Z_{2}$, with $\ket{y_{\pm}}$ as its two basis vectors and
conjugation acting by group multiplication. This is hypothesis (H1) in the case
at hand.

\paragraph*{Agreement of frames.}
On the $k$ frame registers of a single source put
\begin{equation}
  \label{eq:PiL}
  \Pi_{k} \;=\; P_{+}^{\otimes k} + P_{-}^{\otimes k} ,
  \qquad
  L_{k} \;=\; i \bigl( P_{+}^{\otimes k} - P_{-}^{\otimes k} \bigr) .
\end{equation}
Both are real, $\Pi_{k}$ is a projector, and the cross terms vanish by
orthogonality, so that
\begin{equation}
  \label{eq:LsqPi}
  L_{k}^{2} \;=\; -\,\Pi_{k} ;
\end{equation}
thus $L_{k}$ is a complex structure on the range of $\Pi_{k}$. That range has a
basis-free description. Writing $\mathcal{J}_{i}$ for the complex structure of the $i$th
frame, one has $\mathcal{J}_{i}\mathcal{J}_{j} = -\Id$ on $P_{+}^{\otimes k}$ and on
$P_{-}^{\otimes k}$, whereas $\mathcal{J}_{i}\mathcal{J}_{j} = +\Id$ on any sign pattern in which
the $i$th and $j$th entries differ. Hence $\Pi_{k}$ is the projector onto
\begin{equation}
  \label{eq:Piagree}
  \bigcap_{i<j} \ker \bigl( \mathcal{J}_{i}\mathcal{J}_{j} + \Id \bigr) ,
\end{equation}
the subspace on which all $k$ frames carry the \emph{same} complex structure.
For $k = 2$ this reads $\Pi_{2} = \tfrac12(\Id - \mathcal{J} \otimes \mathcal{J})$ and
$L_{2} = \tfrac12(\Id \otimes \mathcal{J} + \mathcal{J} \otimes \Id)$, which are perhaps the more
recognisable forms.

\paragraph*{The state encoding.}
Split a Hermitian $\rho$ as $\rho = \mathrm{Re}\,\rho + i\,\mathrm{Im}\,\rho$,
with $\mathrm{Re}\,\rho$ real symmetric and $\mathrm{Im}\,\rho$ real
antisymmetric, so that
$\rho^{\mathsf T} = \mathrm{Re}\,\rho - i\,\mathrm{Im}\,\rho$. Substituting
into~\eqref{eq:stateenc} and using~\eqref{eq:PiL},
\begin{equation}
  \label{eq:basisfreestate}
  \tilde\rho \;=\; \tfrac12 \bigl(
    \mathrm{Re}\,\rho \otimes \Pi_{k}
    \;+\; \mathrm{Im}\,\rho \otimes L_{k} \bigr) .
\end{equation}
Both terms are manifestly real, which reproves the reality of $\tilde\rho$
without appeal to the exchange of sectors. On the range of $\Pi_{k}$ the
operator $L_{k}$ plays the role that multiplication by $i$ plays in the
unsymmetrized theory, and this is the sense in which the frame carries the
complex structure that $\sTh$ has given up.

\paragraph*{One frame or many.}
Equation~\eqref{eq:basisfreestate} also makes the difference from the standard
real-space simulation explicit. If a single frame register were shared by every
system in the experiment, then~\eqref{eq:basisfreestate} with $k$ equal to the
total number of wires would be the embedding of complex quantum theory into
real quantum theory used by McKague, Mosca and Gisin~\cite{McKague2009}, and
the simulation would succeed trivially. What the network scenarios exploit is
that the causal structure forbids that frame. The construction of
Sec.~\ref{sec:enc} uses one register per source instead, which the causal
structure does permit, and the price is that $\Pi$ and $L$ are defined only
within a source, so that a party holding wires from two sources holds two
complex structures that need not agree. The positivity lemma is the statement
that $\effset$ is exactly large enough to tolerate that disagreement, and no
larger.

\section{Encoded operations at the level of Choi matrices}
\label{app:sectors}

Section~\ref{sec:enc:ops} works with maps. Passing to Choi matrices adds
nothing to the argument but turns the sector conditions into semidefinite
constraints on a single operator, which is how membership in $\opset$ is
actually decided, and it makes the computation of Sec.~\ref{sec:enc:frames}
about fresh frames short enough to display.

Throughout, the Choi matrix of a map $\Phi$ on $n$ wires is
$J(\Phi) = \sum_{ij} \ketbra{i}{j} \otimes \Phi(\ketbra{i}{j})$, with the
convention that the first factor carries the input labels and the second the
output labels, and $\Phi$ is completely positive if and only if
$J(\Phi) \ge 0$. The convention is the unnormalised one. The Choi matrix of an isometry is then
$\ketbra{v}{v}$ with $\lVert v \rVert^{2}$ equal to the input dimension rather
than to one. Where we call such an object a rank-one projector below, that is
the object meant, and the eigenvalues quoted follow the same convention.

\paragraph*{The sectors are a group-algebra decomposition.}
The block structure that appears throughout this appendix is forced rather than
chosen, and saying why makes the rest of it easier to read.

Let $\mathcal{A}$ be the linear span of $\{\PT{S}\}_{S \subseteq W}$ inside the
real-linear maps on $\Herm$. By Sec.~\ref{sec:theory:def} the $\PT{S}$ commute
and each squares to the identity, so $\mathcal{A}$ is the group algebra of
$\Z_{2}^{\,n}$. That group is abelian with all elements of order two, so it has
$2^{n}$ one-dimensional characters, indexed by subsets $v \subseteq W$ through
$\chi_{v}(S) = (-1)^{|S \cap v|}$, and $\mathcal{A}$ is commutative and
semisimple, isomorphic to $\mathbb{R}^{2^{n}}$. Its primitive idempotents are
\begin{equation}
  \label{eq:idempotents}
  P_{v} \;=\; \frac{1}{2^{n}} \sum_{S \subseteq W} (-1)^{|S \cap v|}\, \PT{S} ,
\end{equation}
which satisfy $P_{v} P_{w} = \delta_{vw} P_{v}$ and $\sum_{v} P_{v} = \id$.
Any map commuting with every $\PT{S}$ is therefore a real combination of the
$P_{v}$, and any object built to be covariant for the group decomposes along
them.

The encoding of Sec.~\ref{sec:enc} does not use the whole group. A source
partition of the wires picks out the subgroup of \emph{source-aligned}
transpositions, isomorphic to $\Z_{2}^{\,m}$ with $m$ the number of sources,
generated by the $\PT{V(\sigma)}$ with $V(\sigma)$ the wire set of a single
source. The frame registers carry the regular representation of that subgroup,
one factor per source, and Eq.~\eqref{eq:opencoding} is exactly the resolution
of the encoded map into its isotypic components under it. The sector index
$\vec s$ of that equation is a character of $\Z_{2}^{\,m}$, and the projector
onto the sector is $P_{\vec s}$ in the sense of~\eqref{eq:idempotents} with the
sum restricted to the subgroup. Nothing about the conditions
$\PT{V}\Lambda\PT{V}$ completely positive, one per sector, is a modelling
choice. They are what the isotypic decomposition delivers, which is the sense
in which Sec.~\ref{sec:main:scope} says the class of Definition~\ref{def:T} is
forced and nothing weaker appears.

Two consequences are worth separating. First, since $\effset$ is invariant
under the full group $\Z_{2}^{\,n}$, it is a fortiori invariant under the
source-aligned subgroup, whatever the source partition happens to be. This is
why the encoding lands inside the effect set for every causal structure at once
rather than for one at a time, and it is the technical content of
Theorem~\ref{thm:main}. Second, a theory whose effect set were invariant only
under some proper subgroup would support the encoding for the source partitions
aligned with that subgroup and no others. The effect set $\effset_{1}$ of
Sec.~\ref{sec:disc} is not of that kind, since the singletons generate
the whole group, but it shows that constraining a generating set is weaker than
constraining the group it generates.

\paragraph*{The encoded channel is a direct sum.}
Let $\Lambda \in \opset$ and let $\tilde\Lambda$ be as in
Definition~\ref{def:encop}. Writing an input basis vector of the enlarged
system as $\ket{i,\vec s}$ and using
$\tilde\Lambda(\ketbra{i,\vec s}{j,\vec t}) = \delta_{\vec s \vec t}\,
\Lambda_{V(\vec s)}(\ketbra{i}{j}) \otimes \proj{\vec s}$, one obtains
\begin{equation}
  \label{eq:choiencoded}
  J(\tilde\Lambda) \;=\; \sum_{\vec s} J\bigl( \Lambda_{V(\vec s)} \bigr)
  \otimes \proj{\vec s}_{R_{\rm in}} \otimes \proj{\vec s}_{R_{\rm out}} .
\end{equation}
The frame registers therefore appear only through mutually orthogonal
projectors, and $J(\tilde\Lambda)$ is block diagonal with one block for each
sign pattern. Positivity of the whole is positivity of every block, which
recovers Lemma~\ref{lem:encops}(i): $\tilde\Lambda$ is completely positive if
and only if every $\Lambda_{V}$ is.

Since $V$ and its complement give the same condition, deciding whether a given
channel on $n$ wires belongs to $\opset$ amounts to imposing $2^{\,n-1}$
semidefinite constraints on the single matrix $J(\Lambda)$, namely
\begin{equation}
  \label{eq:opsdp}
  \bigl( J(\Lambda) \bigr)^{\mathsf{T}_{V_{\rm in} \cup V_{\rm out}}}
  \;\ge\; 0
  \qquad \text{for all } V ,
\end{equation}
using the identity of the next paragraph with $W = V$. This is a semidefinite
program, and it is the practical sense in which membership of the completely
PPT-preserving class is decidable~\cite{Rains2001}.

\paragraph*{Mismatched frames.}
For arbitrary subsets $V$ of the input labels and $W$ of the output labels,
\begin{equation}
  \label{eq:choimismatch}
  J\bigl( \PT{W} \circ \Lambda \circ \PT{V} \bigr)
  \;=\; \bigl( \PT{V} \otimes \PT{W} \bigr) \bigl( J(\Lambda) \bigr) ,
\end{equation}
the two transpositions acting on the input and output factors of the Choi
matrix respectively. This follows by substituting
$\PT{V}(\ketbra{i}{j}) = \ketbra{j}{i}$ on the wires in $V$ and relabelling the
summation.

Section~\ref{sec:enc:frames} argued that a wire leaving a box must inherit the
frame of the wire entering it, on the ground that independent frames would
require complete positivity of $\PT{W}\Lambda\PT{V}$ for every pair rather than
for the diagonal ones. Equation~\eqref{eq:choimismatch} makes the consequence
explicit. Taking $V = \emptyset$ and $W$ the full set of output labels, the
requirement becomes
\begin{equation}
  \label{eq:inoutppt}
  \bigl( J(\Lambda) \bigr)^{\mathsf{T}_{\rm out}} \;\ge\; 0 ,
\end{equation}
that is, the Choi matrix must have positive partial transpose across the cut
separating inputs from outputs.

No unitary survives this. For $\Lambda = U \cdot U^{\dagger}$ on a
$d$-dimensional space, $J(\Lambda) = \ketbra{U}{U}$ with
$\ket{U} = \sum_{i} \ket{i} \otimes U\ket{i}$, whose Schmidt rank across the
input:output cut is $d$. A rank-one projector onto a vector with Schmidt
coefficients $\lambda_{1},\dots,\lambda_{d}$ has partial transpose with
eigenvalues $\lambda_{i}^{2}$ and $\pm\lambda_{i}\lambda_{j}$ for $i < j$, so
its least eigenvalue is $-\max_{i<j}\lambda_{i}\lambda_{j}$, which is strictly
negative as soon as $d \ge 2$. Independent input and output frames would
therefore leave the simulation with no reversible dynamics at all.

\paragraph*{An explicit exclusion.}
The same identity, used with $W = V$ as in~\eqref{eq:opsdp}, shows why
\textsc{cnot} is not an operation of $\Th$. Its Choi matrix is a rank-one
projector onto a vector of Schmidt rank two across the cut separating the two
parties, taking both their input and output wires together, so the
corresponding partial transpose is not positive. The numerical value is
recorded in Appendix~\ref{app:numerics}.

\section{The frame-extension channel}
\label{app:frameext}

Section~\ref{sec:enc:frames} introduces the channel by which a wire created
inside a circuit acquires a frame aligned with an existing one, and states its
properties. Here they are verified, and the failure of the coherent
alternative is made quantitative.

\paragraph*{Definition and elementary properties.}
The channel $\mathcal{F}$ takes one frame register to two and has Kraus
operators $F_{\pm} = \ket{y_{\pm}y_{\pm}}\bra{y_{\pm}}$. It is trace
preserving,
\begin{equation}
  \label{eq:extTP}
  F_{+}^{\dagger}F_{+} + F_{-}^{\dagger}F_{-} \;=\; P_{+} + P_{-} \;=\; \Id ,
\end{equation}
and it aligns, since
$\sum_{\pm} F_{\pm} P_{s} F_{\pm}^{\dagger} = P_{s} \otimes P_{s}$ for
$s = \pm$, the cross terms vanishing by orthogonality of $P_{+}$ and $P_{-}$.
Its Choi matrix has rank two, one term for each outcome.

\paragraph*{Reality.}
The Kraus operators are not real, since $\ket{y_{\pm}}$ are not real vectors.
That is not what is required. What the encoding needs is that the
\emph{channel} commute with complex conjugation, and it does, since conjugation sends
$F_{\pm}$ to $F_{\mp}$, returning the Kraus family to itself, so
$\mathcal{F}(\bar X) = \overline{\mathcal{F}(X)}$ for every $X$. This is the
same lock as in Secs.~\ref{sec:enc:constr} and~\ref{sec:enc:ops}: conjugation
exchanges the two sectors rather than fixing each.

\paragraph*{Admissibility in \texorpdfstring{$\Th$}{T}.}
The condition~\eqref{eq:defops} has to be checked on all three wires, the input
frame and the two output frames. It holds for a structural reason rather than
by computation. The channel measures and reprepares, so it is entanglement breaking and its
Choi matrix is separable across the cut between input and output. Its output
is moreover the product $P_{s} \otimes P_{s}$ for each outcome, so the Choi
matrix is a sum of products of positive operators across all three
factors,
\begin{equation}
  \label{eq:extchoi}
  J(\mathcal{F}) \;=\; \sum_{s = \pm} P_{s}^{\mathsf T} \otimes P_{s} \otimes P_{s} .
\end{equation}
A fully separable operator has positive partial transpose across every
bipartition, so $\mathcal{F} \in \opset$.

\paragraph*{The coherent copy is not admissible.}
The obvious alternative is the isometry sending $\ket{y_{\pm}}$ to
$\ket{y_{\pm}y_{\pm}}$, which performs the same alignment while preserving
superpositions. Its Choi matrix is the rank-one projector onto
$\sum_{s} \ket{y_{s}} \otimes \ket{y_{s}y_{s}}$, a vector of Schmidt rank two
across the cut separating the input from the outputs, so by the spectral
formula of Appendix~\ref{app:sectors} its partial transpose there has least
eigenvalue $-1$. The coherent copy is excluded from $\Th$, and the restriction
on operations is doing work at this point rather than being carried along
inertly.

Nothing is lost by the incoherent choice. The states that pass through
$\mathcal{F}$ are the frame registers of encoded sources, which by
Eq.~\eqref{eq:stateenc} are diagonal in the sign basis, and no coherence
between sectors is ever called upon by the construction. The one thing
$\mathcal{F}$ must preserve is the perfect correlation of the signs, and
Eq.~\eqref{eq:extchoi} preserves it exactly.

\paragraph*{Numerical values.}
Trace preservation and the two alignment identities hold to $10^{-16}$. The
Choi matrix of $\mathcal{F}$ is positive and has positive partial transpose
across all four inequivalent cuts of its three wires, with worst value
$-1.1 \times
10^{-16}$. The channel commutes with conjugation exactly, and the coherent copy
has least eigenvalue $-1.0000$ across the worst cut. Details in
Appendix~\ref{app:numerics}.

\section{A classical--\texorpdfstring{$\Th$}{T} gap in the bilocality scenario}
\label{app:bilocal}

\begin{table*}[!t]
\caption{\label{tab:numerics}Numerical checks. ``Worst'' is the largest
violation of the stated condition over the trials, or the value of the
quantity being computed. Negative controls, which must fail, are marked.}
\centering
\begin{tabular}{cp{0.40\textwidth}cp{0.34\textwidth}}
\hline\hline
 & Check & Trials & Worst value \\ \hline
 1 & Lemma~\ref{lem:psym}: spectra of $\ptr{P_{\rm sym}}{A}$, $\ptr{P_{\rm asym}}{A}$ & exact
   & $\{\tfrac12,\tfrac12,\tfrac12,\tfrac32\}$
     and $\{-\tfrac12,\tfrac12,\tfrac12,\tfrac12\}$ \\
 2 & Lemma~\ref{lem:positivity}: $\xieff(E)$ admissible iff
     $E \in \effset$ & 400
   & 0 violations; spectra agree to $2.0\times10^{-15}$ \\
 3 & encoded sources real, symmetric, PSD, unit trace & 300
   & $2.8\times10^{-17}$, $0$, $-1.3\times10^{-16}$, $1.1\times10^{-15}$ \\
 4 & \cpptp{} excludes \textsc{cnot}, admits local unitaries & exact
   & $-2.0000$ (control), $-3.5\times10^{-16}$ \\
 5 & \textsc{cnot} pulls $\proj{+0}$ back to the Bell projector & exact
   & $-0.5000$ (control) \\
 6 & Prop.~\ref{prop:validgpt}(f): $\Lambda^{*}(\effset) \subseteq \effset$ & 300
   & $0$ \\
 7 & $\PT{V}\Lambda\PT{V}$ trace preserving, every $V$ & $200 \times 4$
   & $1.0\times10^{-15}$ \\
 8 & Prop.~\ref{prop:validgpt}(d): steering, effect side & 400, 120
   & $0$, $0$ \\
 9 & Prop.~\ref{prop:validgpt}(d): steering, state side & 400
   & $0$ \\
10 & Prop.~\ref{prop:validgpt}(e): parallel composition & 60
   & $-2.2\times10^{-15}$; control $-4.0000$ \\
11 & Prop.~\ref{prop:validgpt}(e): sequential composition & 60
   & $-3.8\times10^{-16}$ \\
12 & frame extension: TP, alignment, PPT, conjugation & 200
   & $\le 1.1\times10^{-16}$ throughout \\
13 & the coherent copy is inadmissible & exact
   & $-1.0000$ (control) \\
14 & Theorem~\ref{thm:main} in bilocality: $p_{\Th} = p_{\sTh}$ & 25, 6
   & $4.4\times10^{-16}$, $3.5\times10^{-16}$ \\
15 & Sec.~\ref{sec:theory:tomo}: \textsc{chsh} with product effects & exact
   & $2.8284271247$ \\
16 & App.~\ref{app:bilocal}: classical--$\Th$ gap in bilocality & exact
   & $\Pr[a{=}j] = 1$, \textsc{chsh} $= 2.8284271247$ \\
17 & App.~\ref{app:bilocal}: the BGP inequality on the same
     measurement & exact
   & $2^{1/4}$; controls $\sqrt2$ and ${<}1$; crossing at
     $\varepsilon = 0.535$ \\
18 & Sec.~\ref{sec:sectorial:not}: the Horodecki candidate & exact
   & ranks $7$ and $6$; realignment $1.0005$--$1.0030$ \\
19 & Sec.~\ref{sec:sectorial:tight}: $\kappa$ at $P_{\rm sym}$, and
     Eq.~\eqref{eq:scseries} & $3 \times 1500$
   & $\kappa = (1+d)/2$ exactly, not exceeded; residual $2\times10^{-16}$ \\
20 & App.~\ref{app:cpptp}: on four wires the singleton cuts do not imply the
     pair cuts & 4000 & singleton slack $\ge 0$, pair violation $-0.0869$;
     the exact optimum is $-\tfrac12$, by Eq.~\eqref{eq:imagwitness} \\
21 & Prop.~\ref{prop:noswap}: the swapped state of $\Th$ is always PPT
   & $2 \times 400$ & $0$; control, Bell projector, $-0.5000$ \\
22 & App.~\ref{app:e1}: $E_{t}$ lies in $\effset_{1}$ iff
     $|t| \le \tfrac12$, in $\effmax$ iff $|t| \le \tfrac14$ & exact
   & thresholds attained \\
23 & Eq.~\eqref{eq:opsvalid} against the definition it encodes
   & $5 \times 300$ & four channels admitted, {\sc cnot} pull-back at
     $-0.5000$ \\
24 & Prop.~\ref{prop:noswap} with four sources & exact
   & $1$:$3$ cuts at $0$, $2$:$2$ at $-0.0714$; control $-0.5$ \\
25 & Eq.~\eqref{eq:imagwitness}: an escape with invariant part $\Id/2$
   & exact & $H^{\mathsf T} + H = 0$, $r = 2$; control $r = 1$ \\
26 & Eq.~\eqref{eq:opschain}: the certificate, thirteen closed-form
     eigenvalue claims & exact
   & pairing $-0.0573$, worst margin $+1.6\times10^{-3}$ \\
\hline\hline
\end{tabular}
\end{table*}

This appendix supplies the detail behind the first of the four checks in
Sec.~\ref{sec:main:necessary}. Nothing in it is new as a statement about networks. That separable
measurements at the central node already suffice for network nonlocality is
known~\cite{Andreoli2017,Polino2026}. What has to be
shown here is only that a witness of that kind can be built from effects of
$\Th$, which the restriction~\eqref{eq:defeffects} makes a nontrivial demand.

\paragraph*{The protocol.}
Both sources emit $\proj{\phi^{+}}$. The central party holds one wire from
each, labelled $B_{1}$ and $B_{2}$, and has no setting. It measures $B_{1}$ in
the eigenbasis of $Z$ with outcome $j$, and then measures $B_{2}$ with a
setting determined by $j$, obtaining $m$. Its outcome is $b = (j,m)$ and its
POVM is
\begin{equation}
  \label{eq:bilocalPOVM}
\begin{split}
  E_{(j,m)} \;=\; \proj{j} \otimes \tfrac12 \bigl(
  &\Id + (-1)^{m} \cos\theta_{j} \, Z \\
  &+ (-1)^{m}\sin\theta_{j} \, X \bigr) ,
\end{split}
\end{equation}
with $\theta_{0} = \pi/4$ and $\theta_{1} = -\pi/4$. Alice measures $Z$ on her
wire, and Charlie measures $Z$ or $X$ according to his setting $z$.

Every element of~\eqref{eq:bilocalPOVM} is a product of two quantum effects,
hence an effect of $\Th$ by Sec.~\ref{sec:theory:tomo}. This is the only point
at which the restriction is in question, and it is where an entangled joint measurement would have failed. A Bell-state measurement
has elements with
negative partial transpose and is unavailable in $\Th$.

\paragraph*{The classical bound.}
Alice's outcome $a$ equals $j$ with probability one, since both are outcomes of
$Z$ measurements on the two halves of $\proj{\phi^{+}}$. In a classical model
of the form~\eqref{eq:bilocalclassical} an event of probability one must occur
for almost every pair $(\lambda_{1},\lambda_{2})$, so $a$ and $j$ are both
deterministic and equal there. As $a$ depends only on $\lambda_{1}$, so
does $j$. Write $j = f(\lambda_{1})$.

Conditioning on the value of $j$ therefore conditions only on $\lambda_{1}$,
and by independence of the sources the conditional distribution of
$\lambda_{2}$ is unchanged, $q_{2}(\lambda_{2}|j) = q_{2}(\lambda_{2})$. What
remains is
\begin{equation}
  \label{eq:bilocalreduced}
  P(m,c \,|\, j,z)
  \;=\; \int q_{2}(\lambda_{2}) \, R(m|j,\lambda_{2}) \, P(c|z,\lambda_{2}) ,
\end{equation}
a local hidden-variable model in which $j$ plays the part of a setting for the
central party and $\lambda_{2}$ is the shared variable. Its \textsc{chsh} value is
therefore at most $2$. Note where the causal structure enters. It is the independence of the two
sources, and nothing else, that makes $q_{2}$ in
\eqref{eq:bilocalreduced} free of $j$.

\paragraph*{The quantum value.}
With the settings above, the correlator between $m$ and $c$ conditioned on $j$
is $\cos(\theta_{j} - \varphi_{z})$ with $\varphi_{0} = 0$ and
$\varphi_{1} = \pi/2$, giving
\begin{equation}
  \label{eq:bilocalchsh}
  \sum_{j,z} (-1)^{jz} \, \langle m \, c \rangle_{j,z}
  \;=\; 2\sqrt2 .
\end{equation}
The distribution is therefore not classically realisable in the bilocality
causal structure, and $\Th$ exhibits a gap against classical probability theory
there, as Proposition~2 of Appendix~J of Ref.~\cite{Ying2025foil} requires.

\paragraph*{The same measurement against the standard inequality.}
The argument above is self-contained. It is instructive nonetheless to see the
protocol against the bilocality inequality of Branciard, Rosset, Gisin and
Pironio~\cite{Branciard2012}, which is the standard test for the scenario and
which refines the analysis begun in Ref.~\cite{Branciard2010}. Label the central outcome by the pair
$(b^{0},b^{1}) = (j \oplus m, m)$ and define $I$ and $J$ as in
Ref.~\cite{Branciard2012}, so that bilocal-classical correlations satisfy
$\sqrt{|I|} + \sqrt{|J|} \le 1$. Give Alice a second setting, taken to be the
trivial one, and let Charlie measure $Z$ and $X$. The resulting distribution
gives
\begin{equation}
  \label{eq:bgpvalue}
\begin{split}
  |I| \;=\; |J| &\;=\; \tfrac{1}{2\sqrt2} , \\
  \sqrt{|I|} + \sqrt{|J|} &\;=\; 2^{1/4} \;\approx\; 1.1892 ,
\end{split}
\end{equation}
against the entangled-measurement value of $\sqrt2$. So $\Th$ violates the
inequality outright, with a measurement all of whose elements are rank-one
product projectors. The violation is not an artefact of the trivial setting.
Replacing that setting by the two-outcome measurement with first element
$(1-\varepsilon)\Id + \tfrac{\varepsilon}{2}(\Id + Z)$, which interpolates
between the trivial setting at $\varepsilon = 0$ and a projective one at
$\varepsilon = 1$, leaves the value above $1$ up to
$\varepsilon \approx 0.535$, with $1.1095$ at $\varepsilon = 0.25$ and
$1.0151$ at $\varepsilon = 0.5$.

That separable measurements can violate this family of inequalities is not
new. Branciard \emph{et al.} observe it themselves for the variant in which the
central party has a setting~\cite{Branciard2012}. What matters
here is only that the violating measurement can be assembled from effects of
$\Th$, which the restriction~\eqref{eq:defeffects} makes a real constraint.

One feature of this is easy to misread. If Alice's two
settings are both projective, the same measurement gives
$\sqrt{|I|} + \sqrt{|J|}$ equal to $0.5946$ or $0.8409$ depending on which two
she chooses, and the inequality is not violated at all, even though the
distribution remains nonbilocal by the argument above. The sensitivity of the
inequality here therefore depends on Alice's settings rather than on anything
the central party did, so the perfect-correlation argument and the inequality belong side by side.

\paragraph*{Numerical check.}
All of the above was verified directly. The four elements
of~\eqref{eq:bilocalPOVM} sum to the identity to machine precision. Each is a
product operator and satisfies $0 \le \ptr{E}{S} \le \Id$ for all four subsets
$S$, with worst violation $0.0$. The distribution is normalised for each $z$
and gives $\Pr[a=j] = 1$ to ten decimal places, and the \textsc{chsh} combination
evaluates to $2.8284271247$. The values of $\sqrt{|I|}+\sqrt{|J|}$ quoted
above are computed from the joint distribution by the definitions
of Ref.~\cite{Branciard2012} rather than from any operator expression, and the
convention is calibrated on two controls, an entanglement-swapping protocol
with a Bell-state measurement, which returns $\sqrt2$, and the requirement that
bilocal-classical distributions stay below $1$. Seed and script details are in
Appendix~\ref{app:numerics}.

\section{The effect set \texorpdfstring{$\effset_{1}$}{E1}}
\label{app:e1}

This appendix supplies what Sec.~\ref{sec:disc} defers about the effect set
$\effset_{1}$. Both witnesses below live on four qubit wires, which by the
argument there is the smallest number at which $\effset_{1}$ and $\effmax$
differ.

Along the ray $E = \Id/2 + sH$ with $H$ Hermitian, the
condition $0 \le \ptr{E}{S} \le \Id$ reads
$s \lVert \ptr{H}{S} \rVert_{\infty} \le 1/2$, so the boundary of
$\effset_{1}$ on that ray lies inside $\effmax$ if and only if
\begin{equation}
  \label{eq:ratio}
  r(H) \;:=\;
  \frac{\max_{|S| = 2} \lVert \ptr{H}{S} \rVert_{\infty}}
       {\max_{|S| \le 1} \lVert \ptr{H}{S} \rVert_{\infty}}
  \;\le\; 1 .
\end{equation}
The ratio is scale invariant, so no renormalisation is needed to compare
directions, and it is bounded: a $2$:$2$ transposition is a single-wire one
followed by a one-qubit transposition, whose completely bounded norm is two, so
$r \le 2$ always. Both of the witnesses below attain that ceiling.

A closed-form element of $\effset_{1} \setminus \effmax$ is
\begin{equation}
  \label{eq:Et}
  E_{t} \;=\; \tfrac12\bigl( \Id + t\,\SWAP_{12} \otimes \SWAP_{34} \bigr) ,
\end{equation}
whose spectra are exact. The swap operator has spectrum $\{\pm1\}$ and
$\ptr{\SWAP}{1} = d\proj{\phi^{+}}$ has spectrum $\{0,d\}$, so $E_{t}$ itself
has eigenvalues $(1\pm t)/2$, a single-wire partial transpose has
$(1\pm 2t)/2$ and $1/2$, and a partial transpose across a crossing $2$:$2$ cut
has
$(1+4t)/2$ and $1/2$. Hence $E_{t} \in \effset_{1}$ exactly when $|t| \le 1/2$ and
$E_{t} \in \effmax$ exactly when $|t| \le 1/4$, the upper bound
$\ptr{E}{S} \le \Id$ binding for positive $t$ and the lower bound for negative
$t$, and any $t$ in between separates
the two. Here $r = 2$.

The first witness is a real matrix and is therefore invariant under the
symmetry. The second is not, which is what Sec.~\ref{sec:disc} appeals to. Take
\begin{equation}
  \label{eq:imagwitness}
  H \;=\; \tfrac12 \bigl\{\, \SWAP_{12} \otimes \SWAP_{34} ,\;
  \Id Z X Y \,\bigr\} ,
  \qquad
  E \;=\; \tfrac12 ( \Id + H ) ,
\end{equation}
the anticommutator of the direction of $E_{t}$ with a Pauli string carrying an
odd number of $Y$ factors. Then $H^{\mathsf T} = -H$ exactly, so the invariant
part of $E$ is $\Id/2$ and does not move off centre at all. The spectra are
integers: $H$ and each of its single-wire partial transposes have spectrum
$\{-1,0,1\}$, and each of the four crossing $2$:$2$ transpositions has
$\{-2,0,2\}$. So $E$ and every $\ptr{E}{\{i\}}$ have spectrum
$\{0,\tfrac12,1\}$, saturating the conditions defining $\effset_{1}$ with zero
slack, while a crossing cut gives $\{-\tfrac12,\tfrac12,\tfrac32\}$ and
violates $\effmax$ by $\tfrac12$ at both ends. Again $r = 2$.

\section{Numerical verification}
\label{app:numerics}

Every claim in this paper is proved analytically. The checks recorded here are
corroboration, not evidence, and they cover the numerical claims of the body
rather than every closed-form spectrum stated in it. They are of a particular
kind. Each is a closed-form identity or an eigenvalue evaluation, so there is no solver, no
optimisation and no search over a feasible set. A check that passes verifies
algebra rather than failing to find a counterexample. The script accompanies this paper as an ancillary file and is deposited
separately~\cite{verifyT}. It uses only NumPy and runs in seconds. The certificate for Eq.~\eqref{eq:opschain} is deposited
alongside it, as fourteen operators in a NumPy archive together with a
solver-free checker. The checker distributed with it and the one used here were
written independently of each other and agree to the last digit.

Table~\ref{tab:numerics} collects the results. All figures come from a single
run with seed $20260823$ under NumPy $2.4.4$. Values of order $10^{-16}$ are at the level of double-precision
rounding.

Three caveats belong with the table, and none of them is cosmetic.

The random elements of $\opset$ in checks~6, 7, 10, 11 and~14 are sampled as
convex mixtures of local unitary channels, which form a proper subset of the
completely PPT-preserving class. Those checks therefore test sufficiency on
that subset only. Necessity is tested exactly rather than randomly, on
\textsc{cnot} in checks~4 and~5 and on $\mathrm{CNOT} \otimes \Id$ in check~10, and
on the coherent copy in check~13. These are the negative controls that keep
the tests from being vacuous.

The Kraus operators of the frame-extension channel are not real, and check~12
does not claim otherwise. What the encoding requires is that the channel
commute with complex conjugation, which is the quantity reported.

Finally, none of this addresses necessity of the restriction on operations, in
the sense of Sec.~\ref{sec:main:scope}. Checks~4, 5 and~13 show that specific
channels fall outside $\opset$ and that the encoding fails for them. They do
not show that no other encoding succeeds, and no numerical experiment could.

\bibliographystyle{quantum}
\bibliography{refs}

\end{document}